\documentclass[12pt]{article}
\RequirePackage[OT1]{fontenc}

\usepackage{algorithm,algpseudocode}
\usepackage{amsmath,amsfonts,amsthm,amssymb,amsbsy}
\usepackage{array,epsfig,fancyhdr,rotating}
\usepackage{booktabs}
\usepackage{color, dsfont}
\usepackage{enumerate}
\usepackage{float}
\usepackage{graphicx,psfrag,epsf}
\usepackage[]{hyperref}  
\usepackage{ltxtable}
\usepackage{multirow}
\usepackage{pdflscape}
\usepackage[sectionbib]{natbib}
\usepackage{sectsty, secdot}
\usepackage{subcaption}
\usepackage{threeparttable}
\usepackage{url} 

\renewcommand{\theequation}{\thesection.\arabic{equation}}
\def\thefigure{\arabic{figure}}
\def\thetable{\arabic{table}}

\newtheorem{corollary}{Corollary}

\newtheorem{example}{Example}
\newtheorem{lemma}{Lemma}

\newtheorem{remark}{Remark}
\newtheorem{theorem}{Theorem}

\begin{document}


\title{\Large\bf Expected Weighted D-optimal Designs for Experiments with Mixed Factors}
\author{Siting Lin$^1$, Yifei Huang$^2$, and Jie Yang$^1$\\ 
	$^1$University of Illinois at Chicago and $^2$Astellas Pharma Global Development, Inc.
}

\maketitle


\begin{abstract}
Optimal designs can help experimenters obtain more accurate parameter estimates with reduced experimental time and cost. In this paper, we characterize the Expected Weighted (EW) D-optimal designs as robust designs against unknown parameter values for experiments under a general parametric model with discrete and continuous factors. When a pilot study is available, we recommend sample-based EW D-optimal designs for subsequent experiments. Otherwise, we recommend EW D-optimal designs under a prior distribution for model parameters. We propose an EW ForLion algorithm for finding EW D-optimal designs with mixed factors, and justify that the designs found by our algorithm are EW D-optimal. To facilitate potential users in practice, we also develop a rounding algorithm that converts an approximate design with mixed factors to exact designs with prespecified grid points and the total number of experimental units. By applying our algorithms for real experiments under multinomial logistic models or generalized linear models, we show that our designs are highly efficient with respect to locally D-optimal designs and more robust against parameter value misspecifications.
\end{abstract}


\section{Introduction}
\label{sec:intro}

In this paper, we consider robust designs against unknown parameter values for experiments with discrete and continuous factors. A motivating example is the paper feeder experiment described in \cite{joseph2004}. The goal was to ensure precise feeding of one sheet of paper each time to a copier. Two common failure modes are frequently observed, namely  {\tt misfeed} when the feeder fails to feed any sheet of paper, and  {\tt multifeed} when the feeder picks up more than one sheets of paper at a time. The experiment involves eight discrete control factors (see Table~2 in \cite{joseph2004} or Section~\ref{sec:Model_Selection_Paper_Feeder_Experiment} in the Supplementary Material) and one continuous factor, the stack force. The responses fall into three mutually exclusive categories, namely  {\tt misfeed} (typically with low stack force), {\tt normal}, and {\tt multifeed} (typically with high stack force). The results of the original experiment were listed in  Table~3 of \cite{joseph2004}. The original design was conducted via a control array modified from $OA(18, 2^1\times 3^7)$ (see Table~5 in \cite{joseph2004}) for the eight discrete factors and an essentially uniform allocation on a predetermined set of discrete levels of the continuous factor.
If a subsequent experiment will be conducted to study the factor effects, how can we do better?

In the original analysis by \cite{joseph2004}, two generalized linear models \citep{pmcc1989, dobson2018} with probit link were used to model {\tt misfeed} and {\tt multifeed} separately. However, since {\tt misfeed} and {\tt multifeed} do not occur simultaneously and are both possible outcomes of the experiment, it would be more appropriate to adopt one multinomial model rather than two separate binomial models, so that the factor effects can be estimated more precisely due to the combined data and also be comparable for different types of failures. In this study, we adopt the multinomial logistic models \citep{pmcc1995, atkinson1999, bu2020} for analyzing the paper feeder experiment, which include baseline-category (also known as multiclass logistic), cumulative, adjacent-categories, and continuation-ratio logit models. 
In addition, three different parameter structures or assumptions are considered for these four logit models: {\tt proportional odds (po)} assuming that parameters are identical across response categories; {\tt non-proportional odds (npo)} allowing parameters to vary by category; and {\tt partial proportional odds (ppo)} containing both constant and varying parameters across categories (see \cite{bu2020} and references therein).

For experiments with discrete factors only, many numerical algorithms have been proposed for finding optimal designs, including Fedorov-Wynn \citep{fedorov1972, fedorov2025model}, multiplicative \citep{titterington1976,titterington1978,silvey1978}, cocktail \citep{yu2011d}, lift-one \citep{ym2015,ytm2016, bu2020}, as well as classical optimization techniques (see \cite{huang2024forlion} for a review). According to \cite{ymm2016} and \cite{huang2024forlion}, the lift-one algorithm outperforms many commonly used optimization algorithms and obtains D-optimal designs with fewer distinct design points, including Fedorov-Wynn, multiplicative, cocktail, as well as Nelder-Mead, quasi-Newton, and simulated annealing algorithms. 

For experiments with continuous factors only, \cite{yangmin2013} and \cite{harman2020randomized} discretized the continuous factors into grid points and treated the experiments as with discrete factors, and \cite{ai2023locally} incorporated the Fedorov-Wynn algorithm with the lift-one algorithm for continuation-ratio link models. 

For experiments with mixed factors, particle swarm optimization (PSO) algorithms have been applied for D-optimal designs under generalized linear models with binary responses \citep{lukemire2018} and cumulative logit models \citep{lukemire_optimal_2022}, which, however, cannot guarantee the optimal solution to be ever found \citep{kennedy1995particle, poli2007particle}. \cite{harman2021optimal} proposed a grid-exploration method, called the galaxy exploration (GEX), for finding D-optimal designs under linear, generalized linear, and nonlinear models after discretizing the continuous factors. \cite{huang2024forlion} proposed a ForLion algorithm by adding a merging step to \cite{ai2023locally}'s algorithm, which can significantly reduce the number of distinct design points and thus the experimental time and cost. Unlike PSO-types of algorithms, the D-efficiency of the ForLion algorithm is guaranteed by the general equivalence theorem \citep{kiefer1974, pukelsheim1993, atkinson2007,stufken2012, fedorov2014} when the convergence condition is met.

The ForLion algorithm has been successfully applied to both generalized linear models and multinomial logistic models with mixed factors \citep{huang2024forlion}. However, it relies on assumed parameter values and produces designs known as {\it locally optimal} \citep{chernoff1953}, which may not be robust when the parameter values are misspecified.
To address the limitation due to local optimality, Bayesian D-optimal designs have been proposed \citep{chaloner1995}, which maximize $E(\log|{\mathbf F}|)$ with respect to a prior distribution on unknown parameters, where ${\mathbf F}$ is the Fisher information matrix associated with the design. However, a major drawback of Bayesian D-optimal designs is its computational intensity, even with D-criterion that maximizes the determinant of the Fisher information matrix. As a promising alternative, Expected Weighted (EW) D-optimal designs have been proposed  for generalized linear models \citep{ymm2016}, cumulative link models \citep{ytm2016}, and multinomial logistic models \citep{bu2020} with discrete factors, which maximize $\log |E({\mathbf F})|$ or equivalently $|E({\mathbf F})|$. Compared with Bayesian D-optimal designs, EW D-optimal designs are computationally much faster and often highly efficient in terms of the Bayesian D-criterion \citep{ymm2016, ytm2016, bu2020}. 
Along this line, \cite{tong2014} provided analytical solutions for certain cases of D-optimal designs for generalized linear models, and \cite{bu2020} incorporated the lift-one and exchange algorithms with the EW D-criterion.

In this paper, we adopt EW D-criterion and adapt the ForLion algorithm for finding robust designs with mixed factors. The original EW D-optimal designs were proposed for experiments with discrete factors or discretized continuous factors. 
To solve our problems, we formulate in Section~\ref{sec:EW_D_optimal} the EW D-optimal criterion and characterize EW D-optimal designs for general parametric statistical models with mixed factors. In particular, we theoretically justify the existence of EW D-optimal designs that contains no more than $p(p+1)/2$ design points with continuous or mixed factors (see Theorem~\ref{thm:thm 2.2 page 64}), where $p$ is the number of parameters (see also Theorems~\ref{thm:MLM design points} and \ref{thm:GLM design points}). In Section~\ref{sec:Algorithms}, we develop algorithms for finding two types of EW D-optimal designs with mixed factors, namely  sample-based EW D-optimal designs for experiments with a pilot study, and integral-based EW D-optimal designs for experiments with a prior distribution for unknown parameters. We theoretically justify that the designs obtained by our EW ForLion algorithm are EW D-optimal (see Corollary~\ref{cor:EW_D_optiality}) when the convergence condition is met. We further formulate EW D-optimal designs for multinomial logistic models and generalized linear models, and propose an algorithm for converting approximate designs with mixed factors to exact designs with grid points, which is much faster than finding optimal designs restricted on the same set of grid points. In Section~\ref{sec:applications},  we construct the proposed robust designs for two real experiments, namely the paper feeder experiment \citep{joseph2004}, and an experiment on minimizing surface defects \citep{wu2008}. 
In Section~\ref{section:Conclusion and Discussion}, we conclude our findings. 
The algorithms are implemented in R (version 4.4.2, \cite{R-base}) and the simulations are run on a Windows 11 laptop with 32GB of RAM and a 13th Gen Intel Core i7-13700HX processor. An implementation of these algorithms is provided in the R package \texttt{ForLion} (version 0.4.0, \cite{huang2025forlionpackage}).

\section{EW D-optimal Designs with Mixed Factors}
\label{sec:EW_D_optimal}

Following \cite{huang2024forlion}, we first consider experiments with a design region ${\cal X} = \prod_{j=1}^d I_j \subset \mathbb{R}^d$, where $I_j$ is either a closed finite interval for the first $k$ factors, or a finite set of distinct numerical levels for $k+1\leq j\leq d$. In both cases, we denote $a_j = \min I_j > -\infty$ and $b_j = \max I_j < \infty$. If $k=0$, all factors are discrete or qualitative; if $k=d$, all factors are continuous. Note that ${\cal X}$ here is compact and constructed as a product of sets, which is common for typical applications. 

For some applications, e.g., when the levels of some discrete factors are not numerical (see Example~\ref{ex:minimizing_surface_example} below), or the level combinations of discrete factors under consideration are restricted to a true subset of all possible level combinations (see Example~\ref{ex:paper_feeder_example} below), we need to consider more general design regions in the form of ${\cal X} = \prod_{j=1}^k I_j\ \times \ {\cal D}$, where ${\cal D} \subseteq \prod_{j=k+1}^d I_j$~. Apparently, $\prod_{j=1}^d I_j$ is a special case of $\prod_{j=1}^k I_j\ \times\ {\cal D}$ with  ${\cal D} = \prod_{j=k+1}^d I_j$~.

\begin{example}\label{ex:minimizing_surface_example}
A polysilicon deposition process for circuit manufacturing described by \cite{phadke1989} involved one qualitative control factor, {\tt cleaning method} ({\it None} for no cleaning, $CM_2$ for cleaning inside the reactor, and $CM_3$ for cleaning outside the reactor), and five other factors that are continuous in nature (see Table~4.1 in \cite{phadke1989}). Since the levels of cleaning method are not numerical, it is more appropriate to convert it to two (binary) indicator variables $({\mathbf 1}_{CM_2}, {\mathbf 1}_{CM_3})$ with $(0,0), (1,0), (0,1)$ standing for {\it None}, $CM_2$, and $CM_3$, respectively. In this case, the corresponding design region is ${\cal X} = \prod_{j=1}^5 I_j\ \times\ \{(0,0), (1,0), (0,1)\}$, where $I_1, \ldots, I_5$ may be defined as in Table~S1 of the Supplementary Material of \cite{huang2024forlion}. Note that the cleaning method was simplified into a binary factor ($CM_2$ or $CM_3$) in \cite{lukemire_optimal_2022} and \cite{huang2024forlion}.
\hfill{$\Box$}
\end{example}

\begin{example}\label{ex:paper_feeder_example}
In the paper feeder experiment described by \cite{joseph2004}, besides one continuous factor, there are eight discrete or qualitative control factors, including two binary ones and six three-level ones (see Section~\ref{sec:Model_Selection_Paper_Feeder_Experiment} in the Supplementary Material). Instead of considering all $2^2 \times 3^6=2,916$ level combinations, the experimenters adopted 18 out of them modified from an orthogonal array $OA(18, 2^1\times 3^7)$ (see Table~5 in \cite{joseph2004}). In this case, for illustration purposes, we may consider a restricted design region ${\cal X} = [0,160] \times {\cal D}$, where $[0,160]$ is the range of the continuous factor, and ${\cal D} = \{(x_{i2}, \ldots, x_{i9})\mid i=1, \ldots, 18\}$ consists of the original 18 level combinations of eight discrete factors.
\hfill{$\Box$}
\end{example}

Given an experimental setting ${\mathbf x}_i = (x_{i1}, \ldots, x_{id})^T \in {\cal X}$, suppose there are $n_i \geq 0$ experimental units assigned to this experimental setting. Their responses (could be vectors) are assumed to be iid from a parametric model $M({\mathbf x}_i, \boldsymbol{\theta})$, where the parameter vector $\boldsymbol{\theta} \in \boldsymbol{\Theta} \subseteq \mathbb{R}^p$ with the parameter space $\boldsymbol{\Theta}$.  Suppose the responses are independent across different experimental settings. Under regularity conditions (see, e.g., Sections~2.5 and 2.6 in \cite{lehmann1998theory}), the corresponding Fisher information matrix can be written as $\sum_{i=1}^m n_i {\mathbf F}({\mathbf x}_i, \boldsymbol{\theta})$, with respect to the corresponding design $\{({\mathbf x}_i, n_i), i=1, \ldots, m\}$, known as an {\it exact design}, where ${\mathbf F}({\mathbf x}_i, \boldsymbol{\theta})$ is a $p\times p$ matrix, known as the Fisher information at ${\mathbf x}_i$~. In design theory, $\boldsymbol{\xi} = \{({\mathbf x}_i, w_i), i=1, \ldots, m\}$ with $w_i = n_i/n \geq 0$ and $n = \sum_{i=1}^m n_i$~, known as an {\it approximate design}, is often considered first \citep{kiefer1974, pukelsheim1993, atkinson2007}. The collection of all feasible designs is denoted by $\boldsymbol{\Xi} = \{ \{({\mathbf x}_i, w_i), i=1, \ldots, m\} \mid m\geq 1, {\mathbf x}_i \in {\cal X}, w_i\geq 0, \sum_{i=1}^m w_i=1 \}$.

\subsection{EW D-optimal designs}
\label{sec:EW_bootstrapped_EW}

Following \cite{bu2020} and \cite{huang2024forlion}, we adopt the D-criterion for optimal designs. When $\boldsymbol{\theta}$ is given or assumed to be known,  an $\boldsymbol{\xi} \in \boldsymbol{\Xi}$ that maximizes $f(\boldsymbol{\xi}) = |{\mathbf F}(\boldsymbol{\xi}, \boldsymbol{\theta})|$ is
known as a {\it locally D-optimal} design \citep{chernoff1953}, where ${\mathbf F}(\boldsymbol{\xi}, \boldsymbol{\theta}) = \sum_{i=1}^m w_i {\mathbf F}({\mathbf x}_i, \boldsymbol{\theta})$ is the Fisher information matrix associated with  $\boldsymbol{\xi}$. 
The ForLion algorithm was proposed by \cite{huang2024forlion} for finding locally D-optimal  designs with mixed factors.

When $\boldsymbol{\theta}$ is unknown while a prior distribution or probability measure $Q(\cdot)$ on $\boldsymbol{\Theta}$ is available instead, we adopt the EW D-optimality \citep{atkinson2007, ymm2016, ytm2016, bu2020, huang2025constrained} and look for $\boldsymbol{\xi} \in \boldsymbol{\Xi}$ maximizing 
\begin{equation}\label{eq:f_EW}
f_{\rm EW}(\boldsymbol{\xi}) = |E\{{\mathbf F}(\boldsymbol{\xi}, \boldsymbol{\Theta})\}| = \left| \sum_{i=1}^m w_i E\left\{ {\mathbf F}({\mathbf x}_i, \boldsymbol{\Theta})\right\} \right|\ ,
\end{equation}
called an {\it integral-based EW D-optimal} approximate design, where 
\begin{equation}\label{eq:E_F_xi_theta}
E\left\{ {\mathbf F}({\mathbf x}_i, \boldsymbol{\Theta})\right\} = \int_{\boldsymbol{\Theta}} {\mathbf F}({\mathbf x}_i, \boldsymbol{\theta}) Q(d\boldsymbol{\theta})
\end{equation}
is a $p\times p$ matrix after entry-wise expectation with respect to $Q(\cdot)$ on $\boldsymbol{\Theta}$. By replacing $\boldsymbol{\theta}$ in ${\mathbf F}(\boldsymbol{\xi}, \boldsymbol{\theta})$ with $\boldsymbol{\Theta}$, we indicate that the expectation $E\left\{ {\mathbf F}(\boldsymbol{\xi}, \boldsymbol{\Theta})\right\}$ is taken with respect to a random parameter vector in $\boldsymbol{\Theta}$. 

When $\boldsymbol{\theta}$ is unknown but a dataset from a prior or pilot study is available, following \cite{bu2020} and \cite{huang2025constrained}, we bootstrap the original dataset to obtain a set of bootstrapped datasets and their corresponding parameter vectors $\{\hat{\boldsymbol{\theta}}_1, \ldots, \hat{\boldsymbol{\theta}}_B\}$ by fitting the parametric model on each of the bootstrapped datasets. In other words, we may replace the prior distribution $Q(\cdot)$ in \eqref{eq:E_F_xi_theta} with the empirical distribution of $\{\hat{\boldsymbol{\theta}}_1, \ldots, \hat{\boldsymbol{\theta}}_B\}$. That is, we may look for $\boldsymbol{\xi} \in \boldsymbol{\Xi}$ maximizing 
\begin{equation}\label{eq:f_BEW}
f_{\rm SEW}(\boldsymbol{\xi}) = |\hat{E}\{{\mathbf F}(\boldsymbol{\xi}, \boldsymbol{\Theta})\}| = \left| \sum_{i=1}^m w_i \hat{E}\left\{ {\mathbf F}({\mathbf x}_i, \boldsymbol{\Theta})\right\} \right|\ ,
\end{equation}
called a {\it sample-based EW D-optimal} approximate design, where 
\begin{equation}\label{eq:hat_E_F_xi_theta}
\hat{E}\left\{ {\mathbf F}({\mathbf x}_i, \boldsymbol{\Theta})\right\} = \frac{1}{B} \sum_{j=1}^B {\mathbf F}({\mathbf x}_i, \hat{\boldsymbol{\theta}}_j)
\end{equation}
is a bootstrapped estimate of $E\left\{ {\mathbf F}({\mathbf x}_i, \boldsymbol{\Theta})\right\}$ based on $\{\hat{\boldsymbol{\theta}}_1, \ldots, \hat{\boldsymbol{\theta}}_B\}$. According to \cite{bu2020}, the EW D-optimal design obtained via bootstrapped samples is a good approximation to the Bayesian D-optimal design obtained in a similar way.

For multinomial logistic models, such as cumulative logit models, the feasible parameter space $\boldsymbol\Theta$ may not be rectangular (see Example~\ref{ex:trauma_clinical_trial} below), which makes the integral in \eqref{eq:E_F_xi_theta} difficult to calculate. In that case, even with a prior distribution instead of a dataset, we may also draw a random sample $\{\hat{\boldsymbol{\theta}}_1, \ldots, \hat{\boldsymbol{\theta}}_B\}$ from the prior distribution, and adopt the sample-based EW D-optimality. 

\begin{example}\label{ex:trauma_clinical_trial}
A trauma clinical trial was originally described by \cite{chuang1997} and analyzed under a cumulative logit model with non-proportional odds \citep{agresti2010,bu2020}. According to Example~5.2 in \cite{bu2020}, the feasible parameter space $\boldsymbol\Theta = \{ \boldsymbol\theta = (\beta_{11}, \beta_{12}, \beta_{21}, \beta_{22}, \beta_{31}, \beta_{32}, \beta_{41},  \beta_{42})^T \in \mathbb{R}^8 \mid \beta_{11} + \beta_{12} x < \beta_{21} + \beta_{22} x < \beta_{31} + \beta_{32} x < \beta_{41} + \beta_{42} x, \mbox{ for all }x\in \{1,2,3,4\}\}$, which is not rectangular. Suppose a multivariate normal distribution $N_8(\hat{\boldsymbol{\mu}}, \hat{\boldsymbol{\Sigma}})$ restricted to $\boldsymbol{\Theta}$ is adopted as the prior distribution, where $\hat{\boldsymbol{\mu}}, \hat{\boldsymbol{\Sigma}}$ are estimated from a previous study \citep{atkinson1999}. We may use the following rejection sampling procedure (see, e.g., Section~6.2.3 in \cite{givens2012computational}, for more general scenarios) to draw a random sample: {\it (i)} Sample $\boldsymbol{\theta} \sim N_8(\hat{\boldsymbol{\mu}}, \hat{\boldsymbol{\Sigma}})$; {\it (ii)} keep $\boldsymbol{\theta}$ only if $\boldsymbol{\theta} \in \boldsymbol{\Theta}$; and {\it (iii)} repeat the previous steps until having accumulated a sample of the desired size.
\hfill{$\Box$}
\end{example}

According to the two simulation studies in Section~\ref{sec:robustness_sampled_based} of the Supplementary Material, in terms of relative efficiency, the sample-based EW D-optimal designs are fairly stable against the choice of random sample or set of bootstrapped parameter vectors.

\subsection{Characteristics of EW D-optimal designs}
\label{sec:theoretical_justification}

In this paper, the design region ${\cal X}$ takes the general form of $\prod_{j=1}^k I_j\ \times \ {\cal D}$. Whenever the number of continuous factors $k\geq 1$, ${\cal X}$ contains infinitely many design points, which is different in nature from the experiments considered by \cite{ymm2016, ytm2016} and \cite{bu2020}. Both theoretical justifications and numerical algorithms for EW D-optimal designs here are far more difficult. 

In this section, to characterize EW D-optimal designs, we fix either a prior distribution $Q(\cdot)$ on $\boldsymbol{\Theta}$ for integral-based EW D-optimality, or a set of bootstrapped or sampled parameter vectors $\{\hat{\boldsymbol{\theta}}_1, \ldots, \hat{\boldsymbol{\theta}}_B\}$ for sample-based EW D-optimality, so that either $E\left\{ {\mathbf F}({\mathbf x}_i, \boldsymbol{\Theta})\right\}$ is defined as in \eqref{eq:E_F_xi_theta}, or $\hat{E}\left\{ {\mathbf F}({\mathbf x}_i, \boldsymbol{\Theta})\right\}$ is defined as in \eqref{eq:hat_E_F_xi_theta}.
Similarly to \cite{fedorov2014} and \cite{huang2024forlion},  we list below the needed assumptions: 

\begin{itemize}
    \item[(A1)] The design region ${\cal X} \subset \mathbb{R}^d$ is compact.    
    \item[(A2)] For each $\boldsymbol{\theta} \in \boldsymbol{\Theta}$, the Fisher information ${\mathbf F}({\mathbf x}, \boldsymbol{\theta})$ is element-wise continuous with respect to all continuous factors of ${\mathbf x}\in {\cal X}$.    
    \item[(A3)] For each $\boldsymbol{\theta} \in \boldsymbol{\Theta}$, the Fisher information $\mathbf{F}(\mathbf{x}, \boldsymbol\theta) = (F_{st}({\mathbf x}, \boldsymbol{\theta}))_{s,t=1,\ldots, p}$ is element-wise continuous with respect to all continuous factors of ${\mathbf x}\in {\cal X}$, and there exists a nonnegative and integrable function $K(\boldsymbol\theta)$ on $\boldsymbol{\Theta}$, such that, $\int_{\boldsymbol{\Theta}} K(\boldsymbol{\theta}) Q(d\boldsymbol{\theta}) < \infty$, and $|F_{st}(\mathbf{x}, \boldsymbol\theta)| \leq K(\boldsymbol\theta)$, for all $s,t\in \{1, \ldots, p\}$, $\mathbf{x} \in \mathcal{X}$, and $\boldsymbol{\theta}\in \boldsymbol{\Theta}$.
\end{itemize}

In this paper, ${\cal X}$ can be $\prod_{j=1}^d I_j$ with $k=d$,  $\prod_{j=1}^k I_j\ \times {\cal D}$ with $1\leq k\leq d-1$, or ${\cal D} \subseteq \prod_{j=k+1}^d I_j$ with $k=0$. In each case, ${\cal X}$ is bounded and closed in $\mathbb{R}^d$ and thus Assumption~(A1) is satisfied.
Apparently, (A3) always implies (A2). Assumption~(A2) is for sample-based EW D-optimality, while (A3) is for integral-based EW D-optimality. If $k=0$, that is, all factors are discrete, (A2) is automatically satisfied.

To characterize the sample-based and integral-based D-optimal EW designs, we follow \cite{fedorov2014} and extend the collection $\boldsymbol{\Xi}$ of designs (each design consists only of a finite number of design points) to $\boldsymbol{\Xi}({\cal X})$, which consists of all probability measures on ${\cal X}$. In other words, a design $\boldsymbol{\xi} \in \boldsymbol{\Xi}({\cal X})$ is also a probability measure on ${\cal X}$, and for each $\boldsymbol{\theta} \in \boldsymbol{\Theta}$,
\[
{\mathbf F}(\boldsymbol{\xi}, \boldsymbol{\theta}) = \int_{\cal X} {\mathbf F}({\mathbf x}, \boldsymbol{\theta})\boldsymbol{\xi}(d{\mathbf x})\ .
\]
Then for the sample-based EW criterion and $\boldsymbol{\xi} \in \boldsymbol{\Xi}({\cal X})$, we denote
\begin{equation}\label{eq:F_SEW_xi}
    {\mathbf F}_{\rm SEW}(\boldsymbol{\xi}) = \hat{E}\{{\mathbf F}(\boldsymbol{\xi}, \boldsymbol{\Theta})\} = \frac{1}{B} \sum_{j=1}^B {\mathbf F}(\boldsymbol{\xi}, \hat{\boldsymbol{\theta}}_j) = \int_{\cal X} \hat{E}\{{\mathbf F}({\mathbf x}, \boldsymbol{\Theta})\} \boldsymbol{\xi}(d{\mathbf x})\ .
\end{equation}
For the integral-based EW criterion, under Assumption~(A3) and Fubini Theorem (see, e.g., Theorem~5.9.2 in \cite{resnick2003probability}), we denote
\begin{eqnarray}
{\mathbf F}_{\rm EW}(\boldsymbol{\xi}) &=& E\{{\mathbf F}(\boldsymbol{\xi}, \boldsymbol{\Theta})\} = \int_{\boldsymbol{\Theta}} {\mathbf F}(\boldsymbol{\xi}, \boldsymbol{\theta}) Q(d\boldsymbol{\theta}) \nonumber\\
&=& \int_{\boldsymbol{\Theta}} \int_{\cal X} {\mathbf F}({\mathbf x}, \boldsymbol{\theta})\boldsymbol{\xi}(d{\mathbf x}) Q(d\boldsymbol{\theta})\nonumber\\
&=& \int_{\cal X} \int_{\boldsymbol{\Theta}} {\mathbf F}({\mathbf x}, \boldsymbol{\theta}) Q(d\boldsymbol{\theta}) \boldsymbol{\xi}(d{\mathbf x})\>\>\>\>\>\mbox{(by Fubini Theorem)}\nonumber\\
&=& \int_{\cal X} E\{{\mathbf F}({\mathbf x}, \boldsymbol{\Theta})\} \boldsymbol{\xi}(d{\mathbf x})\ .\label{eq:F_EW_xi}
\end{eqnarray}
Furthermore, we denote the collections of Fisher information matrices, \[{\cal F}_{\rm SEW}({\cal X}) = \{{\mathbf F}_{\rm SEW}(\boldsymbol{\xi}) \mid \boldsymbol{\xi} \in \boldsymbol{\Xi}({\cal X})\}\] and 
\[{\cal F}_{\rm EW}({\cal X}) = \{{\mathbf F}_{\rm EW}(\boldsymbol{\xi}) \mid \boldsymbol{\xi} \in \boldsymbol{\Xi}({\cal X})\}\ .\]
As a summary of the arguments in Section~\ref{sec:Assumptions_and_proof} of the Supplementary Material, we have the following lemma:

\begin{lemma}\label{lem:compactness}
Under Assumptions~(A1) and (A2), ${\cal F}_{\rm SEW}({\cal X})$ is convex and compact; while under Assumptions~(A1) and (A3), ${\cal F}_{\rm EW}({\cal X})$ is convex and compact.    
\end{lemma}

With the aid of Lemma~\ref{lem:compactness} and Carath\'{e}odory's Theorem (see, e.g., Theorem~2.1.1 in \cite{fedorov1972}), we obtain the following theorem, whose proof, as well as proofs for other theorems, is relegated to Section~\ref{sec:Proofs for the theorems} in the Supplementary Material.

\begin{theorem}\label{thm:thm_2_1_n0} If Assumptions~(A1) and (A2) hold, then for any $p\times p$ matrix ${\mathbf F} \in {\cal F}_{\rm SEW}({\cal X})$, there exists a design $\boldsymbol{\xi} \in \boldsymbol{\Xi} \subset \boldsymbol{\Xi}({\cal X})$ with no more than $m = p(p+1)/2+1$ points, such that ${\mathbf F}_{\rm SEW}(\boldsymbol{\xi}) = {\mathbf F}$. If ${\mathbf F}$ is a boundary point of the convex set ${\cal F}_{\rm SEW}({\cal X})$, then we need no more than $p(p+1)/2$ support points for $\boldsymbol{\xi}$. If Assumptions~(A1) and (A3) are satisfied, then the same conclusions hold for ${\cal F}_{\rm EW}({\cal X})$ and ${\mathbf F}_{\rm EW}(\boldsymbol{\xi})$ as well. 
\end{theorem}

To characterize D-optimality for EW designs, we define the objective function $\Psi({\mathbf F}) = -\log |{\mathbf F}|$ for ${\mathbf F} \in {\cal F}_{\rm SEW}({\cal X})$ or ${\cal F}_{\rm EW}({\cal X})$. Then $f_{\rm SEW}(\boldsymbol{\xi}) = \exp\{-\Psi(\hat{E}\{{\mathbf F}(\boldsymbol{\xi}, \boldsymbol{\Theta})\})\}$, and $f_{\rm EW}(\boldsymbol{\xi}) = \exp\{-\Psi(E\{{\mathbf F}(\boldsymbol{\xi}, \boldsymbol{\Theta})\})\}$. Note that minimizing $\Psi({\mathbf F})$ for ${\mathbf F} \in {\cal F}_{\rm SEW}({\cal X})$ or ${\cal F}_{\rm EW}({\cal X})$ is equivalent to maximizing $f_{\rm SEW}(\boldsymbol{\xi})$ or $f_{\rm EW}(\boldsymbol{\xi})$ for $\boldsymbol{\xi} \in \boldsymbol{\Xi}$, respectively, due to Theorem~\ref{thm:thm_2_1_n0}.

For D-optimality, according to Section~2.4.2 in \cite{fedorov2014}, $\Psi$ always satisfies their Assumptions~(B1), (B2), and (B4). In our notations, we let $\boldsymbol{\Xi}(q)$ denote $\{\boldsymbol{\xi} \in \boldsymbol{\Xi} \mid |\hat{E}\{{\mathbf F}(\boldsymbol{\xi}, \boldsymbol{\Theta})\}|\geq q\}$ for sample-based EW D-optimality, or $\{\boldsymbol{\xi} \in \boldsymbol{\Xi} \mid |E\{{\mathbf F}(\boldsymbol{\xi}, \boldsymbol{\Theta})\}|\geq q\}$ for integral-based EW D-optimality. We still need the following assumption:

\begin{itemize}
    \item[(B3)] There exists a $q>0$, such that $\boldsymbol{\Xi}(q)$ is non-empty.
\end{itemize}

\begin{theorem}\label{thm:thm 2.2 page 64} For sample-based EW D-optimality under Assumptions~(A1), (A2), and (B3), or integral-based EW D-optimality under Assumptions~(A1), (A3), and (B3), we must have (i) there exists an optimal design $\boldsymbol{\xi}^*$ that contains no more than $p(p+1)/2$ design points; (ii) the set of optimal designs is convex; and (iii) a design $\boldsymbol{\xi}$ is EW D-optimal if and only if $\max_{\mathbf{x}\in\mathcal{X}} d(\mathbf{x},\boldsymbol{\xi}) \leq p$, where $d(\mathbf{x},\boldsymbol{\xi}) = {\rm tr}([E\{{\mathbf F}(\boldsymbol{\xi}, \boldsymbol{\Theta})\}]^{-1} E\{{\mathbf F}({\mathbf x}, \boldsymbol{\Theta})\})$ for integral-based EW D-optimality, or ${\rm tr}([\hat{E}\{{\mathbf F}(\boldsymbol{\xi}, \boldsymbol{\Theta})\}]^{-1} \hat{E}\{{\mathbf F}({\mathbf x}, \boldsymbol{\Theta})\})$ for sample-based EW D-optimality.
\end{theorem}

\section{Algorithms for EW D-optimal Designs}
\label{sec:Algorithms}

\subsection{EW lift-one algorithm for discrete factors only}\label{sec:lift_one_EW_discrete}

When all factors are discrete, the design region ${\cal X}={\cal D}$ contains only a finite number of distinct experimental settings. In this case, we may denote the corresponding design settings as ${\mathbf x}_1, \ldots, {\mathbf x}_m$~. The EW D-optimal design problem in this case is to find ${\mathbf w}^* = (w_1^*, \ldots, w_m^*)^T \in S = \{(w_1, \ldots, w_m)^T \in \mathbb{R}^m \mid w_i\geq 0, i=1, \ldots, m; \sum_{i=1}^m w_i=1\}$, which maximizes $f_{\rm EW}({\mathbf w}) = | \sum_{i=1}^m w_i E\{ {\mathbf F}({\mathbf x}_i, \boldsymbol{\Theta})\}|$ or $f_{\rm SEW}({\mathbf w}) = | \sum_{i=1}^m w_i \hat{E} \{ {\mathbf F}({\mathbf x}_i,$ $\boldsymbol{\Theta})\} |$. 

The lift-one algorithm (see Algorithm~3 in the Supplementary Material of \cite{huang2025constrained}) can be used for finding EW D-optimal designs, with $f({\mathbf w})$ replaced by $f_{\rm EW}({\mathbf w})$ or $f_{\rm SEW}({\mathbf w})$, called the {\it EW lift-one algorithm}.

To speed up the lift-one algorithm, \cite{ym2015} provided analytic solutions in their Lemma~4.2 for generalized linear models. \cite{bu2020} mentioned but with no details that the corresponding optimization problems for multinomial logistic models (MLM) also have analytic solutions when the number of response categories $J\leq 5$. In Section~\ref{sec:analytic_solution_MLM} of the Supplementary Material, we provide all the relevant formulae for the lift-one algorithm under an MLM with $J\leq 5$.

\subsection{EW ForLion algorithm and a rounding algorithm}\label{sec:lift_one_EW_mixed}

When at least one factor is continuous, we follow the ForLion algorithm proposed by \cite{huang2024forlion} for locally D-optimal designs, to find EW D-optimal designs, called the {\it EW ForLion algorithm} (see Algorithm~\ref{alg:EW ForLion}).

When there is no confusion, we denote {\it (i)} ${\mathbf F}(\boldsymbol{\xi})$ for $E\{{\mathbf F}(\boldsymbol{\xi}, \boldsymbol{\Theta})\}$ under integral-based EW D-optimality or $\hat{E}\{{\mathbf F}(\boldsymbol{\xi}, \boldsymbol{\Theta})\}$ under sample-based EW D-optimality, for each $\boldsymbol{\xi} \in \boldsymbol{\Xi}$;  and {\it (ii)} ${\mathbf F}_{\mathbf x}$ for $E\{{\mathbf F}({\mathbf x}, \boldsymbol{\Theta})\}$ under integral-based EW D-optimality or $\hat{E}\{{\mathbf F}({\mathbf x}, \boldsymbol{\Theta})\}$ under sample-based EW D-optimality, for each ${\mathbf x} \in {\cal X}$.  
As a result, $d({\mathbf x}, \boldsymbol{\xi})$ in Theorem~\ref{thm:thm 2.2 page 64} can be simplified to $d({\mathbf x}, \boldsymbol{\xi}) = {\rm tr}({\mathbf F}(\boldsymbol{\xi})^{-1} {\mathbf F}_{\mathbf x})$.

\begin{center}
\hrule
\captionof{algorithm}{\bf EW ForLion Algorithm}\label{alg:EW ForLion}
\hrule
\begin{algorithmic}
     \State \textbf{Step 1:} Specify the merging threshold $\delta>0$ (e.g., $10^{-2}$) and the converging threshold $\epsilon>0$ (e.g., $10^{-6}$), and construct an initial design $\boldsymbol{\xi}_0=\{(\mathbf{x}_i^{(0)}, w_i^{(0)}), i=1, \ldots, m_0\} \in \boldsymbol{\Xi}$, such that, $\|\mathbf{x}_i^{(0)}-\mathbf{x}_j^{(0)}\| \geq \delta$ for all $i \neq j$, and $|{\mathbf F}(\boldsymbol{\xi}_0)|>0$ (see Remark~\ref{remark:initial_design}).

     \State \textbf{Step 2:} Merging close design points: Given $\boldsymbol{\xi}_t = \{(\mathbf{x}_i^{(t)}, w_i^{(t)}), i=1, \ldots, m_t\} \in \boldsymbol{\Xi}$ obtained at the $t$th iteration, detect if there are any two points $\mathbf{x}_i^{(t)}$ and $\mathbf{x}_j^{(t)}$ for which $\|\mathbf{x}_i^{(t)}-\mathbf{x}_j^{(t)}\|<\delta$. For such a pair, tentatively merge them into their midpoint $(\mathbf{x}_i^{(t)}+\mathbf{x}_j^{(t)}) / 2$ with combined weight $w_i^{(t)}+w_j^{(t)}$, and denote the resulting design as $\boldsymbol{\xi}_{\rm mer}$~. 
     If $|\mathbf{F}(\boldsymbol{\xi}_{\rm mer})| > 0$, replace the two design points with their merged one and reduce $m_t$ by $1$; otherwise, do not merge them. Stop if all such pairs have been examined.

     \State \textbf{Step 3:} Given $\boldsymbol{\xi}_t \in \boldsymbol{\Xi}$ after the merging step, apply the lift-one algorithm (see Algorithm~3 in the Supplementary Material of \cite{huang2025constrained} and Remark~1 in \cite{huang2024forlion}) with the converging threshold $\epsilon$, and obtain an EW D-optimal allocation $(w_1^*, \ldots, w_{m_t}^*)^T$ for the design points $\{{\mathbf x}_1^{(t)}, \ldots, {\mathbf x}_{m_t}^{(t)}\}$. Replace $w_i^{(t)}$ with $w_i^*$~, respectively.
     
     \State \textbf{Step 4:} Deleting step: Update $\boldsymbol{\xi}_t$ by discarding any ${\mathbf x}_i^{(t)}$ with $w_i^{(t)}=0$.
     
     \State \textbf{Step 5:} Adding new point: Given $\boldsymbol{\xi}_t \in \boldsymbol{\Xi}$, find $\mathbf{x}^* \in {\cal X}$ that maximizes $d(\mathbf{x}, \boldsymbol{\xi}_t)$. More specifically, if all factors are continuous, $\mathbf{x}^*$ can be obtained by the ``L-BFGS-B'' quasi-Newton method \citep{byrd1995limited} directly. If the first $k$ factors are continuous with $1\leq k\leq d-1$, we first find 
        \[
          \mathbf{x}_{(1)}^* \;=\; {\rm argmax}_{\mathbf{x}_{(1)} \in \prod_{j=1}^k [a_j, b_j]}\,
          d \left((\mathbf{x}_{(1)}^T,\, \mathbf{x}_{(2)}^T)^T,\, \boldsymbol{\xi}_t\right)
        \]
     for each $\mathbf{x}_{(2)} \in {\cal D} \subseteq \prod_{j=k+1}^d I_j$~, and then choose $\mathbf{x}_{(2)}^*$ that yields the largest $d(((\mathbf{x}_{(1)}^*)^T,\, \mathbf{x}_{(2)}^T)^T,\, \boldsymbol{\xi}_t)$. Note that $\mathbf{x}_{(1)}^*$ depends on $\mathbf{x}_{(2)}$~. Then $\mathbf{x}^* = ((\mathbf{x}_{(1)}^*)^T,\, (\mathbf{x}_{(2)}^*)^T)^T$. 
    
  \State \textbf{Step 6:} If $d\left(\mathbf{x}^*, \boldsymbol{\xi}_t\right) \leq p$, proceed to Step 7. Otherwise, obtain $\boldsymbol{\xi}_{t+1}$ by adding $(\mathbf{x}^*, 0)$ to $\boldsymbol{\xi}_t$ and increasing $m_t$ by $1$, and return to Step~2.

  \State \textbf{Step 7:} Output $\boldsymbol{\xi}_t$ as an EW D-optimal design.
\end{algorithmic}
\hrule
\end{center}

\begin{remark}\label{remark:new_forlion}
The EW ForLion algorithm (Algorithm~\ref{alg:EW ForLion}) is modified from the ForLion algorithm proposed by \cite{huang2024forlion}, which aims for locally D-optimal designs, to find EW D-optimal designs. Compared with the original ForLion algorithm, Algorithm~\ref{alg:EW ForLion} is more stable and can be applied to more general design regions. In particular, firstly, in Step~2 of Algorithm~\ref{alg:EW ForLion}, we address a possible issue when the merging step leads to a degenerated ${\mathbf F}(\boldsymbol{\xi}_{\rm mer})$, which seems more likely for EW optimality based on our experience. Theoretically, ${\mathbf F}(\boldsymbol{\xi})$ in Algorithm~\ref{alg:EW ForLion}, standing for $E\{{\mathbf F}(\boldsymbol{\xi}, \boldsymbol{\Theta})\}$ or $\hat{E}\{{\mathbf F}(\boldsymbol{\xi}, \boldsymbol{\Theta})\}$, is a convex combination of nonnegative definite matrices. As long as ${\mathbf F}(\boldsymbol{\xi}, \boldsymbol{\theta})$ is positive definite for some $\boldsymbol{\theta}$, so is ${\mathbf F}(\boldsymbol{\xi})$ (see, e.g., Section~10.1 in \cite{seber2008}). Nevertheless, ${\mathbf F}(\boldsymbol{\xi}_{\rm mer})$ may become degenerated in Step~2 when merging some support points of the current design, which was not considered in the original ForLion algorithm. 
Secondly, in Step~5, we allow a subset ${\cal D}$ rather than  $\prod_{j=k+1}^d I_j$ of a product type to be considered for discrete factors, so that Algorithm~\ref{alg:EW ForLion} is more practically useful than the original ForLion algorithm, especially when not all level combinations of discrete factors are under consideration (see Examples~\ref{ex:minimizing_surface_example} and \ref{ex:paper_feeder_example}).
\hfill{$\Box$}
\end{remark}

\begin{remark}\label{remark:initial_design}
Similarly to the original ForLion algorithm \citep{huang2024forlion}, the initial design $\boldsymbol{\xi}_0$ described in Step~1 of Algorithm~\ref{alg:EW ForLion} can be constructed sequentially, except that the design region ${\cal X} = \prod_{j=1}^k I_j\ \times \ {\cal D}$ rather than $\prod_{j=1}^d I_j$~. In particular, we may construct $\boldsymbol{\xi}_0=\{(\mathbf{x}_i^{(0)}, w_i^{(0)}), i=1, \ldots, m_0\}$ as follows: {\it (i)} Sample $\mathbf{x}_1^{(0)}$ uniformly from ${\cal X}$ and denote $s=1$; {\it (ii)} repeat sampling ${\mathbf x}_{\rm new}$ uniformly from ${\cal X}$ until $\|\mathbf{x}_{\rm new}-\mathbf{x}_j^{(0)}\| \geq \delta$ for each $j=1, \ldots, s$; {\it (iii)} denote ${\mathbf x}_{s+1}^{(0)}={\mathbf x}_{\rm new}$ and increase $s$ by $1$; and {\it (iv)} repeat (ii) and (iii) until $s$ attains a large enough integer $m_0$ such that $|{\mathbf F}(\boldsymbol{\xi}_0)|>0$ with $\boldsymbol{\xi}_0=\{(\mathbf{x}_i^{(0)}, m_0^{-1}), i=1, \ldots, m_0\}$.  
\hfill{$\Box$}
\end{remark}

According to Step~6 in Algorithm~\ref{alg:EW ForLion},  $\boldsymbol{\xi}_t$ reported by the EW ForLion algorithm satisfies $\max_{{\mathbf x}\in {\cal X}} d({\mathbf x}, \boldsymbol{\xi}_t) \leq p$. As a direct conclusion of Theorem~\ref{thm:thm 2.2 page 64}, we have the following corollary:

\begin{corollary}\label{cor:EW_D_optiality}
For the sample-based EW D-optimality under Assumptions (A1), (A2), and (B3), or the integral-based EW D-optimality under Assumptions~(A1), (A3), and (B3), the design obtained by Algorithm~\ref{alg:EW ForLion} must be EW D-optimal. 
\end{corollary}

Similarly to the ForLion algorithm \citep{huang2024forlion}, we may relax in practice the stopping rule $d\left(\mathbf{x}^*, \boldsymbol{\xi}_t\right) \leq p$ in Step~6 to $d\left(\mathbf{x}^*, \boldsymbol{\xi}_t\right) \leq p + \epsilon$. Following the arguments in Remark~2 of \cite{huang2024forlion}, as a direct conclusion of Theorem~\ref{thm:thm 2.2 page 64}, Algorithm~\ref{alg:EW ForLion} is guaranteed to stop in finite steps.

Since the design found by Algorithm~\ref{alg:EW ForLion} is an approximate design with continuous design points, to facilitate users in practice, we propose a rounding algorithm (see Algorithm~\ref{alg:appro to exact}) to convert an approximate design to an exact design with a user-specified set of grid points. Different from the rounding algorithms proposed in the literature (see, e.g., Algorithm~2 in \cite{huang2025constrained}), Algorithm~\ref{alg:appro to exact} here involves rounding both the weights and the levels of continuous factors.

\medskip
\begin{center}
\hrule
\captionof{algorithm}{\bf Rounding Algorithm}\label{alg:appro to exact}  \hrule  
\begin{algorithmic}
     \State \textbf{Step 0}: Input: An approximate design $\boldsymbol{\xi} = \{({\mathbf x}_i, w_i), i=1,\ldots,m_0\} \in \boldsymbol{\Xi}$ with all $w_i>0$ and $|{\mathbf F}(\boldsymbol{\xi})|>0$, a pre-specified merging threshold  $\delta_2$ (e.g., $\delta_2=0.1$, typically larger than $\delta$ in Algorithm~\ref{alg:EW ForLion}), grid levels (or pace lengths) $L_1, \ldots, L_k$ for the $k$ continuous factors (e.g., $L_1= \cdots = L_k=0.25$, typically larger than $\delta_2$), respectively, and the total number $n$ of experimental units.

     \State \textbf{Step 1:} Merging step: For each pair of $\mathbf{x}_i$ and $\mathbf{x}_j$~, define their distance $d_{ij}=\left\|\mathbf{x}_{i}-\mathbf{x}_{ j}\right\|$ if their levels of all discrete factors are identical; and $\infty$ otherwise. If $d_{ij} < \delta_2$, tentatively merge ${\mathbf x}_i$ and ${\mathbf x}_j$ into a single design point $(w_i \mathbf{x}_{i}+w_j\mathbf{x}_{j}) /(w_i+w_j)$ with combined weight $w_i+w_j$~, and denote the resulting design as $\boldsymbol{\xi}_{\rm mer}$~. If $|{\mathbf F}(\boldsymbol{\xi}_{\rm mer})|>0$, replace $\boldsymbol{\xi}$ with $\boldsymbol{\xi}_{\rm mer}$; otherwise, do not merge. Stop if all such pairs have been examined.
     
     \State \textbf{Step 2:} Rounding design points to grid levels: Round the levels of the continuous factors of each design point to their nearest multiples of the corresponding $L_j$~, $j=1, \ldots, k$, respectively. 
     
     \State \textbf{Step 3:} Allocating experimental units: First set $n_i=\lfloor n w_i\rfloor$, the largest integer no more than $nw_i$~, for each $i$, and then allocate any remaining units one by one to design points with $n w_i>n_i$~, in the order of increasing ${\mathbf F}(\boldsymbol{\xi})$ the most (see Algorithm~2 in \cite{huang2025constrained}).
     
     \State \textbf{Step 4:} Deleting step: Discard any ${\mathbf x}_i$ for which $n_i=0$ and denote by $m$ the number of remaining design points.

     \State \textbf{Step 5:} Output: Exact design $\{({\mathbf x}_i, n_i), i=1,\ldots,m\}$ with $n_i>0$ and $\sum_{i=1}^m n_i=n$.
\end{algorithmic}
\hrule
\end{center}

\subsection{Calculating $E\{{\mathbf F}({\mathbf x}, \boldsymbol{\Theta})\}$ or $\hat{E}\{{\mathbf F}({\mathbf x}, \boldsymbol{\Theta})\}$}\label{sec:E_F_x}

As illustrated by \cite{ymm2016, ytm2016} and \cite{bu2020}, one major advantage of EW D-optimality over Bayesian D-optimality is that it optimizes $|E\{{\mathbf F}({\mathbf x}, \boldsymbol{\Theta})\}|$ or $|\hat{E}\{{\mathbf F}({\mathbf x},\allowbreak \boldsymbol{\Theta})\}|$, instead of $E\{\log |{\mathbf F}({\mathbf x}, \boldsymbol{\Theta})|\}$.
To implement the EW ForLion algorithm (Algorithm~\ref{alg:EW ForLion}) for EW D-optimal designs, in this section, we show by two examples that instead of calculating ${\mathbf F}_{\mathbf x} = E\{{\mathbf F}({\mathbf x}, \boldsymbol{\Theta})\}$ or $\hat{E}\{{\mathbf F}({\mathbf x}, \boldsymbol{\Theta})\}$ directly, we may focus on key components of the Fisher information matrix to make the computation more efficiently.

\begin{example}\label{ex:MLM_E_Fx}{\bf Multinomial logistic models (MLM):}\quad {\rm   
A general MLM \citep{pmcc1995, atkinson1999, bu2020} for categorical responses with $J$ categories can be defined by ${\mathbf C}^T \log({\mathbf L}\boldsymbol{\pi}_i) = \boldsymbol{\eta}_i = {\mathbf X}_i \boldsymbol{\theta}$, $i=1, \ldots, m$, with $(2J-1)\times J$ constant matrices ${\mathbf C}$ and ${\mathbf L}$, category probabilities $\boldsymbol{\pi}_i = (\pi_{i1}, 
\ldots, \pi_{iJ})^T$, $J\times p$ model matrices ${\mathbf X}_i$~, and model parameters $\boldsymbol{\theta} \in \boldsymbol{\Theta} \subseteq \mathbb{R}^p$.  According to Theorem~2 of \cite{huang2024forlion}, to calculate the Fisher information ${\mathbf F}_{\mathbf x}$ at ${\mathbf x} \in {\cal X}$, it is enough to replace $u_{st}^{\mathbf x}$ (or $u_{st}^{\mathbf x}(\boldsymbol{\theta})$ as it is a function of $\boldsymbol{\theta}$) with $E\{u_{st}^{\mathbf x}(\boldsymbol{\Theta})\} = \int_{\boldsymbol{\Theta}} u_{st}^{\mathbf x}(\boldsymbol{\theta}) Q(d\boldsymbol{\theta})$ or $\hat{E}\{u_{st}^{\mathbf x}(\boldsymbol{\Theta})\} = B^{-1}\sum_{j=1}^B u_{st}^{\mathbf x}(\hat{\boldsymbol{\theta}}_j)$, for $s,t = 1, \ldots, J$.  The formulae for calculating $u_{st}^{\mathbf x}(\boldsymbol{\theta})$ can be found in Appendix~A of \cite{huang2024forlion}.
Then ${\mathbf F}({\mathbf x}, \boldsymbol{\theta}) = {\mathbf X}_{\mathbf x}^T {\mathbf U}_{\mathbf x}(\boldsymbol{\theta}) {\mathbf X}_{\mathbf x}$ according to Corollary~3.1 in \cite{bu2020}, where 
\begin{equation}\label{eq:X_x}
{\mathbf X}_{\mathbf x}= \begin{pmatrix}
 {\mathbf h}_1^T({\mathbf x}) &  \boldsymbol0^T & \cdots & \boldsymbol0^T& {\mathbf h}_c^T({\mathbf x})\\
 \boldsymbol0^T &  {\mathbf h}_2^T({\mathbf x}) &\ddots & \vdots & \vdots\\
\vdots &  \ddots& \ddots &  \boldsymbol0^T & {\mathbf h}_c^T({\mathbf x})\\
\boldsymbol0^T & \cdots & \boldsymbol0^T & {\mathbf h}_{J-1}^T({\mathbf x}) & {\mathbf h}_c^T({\mathbf x})\\
 \boldsymbol0^T & \cdots & \cdots & \boldsymbol0^T & \boldsymbol0^T\\
\end{pmatrix}_{J \times p}
\end{equation}
with predetermined predictor functions ${\mathbf h}_j=(h_{j1}, \ldots, h_{jp_j})^T$, $j=1, \ldots, J-1$, ${\mathbf h}_c=(h_1, \ldots, h_{p_c})^T$, and ${\mathbf U}_{\mathbf x}(\boldsymbol{\theta}) = \left(u_{st}^{\mathbf x}(\boldsymbol{\theta})\right)_{s,t=1, \ldots, J}$~. Then $E\{{\mathbf F}({\mathbf x}, \boldsymbol{\Theta})\} = {\mathbf X}_{\mathbf x}^T E\{{\mathbf U}_{\mathbf x}(\boldsymbol{\Theta})\} {\mathbf X}_{\mathbf x}$ and $\hat{E}\{{\mathbf F}({\mathbf x}, \boldsymbol{\Theta})\} = {\mathbf X}_{\mathbf x}^T \hat{E}\{{\mathbf U}_{\mathbf x}(\boldsymbol{\Theta})\} {\mathbf X}_{\mathbf x}$~.
}\hfill{$\Box$}
\end{example}

\begin{theorem}\label{thm:MLM design points} For MLMs under Assumptions~(A1) and (B3), suppose all the predictor functions ${\mathbf h}_1, \ldots, {\mathbf h}_{J-1}$ and ${\mathbf h}_c$ defined in \eqref{eq:X_x} are continuous with respect to all continuous factors of $\mathbf{x}\in {\cal X}$, and the parameter space $\boldsymbol{\Theta} \subseteq\mathbb{R}^p$ is bounded. Then there exists an EW D-optimal design $\boldsymbol{\xi}$ that contains no more than $p(p+1)/2$ support points, and the design obtained by Algorithm~\ref{alg:EW ForLion} must be EW D-optimal.
\end{theorem}

The boundedness of $\boldsymbol{\Theta}$ in Theorem~\ref{thm:MLM design points} is a practical requirement. For typical applications, the boundedness of a working parameter space is often needed for numerical searches for parameter estimates or theoretical derivations for desired properties \citep{ferguson1996course}.

\begin{example}\label{ex:GLM_E_Fx}{\bf Generalized linear models (GLM):}\quad {\rm   
A generalized linear model \citep{pmcc1989, dobson2018} can be defined by $E(Y_i) = \mu_i$ and  $\eta_i = g(\mu_i) = {\mathbf X}_i^T\boldsymbol{\theta}$, where $g$ is a given link function, ${\mathbf X}_i = {\mathbf h}({\mathbf x}_i) = (h_1({\mathbf x}_i), \ldots, h_p({\mathbf x}_i))^T$, $i=1, \ldots, m$, and $\boldsymbol{\theta} \in \boldsymbol{\Theta} \subseteq \mathbb{R}^p$ (see, e.g., Section~4 in \cite{huang2024forlion}).  Then ${\mathbf F}({\mathbf x}, \boldsymbol{\theta}) = \nu\{{\mathbf h}({\mathbf x})^T \boldsymbol{\theta}\} \cdot {\mathbf h}({\mathbf x}) {\mathbf h}({\mathbf x})^T$, where $\nu = \{(g^{-1})'\}^2/s$ with $s(\eta_i) = {\rm Var}(Y_i)$ (see Table~5 in the Supplementary Material of \cite{huang2025constrained} for examples of $\nu$). To calculate ${\mathbf F}_{\mathbf x} = E\{{\mathbf F}({\mathbf x}, \boldsymbol{\Theta})\}$ or $\hat{E}\{{\mathbf F}({\mathbf x}, \boldsymbol{\Theta})\}$ under GLMs, it is enough to replace $\nu\{{\mathbf h}({\mathbf x})^T \boldsymbol{\theta}\}$ with $E[\nu\{{\mathbf h}({\mathbf x})^T \boldsymbol{\Theta}\}] = \int_{\boldsymbol{\Theta}} \nu\{{\mathbf h}({\mathbf x})^T \boldsymbol{\theta}\} Q(d\boldsymbol{\theta})$ or $\hat{E}[\nu\{{\mathbf h}({\mathbf x})^T \boldsymbol{\Theta}\}] $ $= B^{-1}\sum_{j=1}^B \nu\{{\mathbf h}({\mathbf x})^T \hat{\boldsymbol{\theta}}_j\}$.  That is, $E\{{\mathbf F}({\mathbf x}, \boldsymbol{\Theta})\} = E\left[\nu\{{\mathbf h}({\mathbf x})^T \boldsymbol{\Theta}\}\right] \cdot {\mathbf h}({\mathbf x}) {\mathbf h}({\mathbf x})^T$, and $\hat{E}\{{\mathbf F}({\mathbf x}, $ $\boldsymbol{\Theta})\} = \hat{E}\left[\nu\{{\mathbf h}({\mathbf x})^T \boldsymbol{\Theta}\}\right] \cdot {\mathbf h}({\mathbf x}) {\mathbf h}({\mathbf x})^T$.

To calculate ${\mathbf F}(\boldsymbol{\xi}, \boldsymbol{\theta})$ under GLMs, instead of $\sum_{i=1}^m w_i {\mathbf F}({\mathbf x}_i, \boldsymbol{\theta})$, we recommend its matrix form ${\mathbf F}(\boldsymbol{\xi}, \boldsymbol{\theta}) = {\mathbf X}_{\boldsymbol\xi}^T{\mathbf W}_{\boldsymbol\xi}(\boldsymbol{\theta}) {\mathbf X}_{\boldsymbol\xi}$~, 
where ${\mathbf X}_{\boldsymbol\xi} =$ $({\mathbf h}({\mathbf x}_1), \ldots,$ ${\mathbf h}({\mathbf x}_m))^T \in \mathbb{R}^{m\times p}$, and ${\mathbf W}_{\boldsymbol\xi}(\boldsymbol{\theta}) =  {\rm diag}\{w_1 \nu\{{\mathbf h}({\mathbf x}_1)^T {\boldsymbol\theta}\}, \ldots, w_m \nu\{{\mathbf h}({\mathbf x}_m)^T {\boldsymbol\theta}\}\}$. Then $E\{{\mathbf F}(\boldsymbol{\xi}, \boldsymbol{\theta})\} = {\mathbf X}_{\boldsymbol\xi}^T E\{{\mathbf W}_{\boldsymbol\xi}(\boldsymbol{\Theta})\}$ ${\mathbf X}_{\boldsymbol\xi}$ and $\hat{E}\{{\mathbf F}(\boldsymbol{\xi}, \boldsymbol{\theta})\} = {\mathbf X}_{\boldsymbol\xi}^T \hat{E}\{{\mathbf W}_{\boldsymbol\xi}(\boldsymbol{\Theta})\}{\mathbf X}_{\boldsymbol\xi}$~.
The matrix form is much more efficient when programming in R.
}\hfill{$\Box$}
\end{example}

\begin{theorem}\label{thm:GLM design points} For GLMs under Assumptions~(A1) and (B3), suppose all the predictor functions ${\mathbf h} = (h_1, \ldots, h_p)^T$ are continuous with respect to all continuous factors of $\mathbf{x} \in {\cal X}$, and $\boldsymbol{\Theta}$ is bounded. Then there exists an EW D-optimal design $\boldsymbol{\xi}$ that contains no more than $p(p+1)/2$ support points, and the design obtained by Algorithm~\ref{alg:EW ForLion} must be EW D-optimal.
\end{theorem}

\subsection{First-order derivative of sensitivity function}\label{sec:first_order_d}

To implement the EW ForLion algorithm (Algorithm~\ref{alg:EW ForLion}) for EW D-optimal designs, if there is at least one continuous factor, we need to calculate the first-order derivative of the sensitivity function $d({\mathbf x}, \boldsymbol{\xi})$ in Theorem~\ref{thm:thm 2.2 page 64}. 
Following the simplified notation in Section~\ref{sec:lift_one_EW_mixed}, $d({\mathbf x}, \boldsymbol{\xi}) = {\rm tr}({\mathbf F}(\boldsymbol{\xi})^{-1} {\mathbf F}_{\mathbf x})$.

Similarly to Theorem~3 in \cite{huang2024forlion}, we obtain the following theorem for multinomial logistic models (MLM) under EW D-optimality:

\begin{theorem}\label{thm:sensitivity function eq thm for mlm}
For an MLM model under Assumptions~(A1) and (B3), consider $\boldsymbol{\xi}\in \boldsymbol{\Xi}$ satisfying $|{\mathbf F}(\boldsymbol{\xi})| > 0$. If we denote the $p \times p$ matrix ${\mathbf F}(\boldsymbol{\xi})^{-1} = (\mathbf{E}_{st})_{s,t=1, \ldots, J}$ with submatrix $\mathbf{E}_{st} \in \mathbf{R}^{p_s \times p_t}$, and $p_J=p_c$~, ${\mathbf h}_j^{\mathbf x} = {\mathbf h}_j({\mathbf x})$, ${\mathbf h}_c^{\mathbf x} = {\mathbf h}_c({\mathbf x})$, $u_{st}^{\mathbf x} = E\{u_{st}^{\mathbf x}(\boldsymbol{\Theta})\}$ or $\hat{E}\{u_{st}^{\mathbf x} (\boldsymbol{\Theta})\}$ for simplicity, then
\begin{eqnarray*}
d(\mathbf{x}, \boldsymbol{\xi}) &=&  \sum_{t=1}^{J-1} u_{tt}^{\mathbf x} \left(\mathbf{h}_t^{\mathbf{x}}\right)^T \mathbf{E}_{t t} \mathbf{h}_t^{\mathbf{x}} +\sum_{s=1}^{J-1} \sum_{t=1}^{J-1} u_{st}^{\mathbf x} \left(\mathbf{h}_c^{\mathbf{x}}\right)^T \mathbf{E}_{J J} \mathbf{h}_c^{\mathbf{x}} \\
&+& 2 \sum_{s=1}^{J-2} \sum_{t=s+1}^{J-1} u_{st}^{\mathbf x} \left(\mathbf{h}_t^{\mathbf{x}}\right)^T \mathbf{E}_{s t} \mathbf{h}_s^{\mathbf{x}} +2 \sum_{s=1}^{J-1} \sum_{t=1}^{J-1} u_{st}^{\mathbf x} \left(\mathbf{h}_c^{\mathbf{x}}\right)^T \mathbf{E}_{s J} \mathbf{h}_s^{\mathbf{x}}\ .
\end{eqnarray*}
\end{theorem}

\begin{example}\label{ex:MLM_d_x}{\bf Multinomial logistic models (MLM):}\quad {\rm   
To calculate the first-order derivative of $d({\mathbf x}, \boldsymbol{\xi})$ under MLMs, we follow Appendix~C in \cite{huang2024forlion} after {\it (i)} ${\mathbf F}(\boldsymbol{\xi})$ is replaced with $E\{{\mathbf F}(\boldsymbol{\xi}, \boldsymbol{\Theta})\}$ or $\hat{E}\{{\mathbf F}(\boldsymbol{\xi}, \boldsymbol{\Theta})\}$; {\it (ii)} ${\mathbf F}_{\mathbf x}$ is replaced with $E\{{\mathbf F}({\mathbf x}, \boldsymbol{\Theta})\}$ or $\hat{E}\{{\mathbf F}({\mathbf x}, \boldsymbol{\Theta})\}$; {\it (iii)} ${\mathbf U}_{\mathbf x} = (u_{st}^{\mathbf x})_{s,t=1, \ldots, J}$ with $u_{st}^{\mathbf x}$  replaced by $E\{u_{st}^{\mathbf x}(\boldsymbol{\Theta})\}$ or $\hat{E}\{u_{st}^{\mathbf x}(\boldsymbol{\Theta})\}$; and {\it (iv)} when $k\geq 1$, for $i=1, \ldots, k$, $\partial{\mathbf U}_{\mathbf x}/\partial x_i = (\partial u_{st}^{\mathbf x}/\partial x_i)_{s,t=1, \ldots, J}$ with $\partial u_{st}^{\mathbf x}/\partial x_i$ replaced by 
\[
\frac{\partial E\{u_{st}^{\mathbf x}(\boldsymbol{\Theta})\}}{\partial x_i} = \frac{\partial }{\partial x_i}\int_{\boldsymbol{\Theta}} u_{st}^{\mathbf x}(\boldsymbol{\theta}) Q(d\boldsymbol{\theta})\>\mbox{ or }\>\hat{E}\left\{\frac{\partial u_{st}^{\mathbf x}}{\partial x_i}(\boldsymbol{\Theta})\right\} = \frac{1}{B} \sum_{j=1}^B \frac{\partial u_{st}^{\mathbf x}}{\partial x_i}(\hat{\boldsymbol{\theta}}_j)\ .
\]
Note that $(\partial u_{st}^{\mathbf x}/\partial x_i)(\hat{\boldsymbol{\theta}}_j)$ denotes $\partial u_{st}^{\mathbf x}/\partial x_i$ at $\boldsymbol{\theta}=\hat{\boldsymbol{\theta}}_j$~, whose formulae can be found in Appendix~C of \cite{huang2024forlion}. 
}\hfill{$\Box$}
\end{example}

\begin{example}\label{ex:GLM_d_x}{\bf Generalized linear models (GLM):}\quad {\rm   
Under GLMs, the sensitivity function $d({\mathbf x}, \boldsymbol{\xi}) = E[\nu\{{\mathbf h}({\mathbf x})^T \boldsymbol{\Theta}\}]\cdot {\mathbf h}({\mathbf x})^T [E\{{\mathbf F}(\boldsymbol{\xi}, \boldsymbol{\Theta})\}]^{-1} {\mathbf h}({\mathbf x})$ for $f_{\rm EW}(\boldsymbol{\xi})$ or $\hat{E}[\nu\{{\mathbf h}({\mathbf x})^T \boldsymbol{\Theta}\}]\cdot {\mathbf h}({\mathbf x})^T [\hat{E}\{{\mathbf F}(\boldsymbol{\xi}, \boldsymbol{\Theta})\}]^{-1} {\mathbf h}({\mathbf x})$ for $f_{\rm SEW}(\boldsymbol{\xi})$.
Following Section~S.5 of the Supplementary Material of \cite{huang2024forlion}, we denote ${\mathbf A}_E = [{\mathbf X}_{{\boldsymbol\xi}}^T E\{{\mathbf W}_{{\boldsymbol\xi}}(\boldsymbol{\Theta})\} {\mathbf X}_{{\boldsymbol\xi}}]^{-1}$ and ${\mathbf A}_{\hat E} = [{\mathbf X}_{{\boldsymbol\xi}}^T \hat{E}\{{\mathbf W}_{{\boldsymbol\xi}}(\boldsymbol{\Theta})\} {\mathbf X}_{{\boldsymbol\xi}}]^{-1}$.
The first-order derivative of $d({\mathbf x}, \boldsymbol{\xi})$ with respect to continuous factors ${\mathbf x}_{(1)} = (x_1, \ldots, x_k)^T$ is $\partial d({\mathbf x}, \boldsymbol{\xi})/\partial {\mathbf x}_{(1)} = {\mathbf h}({\mathbf x})^T {\mathbf A}_E {\mathbf h}({\mathbf x}) \cdot \partial E[\nu\{{\mathbf h}({\mathbf x})^T {\boldsymbol\Theta}\}]/\partial {\mathbf x}_{(1)} + 2\cdot E[\nu\{{\mathbf h}({\mathbf x})^T {\boldsymbol\Theta}\}] \cdot \{\partial {\mathbf h}({\mathbf x})/\partial {\mathbf x}_{(1)}^T\}^T {\mathbf A}_E {\mathbf h}({\mathbf x})$ for $f_{\rm EW}(\boldsymbol{\xi})$, or 
\begin{eqnarray*}
\frac{\partial d({\mathbf x}, \boldsymbol{\xi})}{\partial {\mathbf x}_{(1)}}
&=& {\mathbf h}({\mathbf x})^T {\mathbf A}_{\hat E} {\mathbf h}({\mathbf x}) \cdot \left\{\frac{\partial {\mathbf h}({\mathbf x})}{\partial {\mathbf x}_{(1)}^T}\right\}^T \cdot \frac{1}{B}\sum_{j=1}^B \nu'\{{\mathbf h}({\mathbf x})^T {\hat{\boldsymbol\theta}}_j\} \hat{\boldsymbol{\theta}}_j\\
&+& \frac{2}{B}\cdot \sum_{j=1}^B \nu\{{\mathbf h}({\mathbf x})^T \hat{\boldsymbol\theta}_j\} \cdot \left\{\frac{\partial {\mathbf h}({\mathbf x})}{\partial {\mathbf x}_{(1)}^T}\right\}^T {\mathbf A}_{\hat E} {\mathbf h}({\mathbf x})
\end{eqnarray*}
for $f_{\rm SEW}(\boldsymbol{\xi})$.
}\hfill{$\Box$}
\end{example}

Similar to Theorem~4 in \cite{huang2024forlion}, we obtain the following result for GLM under EW D-optimality:

\begin{theorem}\label{thm:sensitivity function eq thm for GLM}
For a GLM model under Assumptions~(A1) and (B3), consider $\boldsymbol{\xi} \in \boldsymbol{\Xi}$ satisfying  $|\mathbf{X}_{\boldsymbol{\xi}}^T {\mathbf W}_{{\boldsymbol\xi}} \mathbf{X}_{\boldsymbol{\xi}}|>0$, where ${\mathbf W}_{{\boldsymbol\xi}}$ is either $E\{{\mathbf W}_{{\boldsymbol\xi}}(\boldsymbol{\Theta})\}$ for integral-based EW D-optimality or $\hat{E}\{{\mathbf W}_{{\boldsymbol\xi}}(\boldsymbol{\Theta})\}$ for sample-based EW D-optimality (see Example~\ref{ex:GLM_E_Fx}). Then $\boldsymbol{\xi}$ is EW D-optimal if and only if $\max _{\mathbf{x} \in \mathcal{X}} \allowbreak E[\nu\{{\mathbf h}({\mathbf x})^T \boldsymbol{\Theta}\}]\cdot {\mathbf h}({\mathbf x})^T [\mathbf{X}_{\boldsymbol{\xi}}^T {\mathbf W}_{{\boldsymbol\xi}} \mathbf{X}_{\boldsymbol{\xi}}]^{-1} {\mathbf h}({\mathbf x}) \leq p$ for integral-based D-optimality or $\max _{\mathbf{x} \in \mathcal{X}} \hat{E}[\nu\{{\mathbf h}({\mathbf x})^T \boldsymbol{\Theta}\}]\cdot {\mathbf h}({\mathbf x})^T [\mathbf{X}_{\boldsymbol{\xi}}^T {\mathbf W}_{{\boldsymbol\xi}} \mathbf{X}_{\boldsymbol{\xi}}]^{-1} {\mathbf h}({\mathbf x}) \leq p$ for sample-based EW D-optimality.
\end{theorem}

\section{Applications to Real Experiments}
\label{sec:applications}

In this section, we use the paper feeder experiment (see Section~\ref{sec:intro} and Example~\ref{ex:paper_feeder_example}) and a minimizing surface defects experiment (see also Example~\ref{ex:minimizing_surface_example}) both under MLMs to illustrate the advantages gained by adopting our algorithms and the EW D-optimal designs. For examples under GLMs, please see Section~\ref{sec:Additional Examples: Generalized Linear Models} of the Supplementary Material.

\subsection{Paper feeder experiment}\label{sec:paper_feeder}

In this section, we consider the motivating example mentioned in Section~\ref{sec:intro} and Example~\ref{ex:paper_feeder_example}, the paper feeder experiment. There are eight discrete control factors and one continuous factor, the stack force, under our consideration (see Section~\ref{sec:Model_Selection_Paper_Feeder_Experiment} for more details). As mentioned in Section~\ref{sec:intro}, instead of using two separate GLMs to model the two types of failures, {\tt misfeed} and {\tt multifeed}, we consider multinomial logistic models (MLM) to model all three possible outcomes, namely  {\tt misfeed}, {\tt normal}, and {\tt multifeed}, at the same time. Using Akaike information criterion (AIC, \cite{AIC1973information}) and Bayesian information criterion (BIC, \cite{hastie2009elements}), we adopt a cumulative logit model with non-proportional odds as the most appropriate model for this experiment (see Table~\ref{tab:Model comparison for the Paper Feeder Experiment} in Section~\ref{sec:Model_Selection_Paper_Feeder_Experiment} of the Supplementary Material for more details).

For this experiment, we compare ten different designs (see Tables~\ref{tab:Robust Designs for the Paper Feeder Experiment} and \ref{tab:Efficiencies of designs for paper feeder experiment}): (1) ``Original allocation'': the original design that collected 1,785 observations roughly uniformly at 183 distinct experimental settings; (2) ``EW Bu-appro'': an approximate design  obtained by applying \cite{bu2020}'s EW lift-one algorithm on the original 183 distinct settings; (3) ``EW Bu-exact'': an exact design obtained by applying \cite{bu2020}'s EW exchange algorithm on the 183 distinct settings with $n=1,785$; (4) ``EW Bu-grid2.5'' \& (5) ``EW Bu-grid2.5 exact'': approximate and exact designs by applying \cite{bu2020}'s algorithms after discretizing the range $[0,160]$ of stack force into grid points with equal space 2.5; (6) ``EW ForLion'': the proposed EW D-optimal approximate design by applying Algorithm~\ref{alg:EW ForLion} to the continuous factor, stack force, ranging in $[0,160]$; (7)$\sim$(10) exact designs obtained by applying Algorithm~\ref{alg:appro to exact} with different grid levels or pace lengths ($L=0.1, 0.5, 1, 2.5$) and $n=1,785$. As mentioned in Example~\ref{ex:paper_feeder_example}, for illustration purposes, all designs follow the same 18 runs modified from $OA(18, 2^1\times 3^7)$ for the eight discrete factors, while their design points can be different in terms of the levels of the continuous factor.
Since the feasible space of a cumulative logit model is not rectangular (see, e.g., Example~\ref{ex:trauma_clinical_trial}, as well as Example~5.2 in \cite{bu2020} and Example~8 in \cite{huang2025constrained}), we bootstrap the original dataset to collect $B=100$ (due to computational intensity) samples and the corresponding estimated model parameters, denoted by $\hat{\boldsymbol{\theta}}_1, \hat{\boldsymbol{\theta}}_2, \ldots, \hat{\boldsymbol{\theta}}_{100}$~. All the EW designs are under the sample-based EW D-optimality with respect to the 100 parameter vectors.  
The constructed designs are presented in Tables~\ref{tab:D-Optimal designs for the paper feeder experiment} and \ref{tab:EW Forlion optimal designs} of the Supplementary Material. 

In Table~\ref{tab:Robust Designs for the Paper Feeder Experiment}, we list the numbers of design points $m$, the computational times in seconds for obtaining the designs, the objective function values, and the relative efficiencies compared with the recommended EW ForLion design, that is, $(|\hat{E}\{{\mathbf F}(\boldsymbol{\xi}, \boldsymbol \Theta)\}|/|\hat{E}\{{\mathbf F}(\boldsymbol{\xi}_{\rm EWForLion},$ $\boldsymbol\Theta)\}|)^{1/p}$ with $p=32$. Compared with our EW ForLion design, the original design on $183$ settings only achieves $63.03\%$ relative efficiency. Having applied \cite{bu2020}'s algorithms on the same 183 settings, the relative efficiencies of EW Bu-appro and EW Bu-exact, $88.16\%$ and $88.15\%$, are not satisfactory either. When we apply \cite{bu2020}'s algorithms on grid points with length $2.5$, it takes much longer time to find EW Bu-grid2.5 and EW Bu-grid2.5 exact designs, which also consist of more design points (55 and 44). The EW ForLion design and its exact designs at different grid levels require the minimum number of experimental settings, less computational time, and satisfactory relative efficiencies.

\begin{table}[ht]
    \centering
    \captionsetup{skip=2pt} 
    \caption{Robust designs for the paper feeder experiment}
    {
    \renewcommand{\arraystretch}{0.5}
        \resizebox{\textwidth}{!}{
    \begin{threeparttable}
          \begin{tabular}{ccrccrccccr}
    \toprule
    Designs &  &$m$& & & Time (s) & & &$|\hat{E}\{{\mathbf F}(\boldsymbol{\xi}, \boldsymbol \Theta)\}|$ &   &Relative Efficiency \\
    \midrule
    Original allocation&  &183&  &  & - &  &  &1.342e+28&   &63.03\% \\
     EW Bu-appro & & 38 &  &  & 116s &  &  &6.162e+32&   &88.16\%\\
    EW Bu-exact &  & 38 &  &  & 1095s &  &  &6.159e+32&   &88.15\%\\
   EW Bu-grid2.5&   &55 &  &  &  15201s &  &  &2.078e+34&   &98.40\%\\
     EW Bu-grid2.5 exact&  &44 &  &  & 69907s &  &  &2.078e+34&  &98.40\%\\
      EW ForLion&  & 38 &  &  &  1853s  &  &  &3.482e+34&   &100.00\%\\
     EW ForLion exact grid0.1&   & 38 &  &  &  0.28s &  &  & 2.122e+34 &  &98.46\%\\
     EW ForLion exact grid0.5&   & 38 &  &  &  0.24s &  &  & 2.065e+34 &  &98.38\%\\
     EW ForLion exact grid1&   & 38 &  &  &  0.23s &  &  &  1.579e+34 &  &97.56\%\\
     EW ForLion exact grid2.5&   & 38 &  &  &  0.22s &  &  & 1.326e+34 &  &97.03\%\\
    \bottomrule
    \end{tabular}
     \begin{tablenotes}
         \footnotesize
          \setlength{\baselineskip}{8pt} 
         \item Note: Time for EW ForLion exact designs are for applying Algorithm~\ref{alg:appro to exact} to EW ForLion (approximate) design.
         \normalsize
     \end{tablenotes}
    \end{threeparttable}
    }
    }
    \label{tab:Robust Designs for the Paper Feeder Experiment}
\end{table}

To assess the robustness of these designs, we first use the original ForLion algorithm \citep{huang2024forlion} to find the corresponding locally D-optimal design $\boldsymbol{\xi}_j^*$ for each $\hat{\boldsymbol{\theta}}_j$~.  Then the robustness of a design $\boldsymbol{\xi}\in \boldsymbol{\Xi}$ can be evaluated by the 100 relative efficiencies $(|\mathbf{F}(\boldsymbol{\xi}, \hat{\boldsymbol{\theta}}_j)| / |\mathbf{F}(\boldsymbol{\xi}_j^*, \hat{\boldsymbol{\theta}}_j)|)^{1/p}$, $j=1, \ldots, 100$.
Table~\ref{tab:Efficiencies of designs for paper feeder experiment} (see also Figure~\ref{fig:PFE_simulation_results} in the Supplementary Material) presents the five-number summary of the 100 relative efficiencies. Overall, the EW ForLion design is the most robust one against parameter misspecifications. For practical uses, we recommend the grid-0.1 exact design (EW ForLion exact grid0.1), which has nearly the same robustness as the EW D-optimal approximate design and the smallest number of distinct experimental settings (i.e., 38).

\begin{table}[ht]
\centering
\captionsetup{skip=2pt} 
\caption{Summary of 100 relative efficiencies against locally D-optimal designs for paper feeder experiment}
 {
    \renewcommand{\arraystretch}{0.5}
        \resizebox{\textwidth}{!}{
\begin{tabularx}{\textwidth}{c*{6}{X}}  %
\toprule
Robust design &Min  & $Q_1$  & Median   & $Q_3$ &Max\\  
    \hline
Original allocation  & 0.5075  &0.6005& 0.6115 & 0.6308 &0.6587\\ 
EW Bu-appro  & 0.6813  &0.8133 & 0.8343  & 0.8488&0.8741\\
EW Bu-exact& 0.6813  &0.8134& 0.8346  & 0.8489&0.8742\\
EW Bu-grid2.5  & 0.7592 &0.9183& 0.9327   & 0.9417&0.9513\\
EW Bu-grid2.5 exact  & 0.7594 &0.9184& 0.9328   & 0.9417  &0.9515\\
EW ForLion & 0.8363  &0.9277& 0.9364  & 0.9445 &0.9542\\
EW ForLion exact grid0.1  & 0.8028  &0.9276& 0.9361 &  0.9443  &0.9541\\
EW ForLion exact grid0.5   & 0.7569  &0.9269& 0.9363 & 0.9447 &0.9534\\
EW ForLion exact grid1   & 0.7502  &0.9163 & 0.9315   & 0.9389  &0.9489 \\
EW ForLion exact grid2.5   & 0.7548   &0.9098 & 0.9253  & 0.9348 &0.9439\\
\toprule
\end{tabularx}
}
}
\label{tab:Efficiencies of designs for paper feeder experiment}
\end{table}

\subsection{Minimizing surface defects experiment}\label{sec:Msd_experiment}

\cite{wu2008} presented a study that investigated optimal parameter settings in a polysilicon deposition process for circuit manufacturing (see also Example~\ref{ex:minimizing_surface_example}). This experiment involves responses of $J=5$ categories, one discrete factor $x_{i1}$ and five continuous factors $x_{i2} \in [-25, 25], x_{i3} \in [-200, 200], x_{i4} \in [-150, 0], x_{i5} \in [-100,0], x_{i6} \in [0,16]$ (see Table~S1 in the Supplementary Material of \cite{huang2024forlion} or Table~6 in \cite{lukemire_optimal_2022}).
Having simplified the discrete factor $x_{i1}$ into a binary one in $\{-1,1\}$,  \cite{lukemire_optimal_2022} considered a cumulative logit model with proportional odds:
$$
\log \left(\frac{\pi_{i 1}+\cdots+\pi_{i j}}{\pi_{i, j+1}+\cdots+\pi_{i J}}\right)=\theta_j-\beta_1 x_{i 1}-\beta_2 x_{i 2}-\beta_3 x_{i 3}-\beta_4 x_{i 4}-\beta_5 x_{i 5}-\beta_6 x_{i 6}
$$
with $i=1, \ldots, m$, $j=1,2,3,4$, where $\pi_{ij}$ is the probability that the response associated with ${\mathbf x}_i = (x_{i1}, \ldots, x_{i6})^T$ falls into the $j$th category.

Assuming $\boldsymbol{\theta} = (\theta_1, \ldots, \theta_4, \beta_1, \ldots, \beta_6)^T = (-1.113, 0.183, 1.518, 2.639, - \allowbreak 0.970, 0.077,$ $0.008, -0.007, 0.007, 0.056)^T$, \cite{lukemire_optimal_2022} employed a PSO algorithm and derived a locally D-optimal approximate design with 14 design points, while \cite{huang2024forlion} applied their ForLion algorithm and obtained a locally D-optimal approximate design with 17 design points.

In practice, experimenters may not know the precise values of some or all parameters but often have varying degrees of insight into their actual values. To construct a robust design against parameter misspecifications, we sample $B=1,000$ parameter vectors from the prior distribution listed in Table~S1 of the Supplementary Material of \cite{huang2024forlion}, namely $\theta_1, \ldots, \theta_4, \beta_1, \ldots, \beta_6$ are independently from uniform distributions with ranges $(-2, -1), (-0.5, 0.5), (1, 2), (2.5, 3.5), (-1, 0), (0, 0.2), (-0.1, 0.1),$ $(-0.1, 0.1), (-0.1, 0.1), (0, 0.2)$, respectively. 
 
We use Algorithm~\ref{alg:EW ForLion} to construct a sample-based EW D-optimal approximate design (see the left side of Table~\ref{tab:surface_defects_Minimizing_EW_ForLion_design}). Assuming that the total number of experimental units is $n=1,000$, we apply our Algorithm~\ref{alg:appro to exact} with $L_1 = \cdots = L_5=1$ for continuous control variables and obtain an exact design, as shown on the right side of Table~\ref{tab:surface_defects_Minimizing_EW_ForLion_design}. Both our approximate and exact designs have 17 support points. 

For comparison purposes, we also construct a Bayesian D-optimal design for this experiment under the same prior distribution. Since there are five continuous control factors in this experiment, due to the heavy computational cost of Bayesian D-optimality, we follow \cite{bu2020} and discretize each continuous factor into two evenly distributed grid points. The constructed Bayesian D-optimal design, which costs 29,112 seconds (about twice as much as 14,718 seconds cost by the EW ForLion algorithm), contains 64 support points with weights ranging from $0.0007$ to $0.0263$.

\begin{table}[ht]
    \centering
    \captionsetup{skip=2pt} 
    \caption{EW D-optimal approximate (left) and exact (right) designs by EW ForLion and rounding algorithms for minimizing surface defects experiment} 
     {
\renewcommand{\arraystretch}{0.8}
    \resizebox{\textwidth}{!}{
    \begin{tabular}{c|ccccccc|c|ccccccc}
        \hline
         Design  & Cleaning & Deposition & Deposition & Nitrogen & Silane & Settling &  & Design  & Cleaning & Deposition & Deposition & Nitrogen & Silane & Settling &   \\
        Point & Method & Temp. & Pressure & Flow & Flow & Time & $w_i$ & Point & Method & Temp. & Pressure & Flow & Flow & Time & $n_i$ \\
        \hline
  1&  -1 & -25.000 & 200.000&   -150 & 0 & 0 & 0.0475&
  1&  -1 &   -25 & 200  &  -150 & 0  & 0 &48 \\
  
  2&  1 & 25.000 &  0 &    0  &  0  &  0&0.0986&
  2 & 1 &  25  & 0 &   0  &  0  &  0& 99\\
  
  3&   1& 25.000 &   200.000 &   -150 &  0  &  16&0.0380&
  3&   1&  25 &   200 &   -150 &   0 &   16& 38 \\
  
  4&  -1& -25.000 &  0 &   0   & 0  & 16&0.1206&
  4&  -1&  -25 &  0  &  0  &  0 &  16 & 121\\
  
  5&   1& 0 &  0 &   0 &0 &   16&0.0984&
  5&   1&  0 & 0 &   0 &0 &   16& 98\\
  
  6&  -1&  25.000 & -104.910 &-150 &-100 &  16&0.0426 &
  6&  -1&   25 & -105& -150& -100 &  16&  43\\
  
  7&   -1&  14.870 &  -1.299  &  0  &  0   & 0&0.1017&
  7&   -1&   15  &  -1  &  0  &  0  &  0 & 102\\
  
  8&   1& -13.168 &  2.802  &  0  &  0  & 16&0.0399&
  8&   1&   -13 &  3  &  0  &  0  & 16& 40 \\
  
  9&   -1& 25.000 & 111.805 &0 &  -100&   16&0.0533&
  9&   -1&  25 & 112 & 0  &  -100  &  16 & 53 \\
  
 10&   -1 &25.000 &  -72.107 &  0  &   -100  & 16 &0.0432&
 10&   -1 &  25  &-72 &   0  &  -100 &  16&43\\
 
 11&    -1& -25.000 & -170.690  &  -150 & 0 &   0 & 0.0401&
 11&    -1&  -25 & -171 &   -150&  0  &  0& 40  \\
 
 12&    1& 25.000 &-158.631 &-150  &  0 &   16 &0.0470&
 12&    1&  25 & -159 & -150 &   0  &  16&  47 \\
 
 13&    1& -25.000& -200.000 &-150 &   -100 &  0&0.0215&
 13&    1&  -25 &-200& -150 & -100  & 0& 21 \\
 
 14&   1& -25.000 &-20.500&   0 &-100 &  0 &0.0843&
 14&   1&  -25& -20  &  0 &-100  & 0& 84 \\
 
 15&    1 & -25.000& 200.000 &-150  & -100 & 0  &0.0415&
 15&    1 &  -25 &200& -150 &  -100 & 0&  42 \\
 
 16&   1 & -2.925 &  0.017  &  0 &   0 &   0&0.0294&
 16&   1 &  -3 &  0  &  0  &  0  &  0&29\\

 17&   -1 & -10.568 &  -1.093  &  0 &   0 &   0&0.0524&
 17&   -1 &  -11 &  -1  &  0  &  0  &  0&52\\
        \hline
    \end{tabular}
}}
\label{tab:surface_defects_Minimizing_EW_ForLion_design}
\end{table}

To compare the robustness of different designs against parameter misspecifications, we sample 10,000 parameter vectors from the same prior distribution. For each sampled parameter vector $\boldsymbol{\theta}_j$~, we calculate the relative efficiency of a design $\boldsymbol{\xi}$ under comparison with respect to our EW ForLion approximate design $\boldsymbol{\xi}_{\rm EW ForLion}$~, i.e., $(|\mathbf{F}(\boldsymbol{\xi}, \boldsymbol{\theta}_j)|/|\mathbf{F}(\boldsymbol{\xi}_{\rm EW ForLion}, \allowbreak \boldsymbol{\theta}_j)|)^{1/p}$ with $p=10$ in this case.

\begin{figure}[ht]
  \centering
   \begin{subfigure}[b]{0.45\textwidth}
    \centering
    \includegraphics[width=\textwidth]{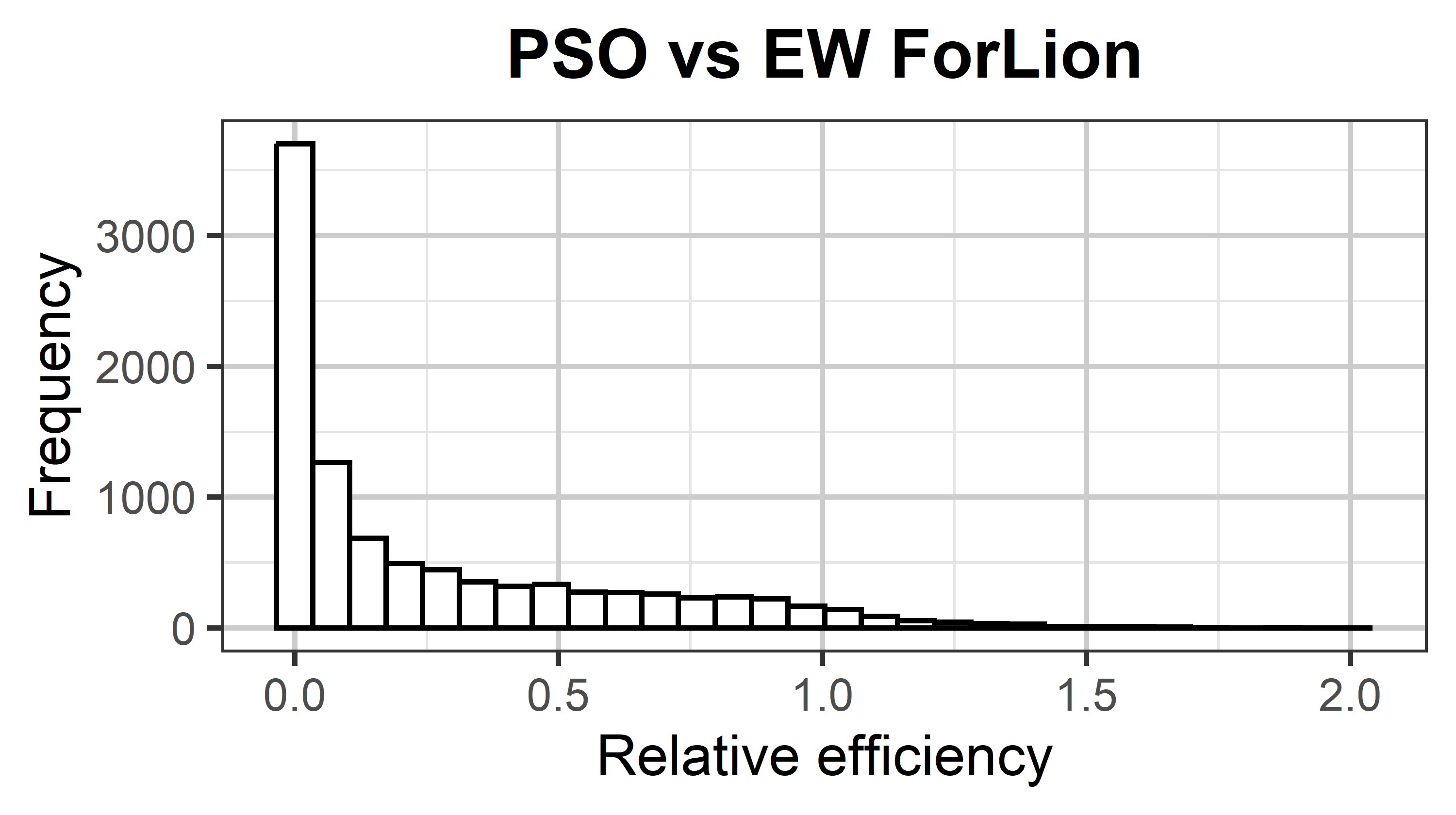} 
    \label{fig:PSO_vs_EW_ForLion}
  \end{subfigure}
  \hfill
  \begin{subfigure}[b]{0.45\textwidth}
    \centering
    \includegraphics[width=\textwidth]{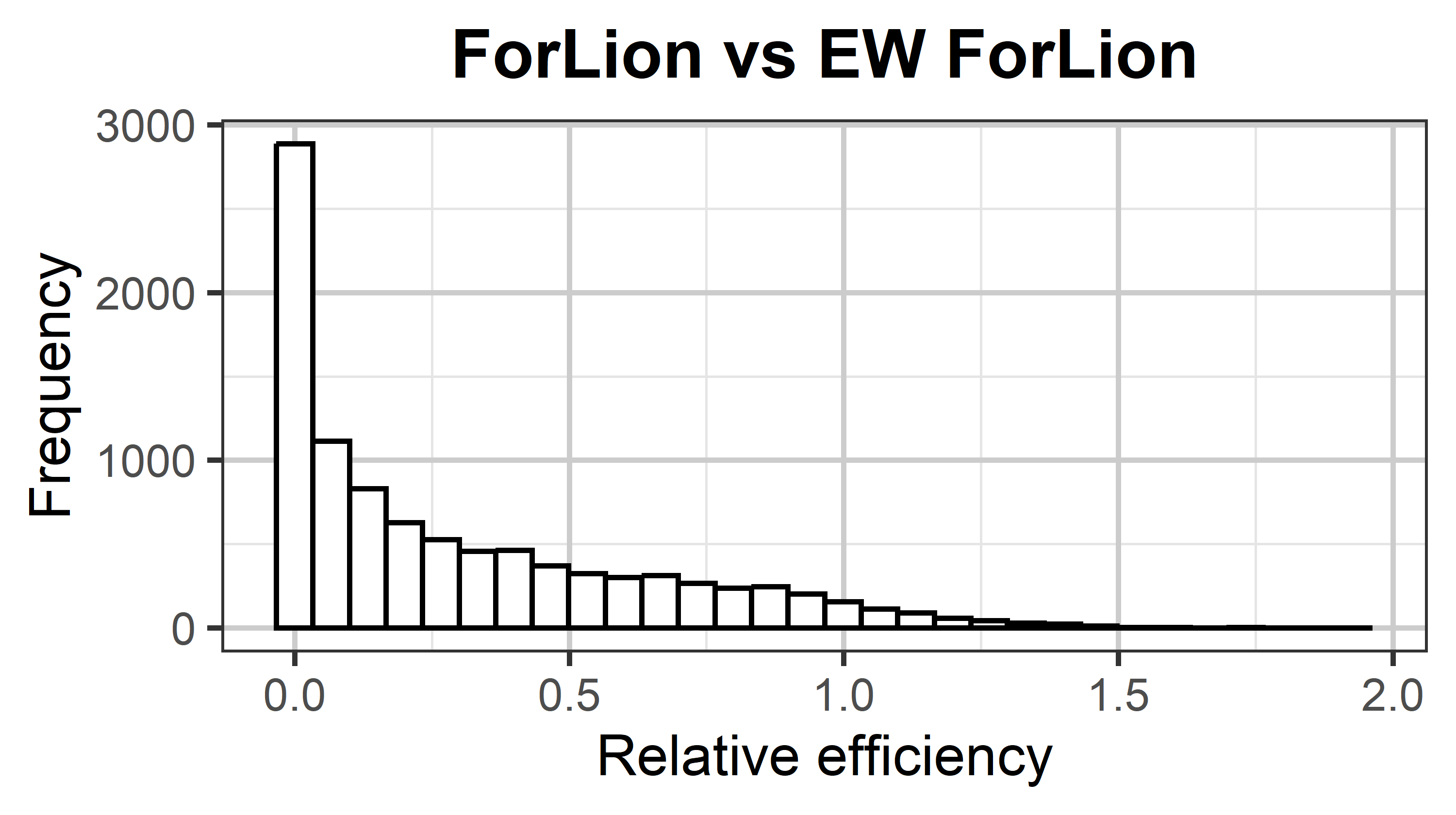} 
    \label{fig: ForLion_vs_EW ForLion}
  \end{subfigure}\\ 
  \begin{subfigure}[b]{0.45\textwidth}
    \centering
    \includegraphics[width=\textwidth]{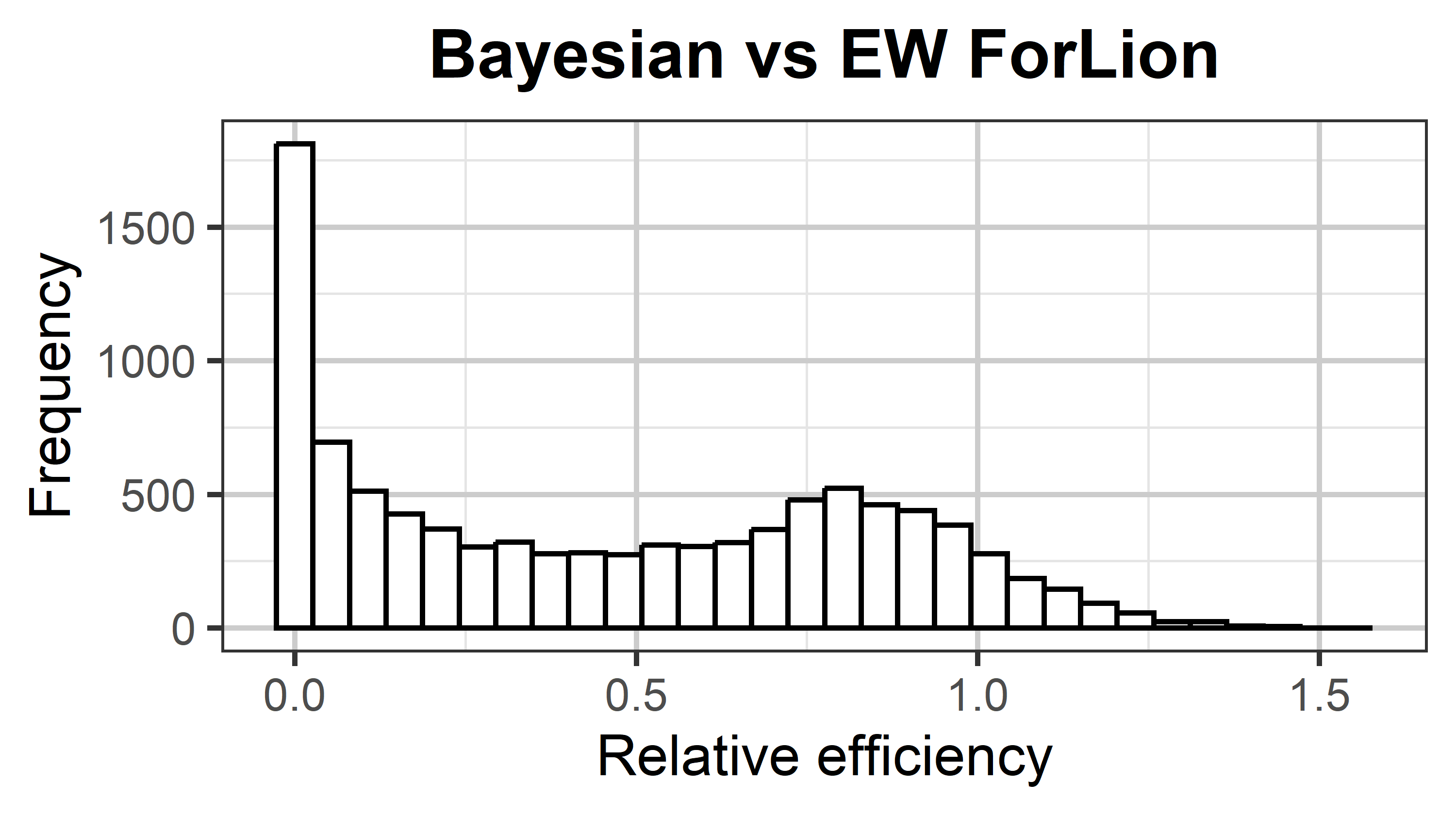} 
    \label{fig:Bayesian_vs_EW ForLion}
  \end{subfigure}
  \hfill
  \begin{subfigure}[b]{0.45\textwidth}
    \centering
    \includegraphics[width=\textwidth]{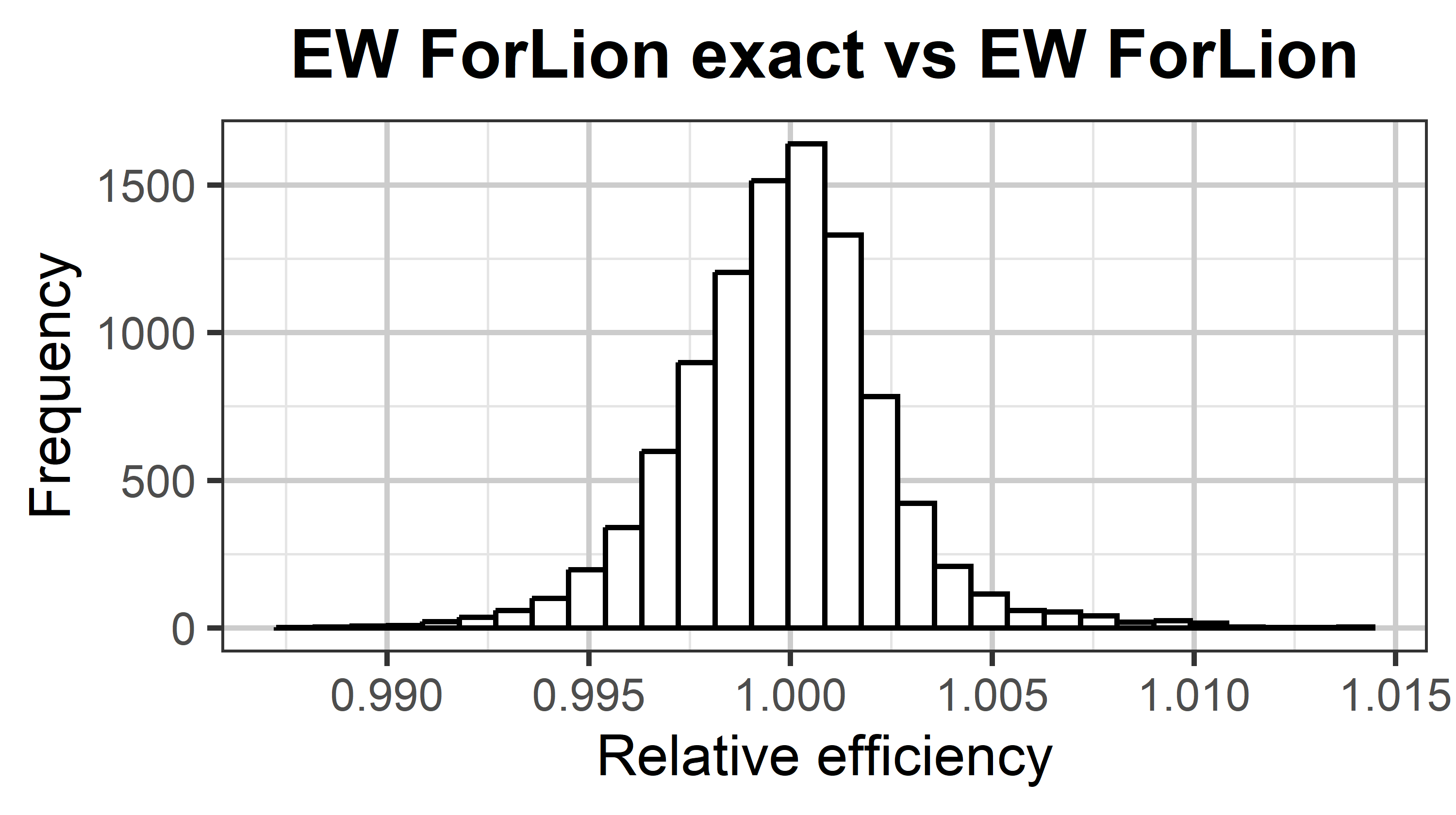} 
    \label{fig:EW_ForLion_exact_vs_EW ForLion}
  \end{subfigure}
  \hfill   \\  
  \begin{subfigure}[b]{0.45\textwidth}
    \centering
    \includegraphics[width=\textwidth]{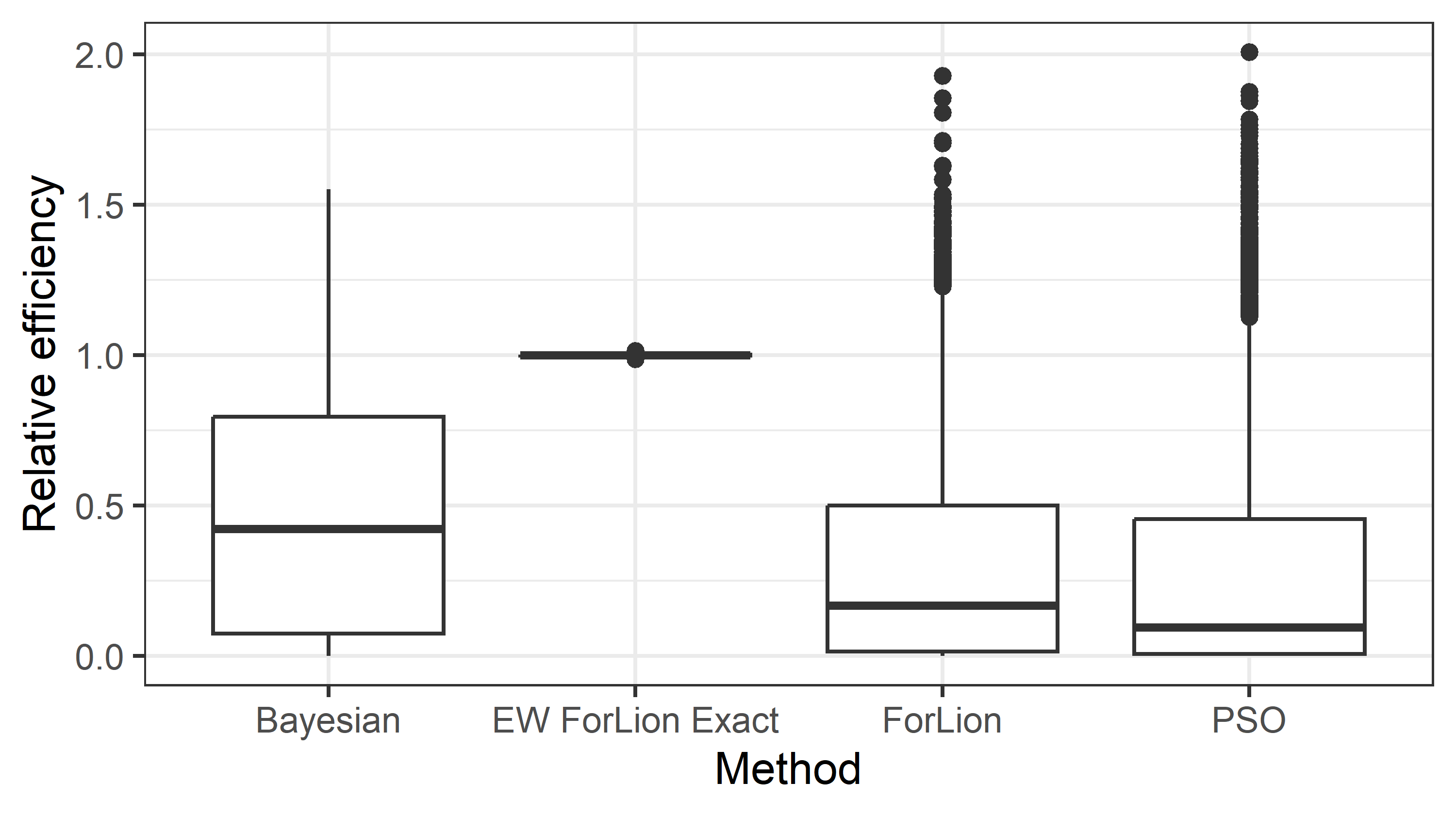} 
    \label{fig:boxplot_for_msd}
  \end{subfigure} \\  [-3ex]
  \caption{Relative efficiency comparison among designs for 10,000 sampled parameter sets in the minimizing surface defects experiment}
  \label{fig:Relative_Efficiency_MSD_updated}
\end{figure}

Figure~\ref{fig:Relative_Efficiency_MSD_updated} shows histograms and boxplots of $10,000$ relative efficiencies of the 14-point PSO-generated
locally D-optimal design (labeled with ``PSO'' in Figure~\ref{fig:Relative_Efficiency_MSD_updated}) by \cite{lukemire_optimal_2022}, the 17-point locally D-optimal design (labeled with ``ForLion'') obtained by  \cite{huang2024forlion}'s original ForLion algorithm, the 64-point Bayesian D-optimal design (labeled with ``Bayesian''),  and the 17-point EW ForLion exact design provided in Table~\ref{tab:surface_defects_Minimizing_EW_ForLion_design} obtained by our Algorithm~\ref{alg:appro to exact}, with respect to the 17-point EW ForLion (approximate) design provided in Table~\ref{tab:surface_defects_Minimizing_EW_ForLion_design}.

As displayed in Figure~\ref{fig:Relative_Efficiency_MSD_updated}, the EW ForLion exact design performs fairly similar to the EW ForLion approximate design. For most scenarios, the EW ForLion approximate design outperforms the two locally D-optimal designs (PSO and ForLion), and their relative efficiencies are below one. It is because the locally D-optimal designs are constructed with respect to a specific nominal parameter vector, which may not perform well if the true parameter vector deviates from the nominal one. On the contrary, the EW ForLion approximate or exact design are much more robust.
As for the Bayesian D-optimal design obtained by discretizing each continuous factor into two levels, it performs better than locally D-optimal designs, but worse than EW ForLion designs. 
Overall we recommend our EW ForLion approximate and exact designs as robust designs for practical applications.

\section{Conclusion and Discussion}
\label{section:Conclusion and Discussion}

In this paper, we characterize EW D-optimal designs for fairly general parametric models with mixed factors and derive detailed formulae for multinomial logistic models and generalized linear models. We develop an EW ForLion algorithm for finding EW D-optimal designs with mixed factors or continuous factors, and a rounding algorithm to convert approximate designs with continuous factors to exact designs with user-specified sets of grid points.  

Depending on the availability of a pilot study or a prior distribution on unknown parameters, we recommend either sample-based or integral-based EW D-optimal designs. Apparently, the obtained EW D-optimal design depends on the set of parameter vectors or the prior distribution chosen.
By applying our algorithms to real experiments, we show that the obtained EW D-optimal designs under a reasonable parameter vector set or prior distribution provide well-performed robust designs against parameter misspecifications. By merging close design settings, our designs may significantly reduce the experimental cost and time by reducing the number of distinct experimental settings.

The sample-based strategy offers an alternative approach for aggregating prior information about model parameters. Compared with the integral-based solution, the sample-based approach can significantly reduce the computational cost while maintaining the robustness of the resulting designs. It may be incorporated with other design criteria, such as minimax and Bayesian optimality, making it a practical tool for researchers.

\section*{Supplementary Materials}

The Supplementary Material includes several sections: S1 provides analytic solutions for the lift-one algorithm with MLMs; S2 discusses assumptions for design region and parametric models; S3 provides proofs of main theorems; S4 displays model selection and design comparison for the paper feeder experiment; S5 discusses robustness of sample-based EW designs; S6 provides two examples for GLMs and integral-based EW D-optimality.
\par

\clearpage



\clearpage
\setcounter{page}{1}
\def\thepage{S\arabic{page}}



\begin{center}
{\Large\bf Expected Weighted D-optimal Designs for\\ Experiments with Mixed Factors}
\end{center}
\vspace{.25cm}
\centerline{Siting Lin$^1$, Yifei Huang$^2$, and Jie Yang$^1$}
\vspace{.4cm}
\centerline{\it $^1$University of Illinois at Chicago and $^2$Astellas Pharma Global Development, Inc.}
\vspace{.55cm}
\centerline{\bf Supplementary Material}
\vspace{.55cm}
\par

\begin{description}
  \item[\textbf{S1}] Analytic Solutions for Multinomial Logistic Models 
  \item[\textbf{S2}] Assumptions and Relevant Results
  \item[\textbf{S3}] Proofs of Main Theorems
  \item[\textbf{S4}] Model Selection and Design Comparison for Paper Feeder Experiment
  \item[\textbf{S5}] Robustness of Sample-based EW Designs
  \item[\textbf{S6}] More Examples
\end{description}

\renewcommand{\thesection}{S\arabic{section}}
\setcounter{section}{0}
\setcounter{equation}{0}
\def\theequation{S\arabic{section}.\arabic{equation}}
\setcounter{table}{0}
\def\thetable{S.\arabic{table}}
\setcounter{figure}{0}
\def\thefigure{S.\arabic{figure}}

\section{Analytic Solutions for Multinomial Logistic Models}\label{sec:analytic_solution_MLM}

According to Theorem~S.9 in the Supplementary Material of \cite{bu2020}, for multinomial logistic models, given a design $\boldsymbol{\xi} = \{({\mathbf x}_i, w_i) \mid i=1, \ldots, m\} \in \boldsymbol{\Xi}$, for $i \in \{1, \ldots, m\}$ and $0<z<1$, we have
\begin{eqnarray*}
  f_i(z) &=& (1-z)^{p-J+1}\sum_{j=0}^{J-1} b_j\,z^j(1-z)^{J-1-j}\ ,\\
f_i'(z) &=& (1-z)^{p-J}\sum_{j=1}^{J-1} b_{j}(j-pz) z^{j-1}(1-z)^{J-1-j}-pb_{0}(1-z)^{p-1}\ ,
\end{eqnarray*}
where $b_0=f_i(0)$, $(b_{J-1},b_{J-2},\ldots,b_1)^T = \mathbf{B}_{J-1}^{-1}\mathbf{c}$, $\mathbf{B}_{J-1} = (s^{t-1})_{s,t=1,2,\ldots,J-1}$~, and $\mathbf{c}=(c_1, c_2, \ldots, c_{J-1})^T$ with $c_j=(j+1)^p\,j^{J-1-p}f_i(1/(j+1))-j^{\,J-1}f_i(0)$, $j=1,\ldots,J-1$.

For typically applications, $p\geq J \geq 3$. Since \(f_i(z)\) is an order-\(p\) polynomial in \(z\) with \(f_i(1)=0\), its maximum on \([0,1]\) is attained either at \(z=0\) or at an interior point \(z \in (0,1)\) satisfying $f_i'(z)=0$, that is, 
\begin{equation}\label{eq:max}
   \sum_{j=1}^{J-1}jb_{j}z^{j-1}(1-z)^{J-j-1}=p\sum_{j=0}^{J-1}b_{j}z^{j}(1-z)^{J-j-1},\> 0<z<1. 
\end{equation}

In \cite{bu2020}, it was mentioned that \eqref{eq:max} has analytic solutions for $J\leq 5$, while no formula was provided accordingly. In this section, we provide explicit solutions to facilitate the programming for multinomial logistic models. Note that we only need the real-valued solutions locating in $(0,1)$ for our purposes. 

According to the analytic solutions to quadratic, cubic, and quartic equations (see, e.g., Sections~3.8.1$\sim$3.8.3 in \cite{abramowitz1964handbook}), we have the following results: 

\begin{lemma}\label{lem:MLM_solution_J_3}
When $J=3$, \eqref{eq:max} can be rewritten as
$A_2z^2+A_1z+A_0=0$ with $A_2 = p(b_{0}-b_{1} + b_2)$, $A_1 = b_{1}(1+p)-2(b_0 p + b_2)$, and $A_0 = b_{0} p -b_{1}$~.
\begin{itemize}
\item[(i)] If $A_2=A_1=A_0=0$, then \eqref{eq:max} has infinitely many solutions in $(0,1)$ and $f_i(z)\equiv f_i(0) \geq 0$. In this case, we may choose $z=0$ to maximize $f_i(z)$.   
\item[(ii)] If $A_2=A_1=0$ but $A_0\neq 0$, then we must have $A_0 > 0$. In this case, \eqref{eq:max} has no solution, while $z=0$ uniquely maximizes $f_i(z)$ with $z\in [0,1]$. 
\item[(iii)] If $A_2=0$ but $A_1\neq 0$, then \eqref{eq:max} has a unique solution $-A_0/A_1 \in \mathbb{R}$.
\item[(iv)] If $A_2 \neq 0$ and $A_1^2-4A_2A_0>0$, then \eqref{eq:max} has two real solutions
\[
\frac{-A_1\pm\sqrt{A_1^2-4A_2A_0}}{2A_2}\ .
\]
\item[(v)] If $A_2 \neq 0$ and $A_1^2-4A_2A_0=0$, then \eqref{eq:max} has only one real solution $-A_1/(2A_2)$.
\item[(vi)] If $A_2 \neq 0$ and $A_1^2-4A_2A_0<0$, then \eqref{eq:max} has no real solution. In this case, we choose $z=0$, which uniquely maximizes $f_i(z)$ with $z\in [0,1]$. 
\end{itemize}
\end{lemma}

\medskip
\noindent
{\bf Proof of Lemma~\ref{lem:MLM_solution_J_3}:}
When $J=3$,
\eqref{eq:max} can be rewritten as
$A_2z^2+A_1z+A_0=0$ with
$A_2 = p(b_{0}-b_{1} + b_2)$, $A_1 = b_{1}(1+p)-2(b_{0}p + b_2)$, $A_0 = b_{0}p -b_{1}$~, and 
\[
f'_i(z) = - (1-z)^{p-J} (A_2z^2+A_1z+A_0)\ .
\]
We only need to verify case~(ii). Actually, in this case, if $A_0 < 0$, then $f'_i(z) = -(1-z)^{p-J} A_0 > 0$ for all $z\in (0,1)$, while $f_i(0) \geq 0 = f_i(1)$ leads to a contradiction.
\hfill{$\Box$}

As a direct conclusion of Sections~3.8.2 in \cite{abramowitz1964handbook}, we obtain the following lemma for $J=4$:

\begin{lemma}\label{lem:MLM_solution_J_4}
When $J=4$, equation~\eqref{eq:max} is equivalent to $A_3z^3+A_2z^2+A_1z+A_0=0$ with $A_0 = b_0 p - b_1$~, $A_1 = - 3 b_0 p + b_1 (2+p) - 2 b_2$~, $A_2 = 3 b_0 p - b_1 (1 + 2p) + b_2 (2 + p) - 3 b_3$~, and $A_3 = p(-b_0 + b_1 - b_2 + b_3)$. The cases with $A_3=0$ has been listed in  Lemma~\ref{lem:MLM_solution_J_3}. If $A_3\neq 0$, \eqref{eq:max} is equivalent to $z^3 + a_2 z^2 + a_1 z + a_0 = 0$ with $a_i=A_i/A_3$, $i=0,1,2$. We let
\[
q=\frac{a_1}{3}-\frac{a_2^2}{9}\ ,\>\>\> r=\frac{a_1a_2-3a_0}{6}-\frac{a_2^3}{27}\ ,
\]
$s_1 = [r+(q^3+r^2)^{1/2}]^{1/3}$, and $s_2 = [r-(q^3+r^2)^{1/2}]^{1/3}$.
\begin{itemize}
\item[(i)] If $q^3 + r^2 >0$, then \eqref{eq:max} has only one real solution $s_1+s_2-a_2/3$~.
\item[(ii)] If $q^3 + r^2 =0$, then \eqref{eq:max} has two real solutions $z_1 = 2 r^{1/3} - a_2/3$ and $z_2 = -r^{1/3} - a_2/3$~.
\item[(iii)] If $q^3 + r^2 < 0$, then \eqref{eq:max} has three real solutions 
\[
z_1 = s_1 + s_2 - \frac{a_2}{3}\ ,\>\>\> z_2, z_3 = -\frac{s_1+s_2}{2} - \frac{a_2}{3} \pm \frac{i\sqrt{3}}{2}(s_1-s_2)
\]
with $i=\sqrt{-1}$~, known as the imaginary unit, which in this case gets involved in $(q^3+r^2)^{1/2}$ as well.
\end{itemize}
\end{lemma}

For $J=5$, we follow the arguments for equation~(12) in \cite{tong2014} and obtain the following results:

\begin{lemma}\label{lem:MLM_solution_J_5}
When $J=5$, equation~\eqref{eq:max} is equivalent to $A_4z^4 + A_3z^3+A_2z^2+A_1z+A_0=0$ with $A_0 = b_0 p - b_1$~, $A_1 = - 4 b_0 p + b_1 (3+p) - 2 b_2$~, $A_2 = 6 b_0 p - 3 b_1 (1 + p) + b_2 (4 + p) - 3 b_3$~, $A_3 = - 4 b_0 p + b_1 (1 + 3p) - 2 b_2 (1+p) + b_3 (3+p) - 4 b_4$~, and $A_4 = p(b_0 - b_1 + b_2 - b_3 + b_4)$. The cases with $A_4=0$ has been listed in  Lemmas~\ref{lem:MLM_solution_J_4} \& \ref{lem:MLM_solution_J_3}. If $A_4\neq 0$, \eqref{eq:max} is equivalent to $z^4 + a_3 z^3 + a_2 z^2 + a_1 z + a_0 = 0$ with $a_i=A_i/A_4$, $i=0,1,2,3$. Then there are four solutions to \eqref{eq:max} calculated as complex numbers:
\[
z_1, z_2 = -\frac{a_3}{4} - \frac{\sqrt{A_*}}{2} \pm  \frac{\sqrt{B_*}}{2}\ ,\>\>\> 
z_3, z_4 = -\frac{a_3}{4} + \frac{\sqrt{A_*}}{2} \pm  \frac{\sqrt{C_*}}{2}
\]
with
\begin{eqnarray*}
A_* &=& -\frac{2  a_2}{3}+\frac{ a_3^2}{4}+\frac{G_*}{3\times 2^{1/3}}\ ,\\
B_* &=& -\frac{4 a_2}{3}+\frac{a_3^2}{2}-\frac{G_*}{3\times 2^{1/3}} - \frac{-8 a_1+4 a_2 a_3-a_3^3}{4 \sqrt{A_*}}\ ,\\
C_* &=& -\frac{4 a_2}{3}+\frac{a_3^2}{2}-\frac{G_*}{3\times 2^{1/3}} + \frac{-8 a_1+4 a_2 a_3-a_3^3}{4 \sqrt{A_*}}\ ,\\
D_* &=& \left(F_*+\sqrt{F_*^2-4 E_*^3}\right)^{1/3}\ ,\\
E_* &=& 12 a_0+a_2^2-3 a_1 a_3\ ,\\
F_* &=& 27 a_1^2-72 a_0 a_2+2 a_2^3-9 a_1 a_2 a_3+27 a_0 a_3^2\ ,\\
G_* &=& D_* + 2^{2/3} E_*/D_*\ .
\end{eqnarray*}
\end{lemma}

\section{Assumptions and Relevant Results}\label{sec:Assumptions_and_proof} 

Following \cite{fedorov2014} and \cite{huang2024forlion},  in this section we discuss the assumptions needed for this paper in details. Note that our assumptions have been adjusted for mixed factors. Recall that a design point or experimental setting ${\mathbf x} = (x_1, \ldots, x_d)^T \in {\cal X} \subseteq \mathbb{R}^d$. Among the $d\geq 1$ factors, the first $k$ factors are continuous, and the last $d-k$ factors are discrete. Note that $0\leq k\leq d$ and $k=0$ implies no continuous factor. When $k\geq 1$, we denote ${\mathbf x}_{(1)} = (x_1, \ldots, x_k)^T$ as the vector of continuous factors.

\begin{lemma}\label{lem:A2_SEW}
Suppose $k\geq 1$ and the parametric model $M({\mathbf x}, \boldsymbol{\theta})$ with ${\mathbf x}\in {\cal X}$ and $\boldsymbol{\theta} \in \boldsymbol{\Theta}$ satisfies Assumption~(A2). Given $\{\hat{\boldsymbol{\theta}}_1, \ldots, \hat{\boldsymbol{\theta}}_B\} \subseteq \boldsymbol{\Theta}$, $\hat{E}\left\{ {\mathbf F}({\mathbf x}, \boldsymbol{\Theta})\right\}$ $ = B^{-1} \sum_{j=1}^B {\mathbf F}({\mathbf x}, \hat{\boldsymbol{\theta}}_j)$ is  element-wise continuous with respect to all continuous factors of ${\mathbf x}\in {\cal X}$.
\end{lemma}

The proof of Lemma~\ref{lem:A2_SEW} is straightforward. It is for sample-based EW designs. For integral-based EW designs, we need the following lemma:

\begin{lemma}\label{lem:A3_EW}
Suppose $k\geq 1$ and the parametric model $M({\mathbf x}, \boldsymbol{\theta})$ with ${\mathbf x}\in {\cal X}$ and $\boldsymbol{\theta} \in \boldsymbol{\Theta}$ satisfies Assumption~(A3).
Then $E\{\mathbf{F}(\mathbf{x}, \boldsymbol\theta)\}=\int_{\boldsymbol{\Theta}} \mathbf{F}(\mathbf{x}, \boldsymbol\theta) Q(d \boldsymbol{\theta})$ is element-wise continuous with respect to all continuous factors of ${\mathbf x} \in {\cal X}$.
\end{lemma}

\begin{proof}[{\bf Proof of Lemma~\ref{lem:A3_EW}}]\quad
Given any ${\mathbf x}_0 = (x_{01}, \ldots, x_{0d})^T \in {\cal X}$, we let $\{{\mathbf x}_n = (x_{n1}, \ldots,$ $ x_{nd})^T \in {\cal X} \mid n\geq 1\}$ be a sequence of design points, such that, $\lim_{n\rightarrow\infty} {\mathbf x}_n = {\mathbf x}_0$~. If $k<d$, then we must have $(x_{n,k+1}, \ldots, x_{nd}) \equiv (x_{0,k+1}, \ldots,\allowbreak x_{0d})$ for large enough $n$.

For any $s,t \in \{1, \ldots, p\}$, since $F_{st}(\mathbf{x}, \boldsymbol\theta)$ is continuous with respect to $\mathbf{x}_{(1)} = (x_1, \ldots,$ $x_k)^T$, we must have $\lim_{n\rightarrow\infty} F_{st}({\mathbf x}_n, \boldsymbol{\theta}) = F_{st}({\mathbf x}_0, \boldsymbol{\theta})$ for each $\boldsymbol\theta \in \boldsymbol{\Theta}$. Since (A3) is satisfied, according to the Dominated Convergence Theorem (DCT, see, e.g., Theorem~5.3.3 in \cite{resnick2003probability}), we must have 
\begin{eqnarray*}
\lim_{n\rightarrow\infty} E\{F_{st}({\mathbf x}_n, \boldsymbol{\Theta})\} &=& \lim_{n\rightarrow\infty} \int_{\boldsymbol{\Theta}} F_{st}({\mathbf x}_n, \boldsymbol{\theta}) Q(d\boldsymbol{\theta})\\
&=& \int_{\boldsymbol{\Theta}} F_{st}({\mathbf x}_0, \boldsymbol{\theta}) Q(d\boldsymbol{\theta})\>\>\> \mbox{ (by DCT)}\\
&=&  E\{F_{st}({\mathbf x}_0, \boldsymbol{\Theta})\}
\end{eqnarray*}
for each ${\mathbf x}_0 \in {\cal X}$, as long as $\lim_{n\rightarrow\infty} {\mathbf x}_n = {\mathbf x}_0$~. That is, $E\{{\mathbf F}({\mathbf x}, \boldsymbol{\Theta})\}$ is element-wise continuous with respect to ${\mathbf x}_{(1)}$~.
\end{proof}

Since the Fisher information matrices ${\mathbf F}({\mathbf x}, \hat{\boldsymbol{\theta}}_j)$ and ${\mathbf F}({\mathbf x}, \boldsymbol{\theta})$ are symmetric and nonnegative definite for all ${\mathbf x} \in {\cal X}$ and $\hat{\boldsymbol{\theta}}_j, \boldsymbol{\theta}\in \boldsymbol{\Theta}$, 
and both integration and convex combination preserve symmetry and nonnegative definiteness, then ${\mathbf F}(\boldsymbol{\xi}, \boldsymbol{\theta})$, ${\mathbf F}(\boldsymbol{\xi}, \hat{\boldsymbol{\theta}}_j)$, ${\mathbf F}_{\rm SEW}(\boldsymbol{\xi})$, and ${\mathbf F}_{\rm EW}(\boldsymbol{\xi})$ are symmetric and nonnegative definite as well.

For any $\boldsymbol{\theta} \in \boldsymbol{\Theta}$, $\mathbf F(\boldsymbol{\xi}, \boldsymbol{\theta})$ is linear in $\boldsymbol{\xi} \in \boldsymbol{\Xi}({\cal X})$. Actually, for any two designs $\boldsymbol{\xi}_1, \boldsymbol{\xi}_2 \in \boldsymbol{\Xi}({\cal X})$ and any $\lambda \in[0,1]$, it can be verified that
\[
\mathbf F\left(\lambda \boldsymbol{\xi}_1+(1-\lambda) \boldsymbol{\xi}_2, \boldsymbol{\theta} \right)=\lambda \mathbf F\left(\boldsymbol{\xi}_1, \boldsymbol{\theta}\right)+(1-\lambda) \mathbf F\left(\boldsymbol{\xi}_2, \boldsymbol{\theta} \right)\ ,
\]
which further implies 
\begin{eqnarray*}
{\mathbf F}_{\rm SEW}\left(\lambda \boldsymbol{\xi}_1+(1-\lambda) \boldsymbol{\xi}_2 \right) &=& \lambda {\mathbf F}_{\rm SEW}\left(\boldsymbol{\xi}_1\right)+(1-\lambda) {\mathbf F}_{\rm SEW}\left(\boldsymbol{\xi}_2\right)\ ,\\
{\mathbf F}_{\rm EW}\left(\lambda \boldsymbol{\xi}_1+(1-\lambda) \boldsymbol{\xi}_2 \right) &=& \lambda {\mathbf F}_{\rm EW}\left(\boldsymbol{\xi}_1\right)+(1-\lambda) {\mathbf F}_{\rm EW}\left(\boldsymbol{\xi}_2\right)\ .
\end{eqnarray*}
That is, both ${\cal F}_{\rm SEW}({\cal X})$ and ${\cal F}_{\rm EW}({\cal X})$ are convex.

Since ${\cal X}$ is compact under Assumption~(A1), it can be verified that $\boldsymbol{\Xi}({\cal X})$ is also compact under the topology of weak convergence. That is, $\boldsymbol{\xi}_n$ converges weakly to $\boldsymbol{\xi}_0$~, as $n$ goes to $\infty$, if and only if $\lim_{n\rightarrow \infty} \int_{\cal X} f({\mathbf x}) \boldsymbol{\xi}_n(d{\mathbf x}) = \int_{\cal X} f({\mathbf x}) \boldsymbol{\xi}_0(d{\mathbf x})$ for all bounded continuous function $f$ on ${\cal X}$ \citep{billingsley1999convergence, fedorov2014}. Further with Assumption~(A2), $\hat{E}\{{\mathbf F}({\mathbf x}, $ $\boldsymbol{\Theta})\}$ is element-wise continuous with respect to all continuous factors of ${\mathbf x}\in {\cal X}$ due to Lemma~\ref{lem:A2_SEW}, and  must be bounded due to the compactness of ${\cal X}$, then ${\mathbf F}_{\rm SEW}$ as a map from $\boldsymbol{\Xi}({\cal X})$ to $\mathbb{R}^{p\times p}$ is also bounded continuous, and thus its image ${\cal F}_{\rm SEW}({\cal X})$ is also compact. Similarly, with Assumption~(A3) and Lemma~\ref{lem:A3_EW}, ${\mathbf F}_{\rm EW}$ is a bounded continuous map from $\boldsymbol{\Xi}({\cal X})$ to $\mathbb{R}^{p\times p}$, and ${\cal F}_{\rm EW}({\cal X})$ is compact as well.

As a summary of the above arguments, we obtain Lemma~\ref{lem:compactness}.

Recall that we denote ${\mathbf F} (\boldsymbol{\xi}) = \hat{E}\{{\mathbf F}(\boldsymbol{\xi}, \boldsymbol{\Theta})\}$ for sample-based EW D-optimality, and $E\{{\mathbf F}(\boldsymbol{\xi}, \boldsymbol{\Theta})\}$ for integral-based EW D-optimality; ${\mathbf F}_{\mathbf x} = \hat{E}\allowbreak\{{\mathbf F}({\mathbf x}, \boldsymbol{\Theta})\}$ for sample-based EW D-optimality, and $E\{{\mathbf F}({\mathbf x}, \boldsymbol{\Theta})\}$ for integral-based EW D-optimality. Then the same set of assumptions listed in Section~2.4.2 of \cite{fedorov2014} are provided in our notations as below:

\begin{itemize}
    \item[(B1)] $\Psi({\mathbf F})$ is a convex function for ${\mathbf F} \in {\cal F}_{\rm SEW}({\cal X})$ or ${\cal F}_{\rm EW}({\cal X})$.
    \item[(B2)]  $\Psi({\mathbf F})$ is a monotonically non-increasing function for ${\mathbf F} \in {\cal F}_{\rm SEW}({\cal X})$ or ${\cal F}_{\rm EW}({\cal X})$.
    \item[(B3)] There exists a $q>0$, such that $\boldsymbol{\Xi}(q)$ is non-empty.
    \item[(B4)] For any $\boldsymbol{\xi} \in \boldsymbol{\Xi}(q)$, $\bar{\boldsymbol{\xi}} \in \boldsymbol{\Xi}$, and $\alpha \in (0,1)$,
\[
\Psi[(1-\alpha) {\mathbf F}(\boldsymbol{\xi}) + \alpha {\mathbf F}(\bar{\boldsymbol{\xi}})] = \Psi[{\mathbf F}(\boldsymbol{\xi})] + \alpha \int_{\cal X} \psi(\mathbf{x}, \boldsymbol{\xi}) \bar{\boldsymbol{\xi}}(d \mathbf{x}) + o(\alpha)\ ,
\]
as $\alpha \rightarrow 0$, where $\psi({\mathbf x}, \boldsymbol{\xi}) = p - {\rm tr}[{\mathbf F}(\boldsymbol{\xi})^{-1} {\mathbf F}_{\mathbf x}]$. 
\end{itemize}

According to Section~2.4.2 in \cite{fedorov2014}, for D-optimality, Assumptions~(B1), (B2), and (B4) are always satisfied.
Further discussions have been provided in Section~\ref{sec:theoretical_justification}.

\section{Proofs of Main Theorems}\label{sec:Proofs for the theorems} 

\begin{proof}[{\bf Proof of Theorem~\ref{thm:thm_2_1_n0}}]\quad
Under Assumptions~(A1) and (A2), following the proof of the lemma in Section~1.26 of \cite{pukelsheim1993}, it can be verified that $\{{\mathbf F}_{\rm SEW} (\boldsymbol{\xi}) \mid \boldsymbol{\xi} \in \boldsymbol{\Xi}\} \subset \mathbb{R}^{p\times p}$, as the convex hull of the set $\{\hat{E}\{{\mathbf F}(\mathbf{x}, \boldsymbol{\Theta})\} \mid \mathbf{x}
 \in {\cal X}\}$ with dimension $p(p+1)/2$ due to their symmetry, is convex and compact.
Due to Lemma~\ref{lem:compactness} and Carath\'{e}odory's Theorem (see, e.g., Theorem~2.1.1 in \cite{fedorov1972}), it can be further verified that for any matrix ${\mathbf F} \in {\cal F}_{\rm SEW}({\cal X})$ with $\boldsymbol{\xi} \in \boldsymbol{\Xi}({\cal X})$, there exists a $\boldsymbol{\xi}_0 \in \boldsymbol{\Xi}$ with at most $p(p+1)/2 + 1$ design points, such that, ${\mathbf F}_{\rm SEW}(\boldsymbol{\xi}_0) = {\mathbf F}$. 
When ${\mathbf F} \in {\cal F}_{\rm SEW}({\cal X})$ is a boundary point (that is, any open set of ${\cal F}_{\rm SEW}({\cal X})$ containing ${\mathbf F}$ contains both matrices inside and outside ${\cal F}_{\rm SEW}({\cal X})$), the number of design points in $\boldsymbol{\xi}_0$ can be at most $p(p+1)/2$.

Similarly, we can obtain the same conclusions for $\{{\mathbf F}_{\rm EW} (\boldsymbol{\xi}) \mid \boldsymbol{\xi} \in \boldsymbol{\Xi}\}$, $\{E\{{\mathbf F}(\mathbf{x}, \boldsymbol{\Theta})\} \mid \mathbf{x}
 \in {\cal X}\}$, and ${\cal F}_{\rm EW}({\cal X})$ under Assumptions~(A1) and (A3).
\end{proof}

\begin{proof}[{\bf Proof of Theorem~\ref{thm:thm 2.2 page 64}}]\quad 
For the sample-based EW D-optimality, under Assumptions (A1), (A2), and (B3), due to Lemma~\ref{lem:compactness}, ${\cal F}_{\rm SEW}({\cal X})$ is convex and compact. Due to Assumption~(B3), there exists a sample-based EW D-optimal design $\boldsymbol{\xi}_* \in \boldsymbol{\Xi}({\cal X})$. Due to the monotonicity of $\Psi$ (see Assumption~(B2) in Section~\ref{sec:Assumptions_and_proof}), $\hat{E}\{{\mathbf F}(\boldsymbol{\xi}_*, \boldsymbol{\Theta})\}$ must be a boundary point of ${\cal F}_{\rm SEW}({\cal X})$. According to Theorem~\ref{thm:thm_2_1_n0}, there exists a sample-based EW D-optimal design with no more than $p(p+1)/2$ support points.

Similarly, for the integral-based EW D-optimality, under Assumptions (A1), (A3), and (B3), there must exist an EW D-optimal design with no more than $p(p+1)/2$ support points.

According to Section~2.4.2 in \cite{fedorov2014}, Assumptions (B1), (B2), and (B4) are always satisfied for D-optimality. The rest parts of Theorem~\ref{thm:thm 2.2 page 64} are direct conclusions of Theorem~2.2 in \cite{fedorov2014}.
\end{proof}

\begin{proof}[{\bf Proof of Theorem~\ref{thm:MLM design points}}]\quad Suppose all the predictor functions ${\mathbf h}_1, \ldots, {\mathbf h}_{J-1}$ and ${\mathbf h}_c$ defined in \eqref{eq:X_x} are continuous with respect to all continuous factors of $\mathbf{x}\in {\cal X}$. Then there exists an $M_x>0$ due to the compactness of ${\cal X}$, such that, $\|{\mathbf h}_j({\mathbf x})\|\leq M_x$ and $\|{\mathbf h}_c({\mathbf x})\| \leq M_x$~, for all $j=1, \ldots, J-1$ and ${\mathbf x}\in {\cal X}$. 

According to Appendix~A of \cite{huang2024forlion}, if we further have a bounded $\boldsymbol{\Theta}$, then there exists an $M_\eta > 0$, such that, $\|\boldsymbol{\eta}_{\mathbf x}\| = \|{\mathbf X}_{\mathbf x}\boldsymbol{\theta}\| \leq M_\eta$ for all $\boldsymbol{\theta} \in \boldsymbol{\Theta}$ and ${\mathbf x} \in {\cal X}$. Then there exist $0 < \delta_x < \Delta_x <1$, such that, $\delta_x \leq \pi_j^{\mathbf x}(\boldsymbol{\theta}) \leq \Delta_x$ for all $j=1, \ldots, J$, $\boldsymbol{\theta} \in \boldsymbol{\Theta}$, and ${\mathbf x} \in {\cal X}$. Then there exists an $M_u>0$, such that, $|u_{st}^{\mathbf x}(\boldsymbol{\theta})| \leq M_u$ for all $s,t=1, \ldots, J$, $\boldsymbol{\theta} \in \boldsymbol{\Theta}$, and ${\mathbf x} \in {\cal X}$. According to Theorem~2 in \cite{huang2024forlion}, there exists an $M_F>0$, such that, $\|{\mathbf F}_{st}^{\mathbf x}(\boldsymbol{\theta})\| = \|{\mathbf F}_{st}({\mathbf x}, \boldsymbol{\theta})\|\leq M_F$ for all $s,t=1, \ldots, J$, $\boldsymbol{\theta} \in \boldsymbol{\Theta}$, and ${\mathbf x} \in {\cal X}$. As a direct conclusion, the Fisher information of an MLM satisfies Assumption~(A3). The rest statements are direct conclusions of Theorem~\ref{thm:thm 2.2 page 64} and Corollary~\ref{cor:EW_D_optiality}.
\end{proof}

\begin{proof}[{\bf Proof of Theorem~\ref{thm:GLM design points}}]\quad Suppose all the predictor functions ${\mathbf h} = (h_1, \ldots, $ $h_p)^T$ are  continuous with respect to all continuous factors of $\mathbf{x}\in {\cal X}$, and $\boldsymbol{\Theta}$ is bounded. Since the number of level combinations of discrete factors is finite, there exist $M_x>0$ and $M_\eta>0$, such that, $\|{\mathbf h}({\mathbf x}){\mathbf h}({\mathbf x})^T\|\leq M_x$ for all ${\mathbf x}\in {\cal X}$, and $|{\mathbf h}({\mathbf x})^T\boldsymbol{\theta}|\leq M_\eta$ for all ${\mathbf x}\in {\cal X}$ and $\boldsymbol{\theta}\in \boldsymbol{\Theta}$. Since $\nu$ is continuous for commonly used GLMs (see, e.g., Table~5 in the Supplementary Material of \cite{huang2025constrained}), there is an $M_F>0$, such that, $|{\mathbf F}_{st}({\mathbf x}, \boldsymbol{\theta})|\leq M_F$ for all $s,t=1, \ldots, p$, $\boldsymbol{\theta} \in \boldsymbol{\Theta}$, and ${\mathbf x} \in {\cal X}$. As a direct conclusion, the Fisher information of a commonly used GLM satisfies Assumption~(A3). The rest statements are direct conclusions of Theorem~\ref{thm:thm 2.2 page 64} and Corollary~\ref{cor:EW_D_optiality}.
\end{proof}

\begin{proof}[{\bf Proof of Theorem~\ref{thm:sensitivity function eq thm for mlm}}]\quad
According to Theorem~\ref{thm:thm 2.2 page 64}, design $\boldsymbol{\xi}$ is EW D-optimal if and only if 
$
\max _{\mathbf{x} \in \mathcal{X}} d(\mathbf{x}, \boldsymbol{\xi}) \leq p
$. By applying Theorem~2 and its proof (see their Lemma~1) of \cite{huang2024forlion}, we obtain
\[
{\rm tr}\left(\mathbf{E}_{s t}  u_{s t}^{\mathbf x} \mathbf{h}_s^{\mathbf{x}}(\mathbf{h}_t^{\mathbf{x}}\right)^T)  = u_{s t}^{\mathbf x} \operatorname{tr}\left(\left(\mathbf{h}_t^{\mathbf{x}}\right)^T \mathbf{E}_{s t} \mathbf{h}_s^{\mathbf{x}}\right) 
 = u_{s t}^{\mathbf x} \left(\mathbf{h}_t^{\mathbf{x}}\right)^T \mathbf{E}_{s t} \mathbf{h}_s^{\mathbf{x}}\ .
\]
Note that $u_{s t}^{\mathbf x}$ here is either $E\{u_{st}^{\mathbf x} (\boldsymbol{\Theta})\}$ or $\hat{E}\{u_{st}^{\mathbf x} (\boldsymbol{\Theta})\}$, which is different from \cite{huang2024forlion}.
\end{proof}

\begin{proof}[{\bf Proof of Theorem~\ref{thm:sensitivity function eq thm for GLM}}]\quad 
According to Theorem~\ref{thm:thm 2.2 page 64}, a necessary and sufficient condition for a design $\boldsymbol{\xi}$ to be EW D-optimal is $\max _{\mathbf{x} \in \mathcal{X}} d(\mathbf{x}, \boldsymbol{\xi}) \leq p$, which is equivalent to $\max _{\mathbf{x} \in \mathcal{X}} 
 {\rm tr}({\mathbf F}(\boldsymbol{\xi})^{-1} {\mathbf F}_{\mathbf x}) \leq p$.
Then the conclusions are obtained due to ${\mathbf F}(\boldsymbol{\xi}) = \mathbf{X}_{\boldsymbol\xi}^T {\mathbf W}_{{\boldsymbol\xi}} \mathbf{X}_{\boldsymbol\xi}$~.
\end{proof}

\section{Model Selection and Design Comparison for Paper Feeder Experiment}\label{sec:Model_Selection_Paper_Feeder_Experiment}

In this section, we provide more details about the model selection process and the designs obtained for the paper feeder experiment, which was carried out at Fuji-Xerox by Y.~Norio and O.~Akira \citep{joseph2004}. 
In this experiment, there are eight discrete control factors, namely {\tt feed belt material} ($x_1$, Type A or Type B), {\tt speed} ($x_2$, 288 mm/s, 240 mm/s, or 192 mm/s), {\tt drop height} ($x_3$, 3 mm, 2 mm, or 1 mm), {\tt center roll} ($x_4$, Absent or Present), {\tt belt width} ($x_5$, 10 mm, 20 mm, or 30 mm), {\tt tray guidance angle} ($x_6$, 0, 14, or 28), {\tt tip angle} ($x_7$, 0, 3.5, or 7), and {\tt turf} ($x_8$, None, 1 sheet, or 2 sheets), one noise factor ({\tt stack quantity}, High or Low) (see also Table~2 in \cite{joseph2004}), and one continuous control factor, {\tt stack force}, taking values in $[0, 160]$ according to the original experiment. The noise factor is skipped in this analysis since it is not in control of users and is not significant either \citep{joseph2004}.
Following \cite{joseph2004}, we adopt the 18 level settings of eight discrete control factors obtained by coding level 3 as level 1 in column $x_4$ of Table~5 in \cite{joseph2004}.

\subsection{Model selection for  paper feeder experiment}\label{sec:Model_Selection_subsection}

For the paper feeder experiment, the summarized response variables include the number of misfeeds $y_{i1}$~, the number of precise feeding of the paper $y_{i2}$~, and the number of multifeeds $y_{i3}$ at the experimental setting ${\mathbf x}_i$~. It is important to notice that these two types of failures, misfeed or multifeed, cannot occur simultaneously. Therefore, a multinomial logistic model is more appropriate than two separate generalized linear models for analyzing this experiment. For those three-level factors, namely  $x_2$, $x_3$, $x_5$, $x_6$, $x_7$, and $x_8$~, which are all quantitative factors, we follow \cite{wu2009} and \cite{joseph2004} to convert the two degrees of freedom for each factor into linear and quadratic components with contrasts $(-1,0,1)$ and $(1,-2,1)$, respectively. Following \cite{joseph2004}, we code each of the two-level factors $x_1$ and $x_4$~, which are qualitative factors, as $(-1,1)$. As for the continuous factor stack force $M$, we transform it to $\log(M+1)$, which is slightly different from $\log M$ in \cite{joseph2004}, to deal with a few cases with $M=0$.

\begin{table}[htbp]
    \centering
    \caption{Model comparison for paper feeder experiment}
    {
    \renewcommand{\arraystretch}{0.5}
        \resizebox{\textwidth}{!}{
    \begin{tabular}{cccccccccccc}
    \toprule
    \multicolumn{1}{c}{ } 
      & \multicolumn{2}{c}{\textbf{Cumulative}}&
      & \multicolumn{2}{c}{\textbf{Continuation}}&
      & \multicolumn{2}{c}{\textbf{Adjacent}}&
      & \multicolumn{2}{c}{\textbf{Baseline}} \\
    \cmidrule(lr){2-3} \cmidrule(lr){5-6} \cmidrule(lr){8-9} \cmidrule(lr){11-12} 
    & po & npo &
    & po & npo&
    & po & npo &
    & po & npo \\
    \midrule
 AIC& 1011.06 & \textbf{937.87}&  &  1037.60 &  970.59&  &1027.72&  962.99&  &1661.38 & 962.99\\
 BIC &  1104.34& \textbf{1113.46}&  & 1130.88  &  1146.18&  & 1121.00 & 1138.58 &  & 1754.66 &  1138.58 \\
    \bottomrule
    \end{tabular}}}
    \label{tab:Model comparison for the Paper Feeder Experiment}
\end{table}

We first use AIC and BIC criteria  to choose the most appropriate multinomial logistic model from eight different candidates \citep{bu2020}. According to Table~\ref{tab:Model comparison for the Paper Feeder Experiment}, we adopt the cumulative logit model with non-proportional odds (npo), which achieves the smallest AIC and BIC values (highlighted in bold font), and is described as follows:

\begin{eqnarray*}
    \log\frac{\pi_{i1}}{\pi_{i2}+\pi_{i3}} &=& \beta_{11}+\beta_{12}\log(M_i+1)+\beta_{13}x_{i1}+\beta_{14}x_{i2l}+\beta_{15}x_{i2q}\\
    &+& \beta_{16}x_{i3l} + \beta_{17}x_{i3q} + \beta_{18}x_{i4} + \beta_{19}x_{i5l}+\beta_{110}x_{i5q}+\beta_{111}x_{i6l}\\
    &+& \beta_{112}x_{i6q} + \beta_{113}x_{i7l}+\beta_{114}x_{i7q} + \beta_{115}x_{i8l}+\beta_{116}x_{i8q}\ ,\\
   \log\frac{\pi_{i1}+\pi_{i2}}{\pi_{i3}} &=& \beta_{21}+\beta_{22}\log(M_i+1)+\beta_{23}x_{i1}+\beta_{24}x_{i2l}+\beta_{25}x_{i2q}\\
   &+& \beta_{26}x_{i3l}+\beta_{27}x_{i3q} + \beta_{28}x_{i4} + \beta_{29}x_{i5l}+\beta_{210}x_{i5q}+\beta_{211}x_{i6l}\\
   &+& \beta_{212}x_{i6q}+\beta_{213}x_{i7l}+\beta_{214}x_{i7q} + \beta_{215}x_{i8l}+\beta_{216}x_{i8q}\ ,
\end{eqnarray*} 
where $i=1, \ldots, m$ with $m=183$ for the original experimental design.

\subsection{Locally D-optimal designs for paper feeder experiment}\label{sec:locally_Dopt_Paper_Feeder_Experiment}

Using the data listed in Table~3 of \cite{joseph2004}, we fit the chosen model, a cumulative logit model with non-proportional odds and $p=32$ parameters. The fitted parameter values are
\begin{eqnarray*}
\hat{\boldsymbol{\theta}} &=& (\hat{\beta}_{11},\hat{\beta}_{12},\hat{\beta}_{13},\hat{\beta}_{14},\hat{\beta}_{15},\hat{\beta}_{16},\hat{\beta}_{17},\hat{\beta}_{18},\\
& & \hat{\beta}_{19},\hat{\beta}_{110},\hat{\beta}_{111},\hat{\beta}_{112},\hat{\beta}_{113},\hat{\beta}_{114},\hat{\beta}_{115},\hat{\beta}_{116},\\
   & & \hat{\beta}_{21},\hat{\beta}_{22},\hat{\beta}_{23},\hat{\beta}_{24},\hat{\beta}_{25},\hat{\beta}_{26},\hat{\beta}_{27},\hat{\beta}_{28},\\
   & & \hat{\beta}_{29},\hat{\beta}_{210},\hat{\beta}_{211},\hat{\beta}_{212},\hat{\beta}_{213},\hat{\beta}_{214},\hat{\beta}_{215},\hat{\beta}_{216})^T\\
&=&(7.995,-3.268,-1.275,1.531,0.044,-0.156,-0.141,0.534,\\
& & -0.261, 0.418, -1.749,-0.084,-0.207,0.759,0.782,0.356,\\
& & 10.928,-2.461, -0.409, 0.711, 0.080,-0.144, -0.120,0.196,\\
& & 0.019,-0.023,-0.931, -0.012, -0.133,0.153, 0.128,0.125)^T\ .
\end{eqnarray*}

Assuming $\hat{\boldsymbol{\theta}}$ to be the true values of the model parameters, we first apply the ForLion algorithm proposed by \cite{huang2024forlion} and obtain a locally D-optimal approximate design, which consists of $41$ design points and is listed as ``ForLion'' in both Table~\ref{tab:Locally D-optimal Designs for the Paper Feeder Experiment efficiencies} and Table~\ref{tab:EW Forlion optimal designs}. 
To convert this approximate design with a continuous factor to an exact design, we apply our rounding algorithm (Algorithm~\ref{alg:appro to exact}) to it with merging threshold $\delta_2=3.2$, grid level $L=2.5$, and the number of experimental units $n=1,785$ (the same sample size as in the original experiment). The obtained exact design is listed as ``ForLion exact grid2.5'' in Tables~\ref{tab:Locally D-optimal Designs for the Paper Feeder Experiment efficiencies} and \ref{tab:EW Forlion optimal designs}. For comparison purpose, we also apply the lift-one and exchange algorithms of \cite{bu2020} to two different sets of design points: one set is the same as the original experiment, and the other is based on the grid points of stack force with pace length 2.5 (that is, $\{0, 2.5, 5, \ldots, 157.5, 160\}$) for each of the 18 experimental settings (runs) of the eight discrete factors. The designs obtained correspondingly are presented as ``Bu-appro'' (lift-one algorithm on original design points) and ``Bu-exact'' (exchange algorithm on original design points) in Tables~\ref{tab:Locally D-optimal Designs for the Paper Feeder Experiment efficiencies} and \ref{tab:D-Optimal designs for the paper feeder experiment}, and ``Bu-grid2.5'' (lift-one algorithm on grid-2.5 points) and ``Bu-grid2.5 exact'' (exchange algorithm on grid-2.5 points) in Tables~\ref{tab:Locally D-optimal Designs for the Paper Feeder Experiment efficiencies} and \ref{tab:EW Forlion optimal designs}. 
According to the summarized information of those designs listed in Table~\ref{tab:Locally D-optimal Designs for the Paper Feeder Experiment efficiencies}, the ForLion design and its exact designs obtained by Algorithm~\ref{alg:appro to exact} with grid level $L=0.1, 0.5, 1, 2.5$, respectively, achieve the highest relative efficiencies with respect to the ForLion design itself, as well as the fewest design points. On the contrary, the lift-one and exchange algorithms of \cite{bu2020} are much more time-consuming even with grid level 2.5, and the original design and Bu's designs (Bu-appro and Bu-exact) on the original set of design points are not satisfactory in terms of relative efficiency.

\begin{table}[htbp]
    \centering
    \caption{Locally D-optimal designs for paper feeder experiment}
     {
    \renewcommand{\arraystretch}{0.5}
        \resizebox{\textwidth}{!}{
    \begin{threeparttable}
          \begin{tabular}{cccccccccccc}
    \toprule
    Designs &  & $m$ &  &  & Time (s) &  &  & $|{\mathbf F}(\boldsymbol{\xi}, \boldsymbol \theta)|$  &  &  & Relative Efficiency\\
    \midrule
    Original allocation&  &183&  &  & - &  &  &1.164e+28 &  &  &62.792\% \\
     Bu-appro&  &41&  &  & 62.58s&  &  & 5.930e+32 &  &  &88.106\%  \\
      
    Bu-exact&  & 41 &  &  & 1598s  &  &  &5.928e+32 &  &  &88.105\% \\
    Bu-grid2.5 appro &  &50 &  &  &  81393s &  &  &  2.397e+34 &  &  &98.903\% \\
    Bu-grid2.5 exact &  & 48 &   &  &  85443s &  &  &  2.396e+34 &  &  &98.902\% \\
     ForLion&  &41 &  &  &  22099s &  &  &  3.411e+34 &  &  &100.000\%  \\
    ForLion exact grid-0.1&  &36 &  &  &  0.01s &  &  &  3.412e+34 &  &  &100.001\%  \\
   ForLion exact grid-0.5&  &36 &  &  &  0.01s &  &  &  3.374e+34&  &  &99.965\%  \\
   ForLion exact grid-1&  &36 &  &  &  0.01s &  &  &  2.556e+34 &  &  &99.102\%  \\
   ForLion exact grid-2.5&  &36 &  &  &  0.01s &  &  &  2.245e+34 &  &  & 98.701\%  \\
    \bottomrule
    \end{tabular}
     \begin{tablenotes}
         \footnotesize
          \setlength{\baselineskip}{8pt} %
         \item Note: $m$ is the number of design points; Relative Efficiency is $(|{\mathbf F}(\boldsymbol{\xi}, \hat{\boldsymbol\theta})|/|{\mathbf F}(\boldsymbol{\xi}_{\rm ForLion}, \hat{\boldsymbol \theta})|)^{1/p}$ with  $p=32$.
     \end{tablenotes}
    \end{threeparttable}
      }
      }
    \label{tab:Locally D-optimal Designs for the Paper Feeder Experiment efficiencies}
\end{table}

The ForLion exact designs (with grid-0.1, grid-0.5, grid-1, and grid-2.5) in Table~\ref{tab:Locally D-optimal Designs for the Paper Feeder Experiment efficiencies} also show the convenience by adopting our rounding algorithm (Algorithm~\ref{alg:appro to exact}). Once an optimal approximate design is obtained, the exact design with a user-specified grid level can be obtained instantly with high relative efficiency. 
To show how we choose the merging threshold $\delta_2=3.2$ for this experiment, in Figure~\ref{fig:mergedesign} we show how the relative efficiency (left panel) and the number of design points (right panel) change along with the relative distance (i.e., the merging threshold $\delta_2$), where the grid levels 0.1, 0.5, 1 and 2.5, and the relative distances between 2 and 6 with increment of 0.1 are used for illustration. 
Notably, the next change on the number of design points after $\delta_2=3.2$ will not occur until $\delta_2=31.29$, which is too large. 
As a result, we recommend the merging threshold $\delta_2=3.2$ for this experiment. The choice of grid level typically depends control accuracy of the factors and the options that the experimenter may have. As illustrated in Table~\ref{tab:Locally D-optimal Designs for the Paper Feeder Experiment efficiencies}, the experimenter may choose the smallest possible grid level when applicable.

\begin{figure}[ht]
    \centering
\includegraphics[width=\textwidth,height=9cm]{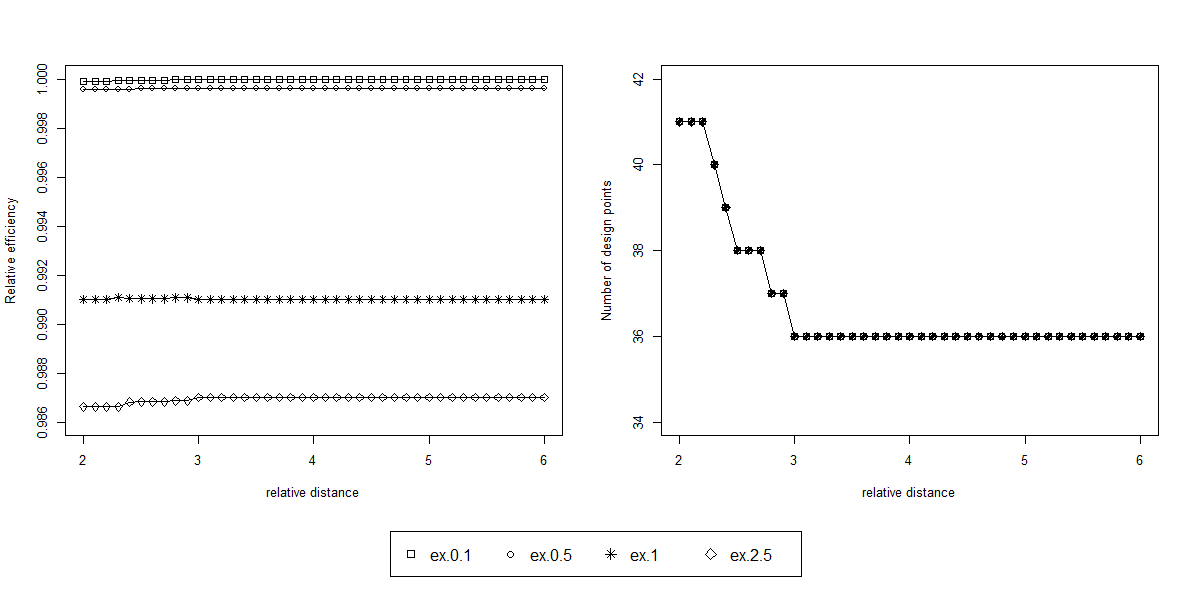} \\
\caption{Relative efficiency and number of design points of exact designs against relative distance (or merging threshold) for paper feeder experiment}
    \label{fig:mergedesign}
\end{figure}

\subsection{More graphs and tables for paper feeder experiment}\label{sec:EW_Dopt_Paper_Feeder_Experiment}

In this section, we provide {\it (i)} Figure~\ref{fig:PFE_simulation_results} for boxplots of relative efficiencies of different robust designs with respect to locally D-optimal designs, which is a graphical display of Table~\ref{tab:Efficiencies of designs for paper feeder experiment}; {\it (ii)} Table~\ref{tab:D-Optimal designs for the paper feeder experiment} for listing designs constructed on the set of design points of the original experiment, including ``Original'' (design), ``Bu-exact'' by \cite{bu2020}'s exchange algorithm, ``Bu-appro'' by \cite{bu2020}'s lift-one algorithm, ``EW Bu-exact'' by \cite{bu2020}'s EW exchange algorithm, and ``EW Bu-appro'' by \cite{bu2020}'s EW lift-one algorithm; and {\it (iii)} Table~\ref{tab:EW Forlion optimal designs} for listing designs constructed on grid points of $[0,160]$ for stack force.

\begin{figure}[ht]
    \centering
    \includegraphics[width=0.8\textwidth,height=10cm]{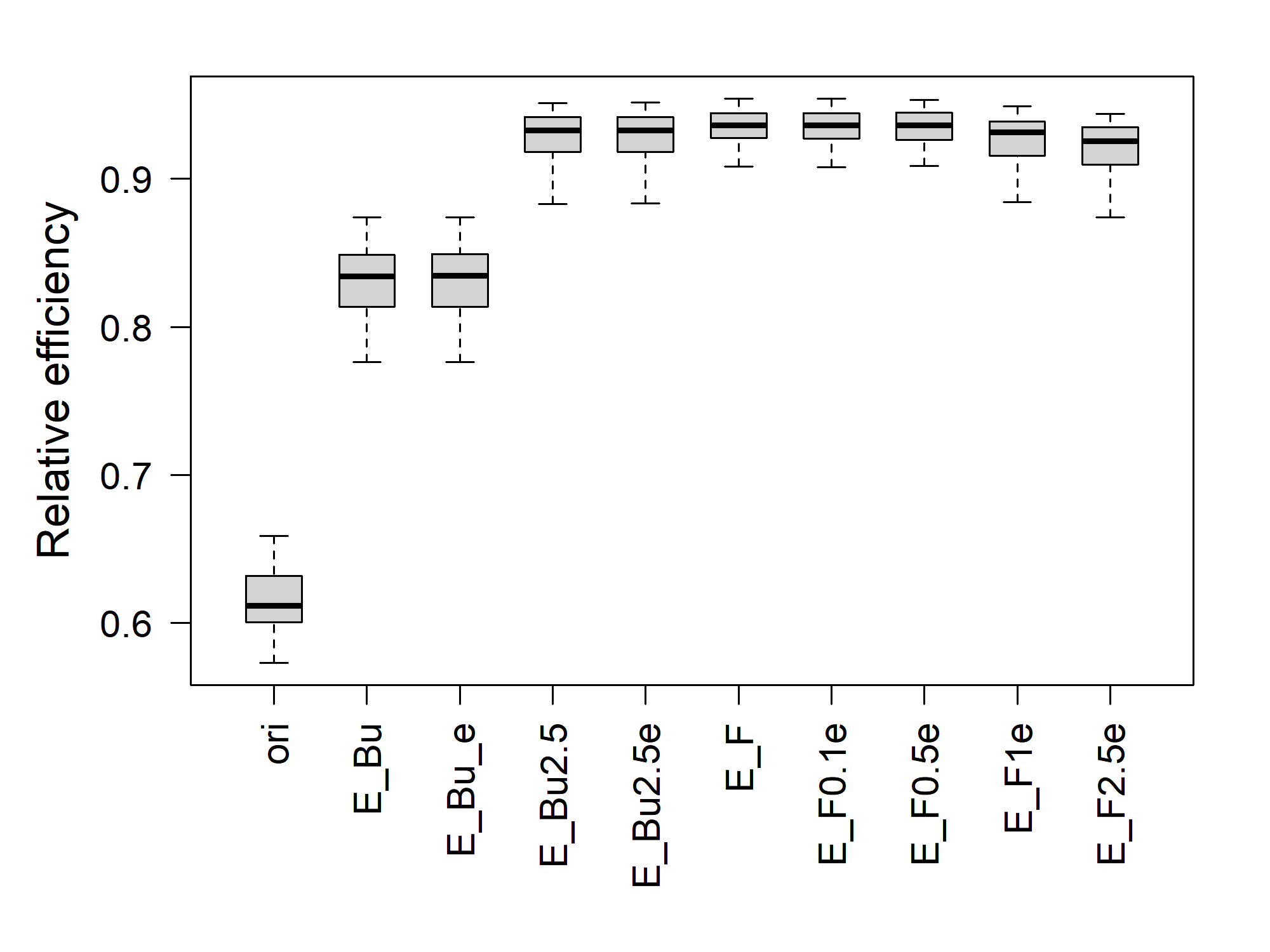}\\[-4ex]
    \caption{Boxplots of relative efficiencies of designs in Table~\ref{tab:Efficiencies of designs for paper feeder experiment} with respect to locally D-optimal designs under 100 different sets of parameter values for paper feeder experiment}
    \label{fig:PFE_simulation_results}
\end{figure}

\begin{landscape}
\footnotesize  
\setlength{\tabcolsep}{0pt}
\setlength\LTleft{0pt}            
\setlength\LTright{0pt}           
\tiny
{
  \renewcommand{\arraystretch}{0.4}%
\begin{longtable}[ht]{@{\extracolsep{\fill}}cccccccccccccccccccccccccccccc@{}}
 \captionsetup{skip=2pt} 
\caption{Designs constructed on the original set of design points for  paper feeder experiment}\label{tab:D-Optimal designs for the paper feeder experiment}\\

\toprule
\multirow{2}*{\textbf{Run}} & \multirow{2}*{\textbf{Design}} &\multicolumn{28}{c}{\textbf{Stack Force}}\\
    \cline{3-30}
    & & 0 & 5 & 10 & 15 & 20 & 25 & 30 & 35 & 40 & 42.5 & 45 & 50 & 55 & 60 & 62.5 & 65 & 70 & 75 & 80 & 82.5 & 85 & 90 & 95 & 100 & 105 & 110 & 120 & 160\\
\midrule
\endfirsthead

\multicolumn{30}{r}{Continued}\\
\toprule
\multirow{2}*{\textbf{Run}} & \multirow{2}*{\textbf{Design}} &\multicolumn{28}{c}{\textbf{Stack Force}}\\
    \cline{3-30}
    & & 0 & 5 & 10 & 15 & 20 & 25 & 30 & 35 & 40 & 42.5 & 45 & 50 & 55 & 60 & 62.5 & 65 & 70 & 75 & 80 & 82.5 & 85 & 90 & 95 & 100 & 105 & 110 & 120 & 160\\
\midrule
\endhead
%
\multicolumn{30}{c}{Continued on next page}\\
\endfoot
\bottomrule
\endlastfoot
    1 & Original & & & & & 10 &  & 5 &  &  10 & 5 & 5 & 10 &  & 15 & 5 & 5 & 10 & 
 & 10 & 5 & 5 & 10 & & & & & 5 & 5\\
  & Bu-exact & & & & & 59 &  & 0 &  &  0 & 0 & 0 & 0 &  & 0 & 0 & 0 & 0 & 
 & 0 & 0 & 0 & 0 & & & & & 50 & 0\\
  & Bu-appro & & & & & 0.033 &  & 0 &  &  0 & 0 & 0 & 0 &  & 0 & 0 & 0 & 0 & 
 & 0 & 0 & 0 & 0 & & & & & 0.028 & 0\\
& EW Bu-exact & & & & & 57 &  & 0 &  &  0 & 0 & 0 & 0 &  & 0 & 0 & 0 & 0 & 
 & 0 & 0 & 0 & 0 & & & & & 48 & 0\\
& EW Bu-appro & & & & & 0.032 &  & 0 &  &  0 & 0 & 0 & 0 &  & 0 & 0 & 0 & 0 & 
 & 0 & 0 & 0 & 0 & & & & & 0.027 & 0\\

 \midrule
 2 & Original & 10 &  & 10 & 10 & 10 &  & 15 & 5 & 20 &  &  &5 &  & 15 &  &  & 5 & 5 & 5 &  &  &  & &  &  & & &  \\
  & Bu-exact & 0 &  & 54 & 0 & 0 &  & 0 & 0 & 0 &  &  &0 &  & 0 &  &  & 0 & 0 & 50 &  &  &  & &  &  & & &  \\
  & Bu-appro & 0 &  &  0.030 & 0 & 0 &  & 0 & 0 & 0 &  &  &0 &  & 0 &  &  & 0 & 0 & 0.028 &  &  &  & &  &  & & &  \\
  & EW Bu-exact & 0 &  &  53 & 0 & 0 &  & 0 & 0 & 0 &  &  &0 &  & 0 &  &  & 0 & 0 & 50 &  &  &  & &  &  & & &  \\
  & EW Bu-appro & 0 &  &  0.029 & 0 & 0 &  & 0 & 0 & 0 &  &  &0 &  & 0 &  &  & 0 & 0 & 0.028 &  &  &  & &  &  & & &  \\
 \midrule

 3 & Original & 15 &  & 10 & 10 & 20 & 10 & 15 & 5 & 15& & & & & & & & & & & & & & & & & & &  \\
 & Bu-exact & 0 &  & 57 & 0 & 0 & 0 & 0 & 0 & 44& & & & & & & & & & & & & & & & & & &  \\
 & Bu-appro & 0 &  & 0.032 & 0 & 0 & 0 & 0 & 0 & 0.025& & & & & & & & & & & & & & & & & & & \\ 
 & EW Bu-exact & 0 &  & 56 & 0 & 0 & 0 & 0 & 0 & 43& & & & & & & & & & & & & & & & & & &  \\
  & EW Bu-appro & 0 &  & 0.031 & 0 & 0 & 0 & 0 & 0 & 0.024& & & & & & & & & & & & & & & & & & &  \\
 \midrule

 4 & Original & 5 &  &  &  & 10 & 10 & 10 & &  15& & & 10& 5&15 & &5 &5 & &5 & & & & & & & & &  \\
 & Bu-exact & 0 &  &  &  & 53 & 0 & 0 & &  0& & & 0& 0&0 & &0 &0 & &48 & & & & & & & & &  \\
 & Bu-appro & 0 &  &  &  & 0.030 & 0 & 0 & &  0& & & 0& 0&0 & &0 &0 & &0.027 & & & & & & & & &  \\
 & EW Bu-exact & 42 &  &  &  & 49 & 0 & 0 & &  0& & & 0& 0&0 & &0 &0 & &47 & & & & & & & & &  \\
  & EW Bu-appro & 0.024 &  &  &  & 0.028 & 0 & 0 & &  0& & & 0& 0&0 & &0 &0 & &0.026 & & & & & & & & &  \\
 \midrule
 
 5 & Original & & & & & 10&10 &15 & &20 & &5 &20 &5 & 10& & & & & & & & & & & & & &  \\
 & Bu-exact & & & & & 58&0 &0 & &0 & &0 &0 &0 & 48& & & & & & & & & & & & & &  \\
 & Bu-appro & & & & & 0.033&0 &0 & &0 & &0 &0 &0 & 0.027& & & & & & & & & & & & & &  \\
 & EW Bu-exact & & & & & 58&0 &0 & &0 & &0 &0 &0 & 47& & & & & & & & & & & & & &  \\
  & EW Bu-appro & & & & & 0.032&0 &0 & &0 & &0 &0 &0 & 0.026& & & & & & & & & & & & & &  \\
 \midrule
 
6 & Original & & & &10 &10 & &15 & &20 & & 5& 10&5 & 5& & & & & & & & & & & & & &  \\
 & Bu-exact & & & &55 &0 & &0 & &0 & & 0& 0&0 & 44& & & & & & & & & & & & & &  \\
 & Bu-appro & & & &0.031 &0 & &0 & &0 & & 0& 0&0& 0.025& & & & & & & & & & & & & &  \\
 & EW Bu-exact & & & &53 &0 & &0 & &0 & & 0& 0&0& 43& & & & & & & & & & & & & &  \\
  & EW Bu-appro & & & &0.030 &0 & &0 & &0 & & 0& 0&0& 0.024& & & & & & & & & & & & & &  \\
 \midrule
 
7 & Original & & & 10& &20 &5 &20 &10 &20 & & & 5& & 5& & & 5& & 5& & & & & & & & &  \\
& Bu-exact & & & 0& &59 &0 &0 &0 &0 & & & 0& & 0& & & 0& & 52& & & & & & & & &  \\
& Bu-appro & & & 0& &0.033 &0 &0 &0 &0 & & & 0& & 0& & & 0& & 0.029& & & & & & & & &  \\
& EW Bu-exact & & & 0& &57 &0 &0 &0 &0 & & & 0& & 0& & & 0& & 50& & & & & & & & &  \\
& EW Bu-appro & & & 0& &0.032 &0 &0 &0 &0 & & & 0& & 0& & & 0& & 0.028& & & & & & & & &  \\
 \midrule
 
8 & Original & & & & 5& 10& & 10& 10& 10& & & & & 10& & & 15&5 &10 & & & & & 10& &5 &5 &  \\
& Bu-exact & & & & 0& 0& & 14& 39& 0& & & & & 0& & & 0&0 &0 & & & & & 0& &0 &46 &  \\
& Bu-appro & & & & 0& 0& & 0.008& 0.022& 0& & & & & 0& & & 0&0 &0 & & & & & 0& &0 &0.026 &  \\
& EW Bu-exact & & & & 0& 0& & 0& 52& 0& & & & & 0& & & 0&0 &0 & & & & & 0& &0 &43 &  \\
& EW Bu-appro & & & & 0& 0& & 0& 0.029& 0& & & & & 0& & & 0&0 &0 & & & & & 0& &0 &0.024 &  \\
 \midrule

9 & Original & & &10 &10 &10 &5 &10 & & 20& & & 5& 5& 10& &5 &10 & &5 & & & & & & & & &  \\
& Bu-exact & & &0 &0 &2 &53 &0 & & 0& & & 0& 0& 0& &0 &0 & &48 & & & & & & & & &  \\
& Bu-appro & & &0 &0 &0.001 &0.029 &0 & & 0& & & 0& 0& 0& &0 &0 & &0.027 & & & & & & & & &  \\ 
& EW Bu-exact & & &0 &0 &0 &53 &0 & & 0& & & 0& 0& 0& &0 &0 & &46 & & & & & & & & &  \\ 
& EW Bu-appro & & &0 &0 &0 &0.030 &0 & & 0& & & 0& 0& 0& &0 &0 & &0.026 & & & & & & & & &  \\ 
 \midrule
 
10 & Original &15 &20 &15 &20 &20 & &10 & & & & & & & & & & & & & & & & & & & & &  \\
& Bu-exact &0 &21 &0 &0 &0 & &44 & & & & & & & & & & & & & & & & & & & & &  \\
& Bu-appro &0 &0.012 &0 &0 &0 & &0.025 & & & & & & & & & & & & & & & & & & & & &  \\
&EW Bu-exact &0 &22 &0 &0 &0 & &44 & & & & & & & & & & & & & & & & & & & & &  \\
&EW Bu-appro &0 &0.013 &0 &0 &0 & &0.025 & & & & & & & & & & & & & & & & & & & & &  \\
 \midrule
 
11 & Original & 15 &20 &20 &20 &20 & &10 & & & & & & & & & & & & & & & & & & & & &  \\
& Bu-exact & 0 &32 &0 &0 &0 & &44 & & & & & & & & & & & & & & & & & & & & &  \\
& Bu-appro & 0 &0.018 &0 &0 &0 & &0.024 & & & & & & & & & & & & & & & & & & & & &  \\
& EW Bu-exact & 0 &30 &0 &0 &0 & &44 & & & & & & & & & & & & & & & & & & & & &  \\
& EW Bu-appro & 0 &0.017 &0 &0 &0 & &0.024 & & & & & & & & & & & & & & & & & & & & &  \\
 \midrule

12 & Original &10 & & 10& 10&10 & &5 & &10 & &5 &10 &5 &5 & & & & & & & & & & & & & &  \\
& Bu-exact &0 & & 26& 33&0 & &0 & &0 & &0 &0 &0 &51 & & & & & & & & & & & & & &  \\
& Bu-appro &0 & & 0.015& 0.018&0 & &0 & &0 & &0 &0 &0 &0.029 & & & & & & & & & & & & & &  \\
& EW Bu-exact &0 & & 27& 30&0 & &0 & &0 & &0 &0 &0 &50 & & & & & & & & & & & & & &  \\
& EW Bu-appro &0 & & 0.015& 0.017&0 & &0 & &0 & &0 &0 &0 &0.028 & & & & & & & & & & & & & &  \\
 \midrule
 
13 & Original &10 & & 10& 10&10 & &5 & &5 & & & &5 &5 & & & & &5 & &5 &5 & & 5& & & &  \\
& Bu-exact &0 & & 0& 58&0 & &0 & &0 & & & &0 &0 & & & & &0 & &0 &0 & & 50& & & &  \\
& Bu-appro &0 & & 0& 0.032 &0 & &0 & &0 & & & &0 &0 & & & & &0 & &0 &0 & & 0.028& & & &  \\
& EW Bu-exact &0 & & 0& 56&0 & &0 & &0 & & & &0 &0 & & & & &0 & &0 &0 & & 50& & & &  \\
& EW Bu-appro &0 & & 0& 0.031 &0 & &0 & &0 & & & &0 &0 & & & & &0 & &0 &0 & & 0.028 & & & &  \\
 \midrule
 
14 & Original & & &10 & &20 &15 &20 &20 &20 & &5 &5 & & & & & & & & & & & & & & & &  \\
 & Bu-exact & & &3 & &52 &0 &0 &0 &0 & &0 &43 & & & & & & & & & & & & & & & &  \\
  & Bu-appro & & &0.002 & &0.029 &0 &0 &0 &0 & &0 &0.024 & & & & & & & & & & & & & & & &  \\
& EW Bu-exact & & &0 & &53 &0 &0 &0 &0 & &0 &42 & & & & & & & & & & & & & & & &  \\
& EW Bu-appro & & &0 & &0.029 &0 &0 &0 &0 & &0 &0.024 & & & & & & & & & & & & & & & &  \\
 \midrule
 
15 & Original &10 &10 &15 &15 &15 & &10 &5 &5 & & &5 & & & & & & & & & & & & & & & &  \\
& Bu-exact &47 &0 &0 &0 &0 & &0 &46 &0 & & &0 & & & & & & & & & & & & & & & &  \\
& Bu-appro &0.026 &0 &0 &0 &0 & &0  & 0.026 &0 & & &0 & & & & & & & & & & & & & & & &  \\
& EW Bu-exact &46 &0 &0 &0 &0 & &0  & 46 &0 & & &0 & & & & & & & & & & & & & & & &  \\
& EW Bu-appro &0.026 &0 &0 &0 &0 & &0  & 0.026 &0 & & &0 & & & & & & & & & & & & & & & &  \\
 \midrule
 
16 & Original & & 10& 10& &15 & & 20& 10& 15& & & 10& 5& 10& & & & & & & & & & & & & &  \\
 & Bu-exact & & 49& 0& &0 & & 0& 0& 0& & & 0& 0& 45& & & & & & & & & & & & & &  \\
 & Bu-appro & & 0.028 & 0& &0 & & 0& 0& 0& & & 0& 0& 0.025& & & & & & & & & & & & & &  \\
  & EW Bu-exact & & 49 & 0& &0 & & 0& 0& 0& & & 0& 0& 45& & & & & & & & & & & & & &  \\
   & EW Bu-appro & & 0.027 & 0& &0 & & 0& 0& 0& & & 0& 0& 0.025& & & & & & & & & & & & & &  \\
 \midrule

17 & Original & & &10 & &10 &5 &10 & &10 & &5 &5 & & & & & & & 5& & & 5&5 &10 &5 &5 & &  \\
& Bu-exact & & &0 & &51 &0 &0 & &0 & &0 &0 & & & & & & & 0& & & 0&0 &0 &0 &47 & &  \\
& Bu-appro & & &0 & &0.029 &0 &0 & &0 & &0 &0 & & & & & & & 0& & & 0&0 &0 &0 &0.027 &  \\
& EW Bu-exact & & &0 & &51 &0 &0 & &0 & &0 &0 & & & & & & & 0& & & 0&0 &0 &0 &46 &  \\
& EW Bu-appro & & &0 & &0.028 &0 &0 & &0 & &0 &0 & & & & & & & 0& & & 0&0 &0 &0 &0.026 &  \\
 \midrule
 
18 & Original & & &10 &5 &10 & &10 &5 &5 & & & & &10 & &5 &10 &5 &10 & & & 10& & & & &5 &  \\
& Bu-exact & & &0 &6 &53 & &0 &0 &0 & & & & &0 & &0 &0 &0 &0 & & & 51& & & & &0 &  \\
& Bu-appro & & &0 &0.003 &0.030 & &0 &0 &0 & & & & &0 & &0 &0 &0 &0 & & & 0.029& & & & &0 &  \\
& EW Bu-exact & & &0 &0 &57 & &0 &0 &0 & & & & &0 & &0 &0 &0 &0 & & & 50& & & & &0 &  \\
& EW Bu-appro & & &0 &0 &0.032 & &0 &0 &0 & & & & &0 & &0 &0 &0 &0 & & & 0.028& & & & &0 &  \\
\end{longtable}
}
\normalsize
\end{landscape}

\tiny  
\setlength{\tabcolsep}{3pt}
\setlength\LTleft{0pt}            
\setlength\LTright{0pt}           
{
  \renewcommand{\arraystretch}{0.4}%
\begin{longtable}[ht]{@{\extracolsep{\fill}}cc|cccccccc|ccccccc@{}}
\caption{Designs on $[0,160]$ of stack force for paper feeder experiment}\label{tab:EW Forlion optimal designs}\\

\toprule
      \textbf{Run} &Design &\multicolumn{8}{c|}{Stack Force}&\multicolumn{7}{c}{$w_i$ or $n_i$}\\ 
\midrule
\endfirsthead

\multicolumn{16}{r}{Continued}\\ \\
\toprule
      \textbf{Run} &Design &\multicolumn{8}{c|}{Stack Force} &\multicolumn{7}{c}{$w_i$ or $n_i$}\\ 
\midrule
\endhead
 
\multicolumn{16}{l}{Continued on next page}\\ 
\endfoot
\bottomrule
\endlastfoot
 
   1 & Bu-grid2.5 &20.0& &112.5&  & - &  & -  &  & 0.0305&  &0.0264& &- &  &-\\
    & Bu-grid2.5 exact &20.0& &112.5&  & - &  & -  & & 55&  &47& &- &  &- \\
   &ForLion &21.839 & &110.554&  & -  &  & - &   & 0.0292 &  &0.0270&  & - &  &-  \\
    &ForLion exact grid2.5 & 22.5& &110.0&  & -  &  & -  &  & 52 &  &48&  &- &  &-\\

     &EW Bu-grid2.5 &20.0& &22.5&  & 110  &  &112.5 & &  0.0173 &  &0.0124&  &0.0065 &  &0.0189\\
    &EW Bu-grid2.5 exact &20.0& &22.5&  & 112.5  &  &-  & &30&  &23&  &45 &  &-\\
    &EW ForLion & 22.127& &111.339&  & -  &  & -  &  & 0.0288&  &0.0255&  &- &  &- \\
    &EW ForLion exact grid2.5&22.5& &112.5&  & -  &  & - &  & 51 &  &45&  &- &  &-\\

    \midrule
    
    2 & Bu-grid2.5 &5.0& &75.0&  & - &  & - & & 0.0291&  &0.0279& &- &  &- \\
    & Bu-grid2.5 exact &5.0& &75.0&  & - &  & - & & 52&  &50& &- &  &- \\
    &ForLion  & 4.749& &74.807&  & -  &  &-& & 0.0286&  &0.0280&  &- &  &- \\
     &ForLion exact  grid2.5&5.0  & &75.0&  & -  &  & -& &51 &  &50&  &- &  &-\\

       &EW Bu-grid2.5 &5.0& &75.0&  & 77.5 &  &- & & 0.0284 &  &0.0169&  &0.0110 &  &-\\
    &EW Bu-grid2.5 exact&5.0& &75.0&  & 77.5  &  &- & &  51&  &37&  &13 &  &- \\ 
    
     &EW ForLion &4.505& &75.677&  & -  &  &  - & & 0.0281 &  &0.0280&  &- &  &- \\
    &EW ForLion exact grid2.5&5.0& &75.0&  & -  &  & -  & & 50 &  &50&  &- &  &-\\
   
    \midrule
    
    3 & Bu-grid2.5 &7.5& &47.5&  & - &  & - & & 0.0284&  &0.0257& &- &  &- \\
    & Bu-grid2.5 exact &7.5& &47.5&  & - &  & - & & 51&  &46& &- &  &- \\
    
    &ForLion  & 7.669& &46.322&  & -  &  &- &  &0.0284 &  &0.0261&  &- &  &- \\
    &ForLion exact grid2.5 &7.5 & &47.5&  & -  &  & - & &51&  &47&  &- &  &- \\

     &EW Bu-grid2.5 &7.5& &45.0&  & 47.5 &  &- & & 0.0279 &  &0.0055&  &0.0197 &  &- \\
    &EW Bu-grid2.5 exact &7.5& &45.0&  & 47.5&  &-  & & 50&  &4&  & 41&  &- \\
    
    &EW ForLion &7.735& &46.387&  & -  &  & - & & 0.0279 &  &0.0252&  &- &  &- \\
    &EW ForLion exact grid2.5&7.5& &47.5&  & -  &  &- & & 50 &  &45&  &- &  &-\\

    \midrule

    4   & Bu-grid2.5&17.5& &20.0&  & 97.5 &  & 100  & & 0.0055&  &0.0237& &0.0175 &  &0.0084 \\
  &  Bu-grid2.5 exact &17.5& &20.0&  & 97.5 &  & 100 & & 10&  &42& &34 &  &12\\
    &ForLion  &18.476& &97.700&  & -&  &- &  & 0.0288 &  &0.0268 &  &- &  & - \\
   &ForLion exact grid2.5 &17.5 & &97.5&  & -  &  &  - & & 51&  &48&  &- &  &-\\

    &EW Bu-grid2.5 &0.0& &20.0&  & 100.0 &  &102.5  & &0.0253 &  &0.0270&  &0.0114 &  &0.0136 \\
   &EW Bu-grid2.5 exact&0.0& &20.0&  & 100.0 &  &102.5 & &  45 &  &48&  &5&  &40\\
   &EW ForLion &0.000& &19.984&  & 100.379  &  & - & & 0.0190 &  &0.0265&  &0.0242 &  &- \\
   &EW ForLion exact grid2.5 &0.0& &20.0&  & 100.0  &  & - & & 34&  &47&  &43 &  &-\\

    \midrule

    5 & Bu-grid2.5 &15.0& &82.5&  & - &  & - & & 0.0321&  &0.0287& &- &  &-\\
    &  Bu-grid2.5 exact &15.0& &82.5&  & - &  & - & & 57&  &51& &- &  &-\\
    &ForLion  &13.925 & &83.114&  & - &  & - &  & 0.0315 &  & 0.0294 &  & - &  & -\\
    &ForLion exact grid2.5 &15.0 & &82.5&  & -  &  & - & & 56 &  &53&  &- &  &-\\
      &EW Bu-grid2.5 &15.0& &82.5&  & 85.0 &  &-& &0.0314 &  &0.0245&  &0.0033 &  & - \\
    &EW Bu-grid2.5 exact&15.0& &82.5&  & -  &  &-  & & 56&  &50&  &-&  &-\\
    &EW ForLion  &14.527& &83.063&  & - &  & - & & 0.0311 &  &0.0280&  &- &  &-\\
    &EW ForLion exact grid2.5 &15.0& &82.5&  & - &  & - & & 55&  &50&  &- &  &-\\

    \midrule

    6  & Bu-grid2.5 &12.5& &15.0&  & 87.5 &  & - & & 0.0075&  &0.0229& &0.0251 &  &-\\
    &  Bu-grid2.5 exact &12.5& &15.0&  & 87.5 &  & - & &14&  &40& &45&  &-\\
    &ForLion   & 12.212 & &86.679&  & -  &  & -&  & 0.0277 &  &0.0259 &  &-  &  & -\\
    &ForLion exact  grid2.5 &12.5 & &87.5&  & -  &  & - & & 49 &  &46&  &- &  &-\\

    &EW Bu-grid2.5 &12.5& &15.0&  & 85.0  &  &87.5  & &0.0239 &  &0.0047&  &0.0046 &  &0.0203\\
    &EW Bu-grid2.5 exact&12.5& &15.0&  & 87.5  &  &- & & 43 &  & 8 &  & 45 &  &-\\
    &EW ForLion  &12.682& &87.709&  & -  &  & - & & 0.0273 &  &0.0250&  &- &  &-\\
    &EW ForLion exact grid2.5&12.5& &87.5&  & -  &  & - & & 49 &  &45&  &- &  &-\\

    \midrule

    7 & Bu-grid2.5 &20.0& &92.5&  & - &  & - & & 0.0329&  &0.0290& &- &  &-\\
   &  Bu-grid2.5 exact &20.0& &92.5&  & - &  & -  & &59&  &52& &- &  &-\\
    &ForLion   & 18.005  & &20.767&  & 90.517  &  & - &  &0.0048 &  & 0.0282 &  & 0.0299 &  &-\\
    &ForLion exact grid2.5  &20.0 & &90.0&  & -  &  & - & & 59 &  &53&  &- &  &-\\
      &EW Bu-grid2.5 &20.0& &90.0&  & 92.5  &  & 95.0 & &0.0322 &  &0.0001&  &0.0254 &  &0.0023\\
    &EW Bu-grid2.5 exact&20.0& & 92.5&  & -  &  & -& &57 &  &50&  &-&  &-\\
    &EW ForLion  &20.206& &92.451&  & -  &  & - & & 0.0323 &  &0.0277&  &- &  &-\\
    &EW ForLion exact grid2.5 &20.0& &92.5&  & -  &  & - & & 58 &  &49&  &- &  &-\\

    \midrule

    8   & Bu-grid2.5 &35.0& &160&  & - &  & -  & & 0.0275&  &0.0261& &- &  &-\\
     & Bu-grid2.5 exact&35.0& &160&  & - &  & - & & 49&  &47& &- &  &-\\ 
    &ForLion  & 33.541 & &35.854&  & 160  &  & - &  & 0.0224 &  &0.0048 &  & 0.0268&  & -\\
     &ForLion exact grid2.5  &35.0 & &160&  & -  &  &  -& & 49 &  &48&  &- &  &-\\
      &EW Bu-grid2.5 &35.0& &160.0&  & -  &  &-& &0.0272&  &0.0251&  &- &  &-\\
     &EW Bu-grid2.5 exact &35.0& &160.0&  & -  &  &- & &  49&  &45&  &-&  &-\\
     
     &EW ForLion  &33.863& &160&  & -  &  & - & & 0.0269 &  &0.0252&  &- &  &-\\
     &EW ForLion exact grid2.5 &35.0& &160&  & -  &  &- & & 48 &  &45&  &- &  &-\\

    \midrule

    9  & Bu-grid2.5 &22.5& &25.0&  & 110.0&  & - & & 0.0111&  &0.0199& &0.0265 &  &-\\
      & Bu-grid2.5 exact &22.5& &25.0&  & 110.0&  & - & & 19&  &36& &47 &  &-\\
    &ForLion   & 24.964 & &27.906&  &106.844&  & - &  & 0.0227&  & 0.0075 &  &0.0275 &  &- \\
     &ForLion exact grid2.5 &25.0 & &107.5&  & -  &  & - & & 54 &  &49&  &- &  &-\\
      &EW Bu-grid2.5 &25.0& &112.5&  & 115.0  &  &- & &0.0302 &  &0.0160&  &0.0090 &  &-\\
     &EW Bu-grid2.5 exact&25.0& &112.5&  & -  &  &-  & & 54&  &45&  &-&  &-\\
     &EW ForLion &1.398& &26.346&  & 115.877 &  & - & & 0.0127&  &0.0275&  &0.0235 &  &-\\
     &EW ForLion exact grid2.5 &2.5& &27.5&  & 115.0  &  & - & & 23 &  &49&  &42 &  &-\\

    \midrule

    10  & Bu-grid2.5 &2.5& &45.0&  & -&  & - & & 0.0253&  &0.0247& &- &  &-\\
     &  Bu-grid2.5 exact&2.5& &45.0&  & -&  & - & & 45&  &44& &- &  &-\\
    &ForLion   & 1.573 & &44.645&  &  - &  & - &  &0.0249 &  & 0.0247 &  & - &  & -\\
      &ForLion exact grid2.5  &2.5 & &45.0&  & -  &  &-  & & 45 &  &44&  &- &  &-\\
       &EW Bu-grid2.5 &2.5& &45.0&  & -  &  &- & &0.0245 &  &0.0247&  &- &  &-\\
     &EW Bu-grid2.5 exact &2.5& &45.0&  & -  &  &-  & & 44&  & 44&  &-&  &-\\
     &EW ForLion  &1.588& &44.899&  & -  &  & - & & 0.0247&  &0.0248&  &- &  &-\\
     &EW ForLion exact grid2.5 &2.5& &45.0&  & -  &  & - & & 44 &  &44&  &- &  &-\\

    \midrule

    11 & Bu-grid2.5 &2.5& &32.5&  & 35.0&  & -  & & 0.0254&  &0.0208& &0.0038&  &-\\
     &  Bu-grid2.5 exact&2.5& &32.5&  & 35.0&  & - & &45&  &38& &6&  &-\\
    &ForLion & 2.105 & &33.392 &  & - &  & -&  & 0.0254&  & 0.0247 &  & - &  & -\\
     &ForLion exact grid2.5  &2.5 & &32.5&  & -  &  & -& & 46 &  &44&  &- &  &-\\
      &EW Bu-grid2.5 &2.5& &32.5&  & -  &  &-  & &0.0248 &  &0.0245&  &- &  &-\\
     &EW Bu-grid2.5 exact&2.5& &32.5&  & -  &  &- & & 44&  &44&  &-&  &-\\
     &EW ForLion  &1.913& &33.162&  & -  &  & - & &0.0250&  &0.0247&  &- &  &-\\
     &EW ForLion exact grid2.5 &2.5& &32.5&  & -  &  & - & & 45 &  &44&  &- &  &-\\

    \midrule

    12   & Bu-grid2.5 &12.5& &15.0&  & 95.0&  & 97.5 & & 0.0302&  &0.0004& &0.0042 &  &0.0241\\
       & Bu-grid2.5 exact&12.5& &15.0&  & 95.0&  & 97.5 & & 54&  &1& & 5 &  &46\\
    &ForLion  & 12.553 & &96.478&  & - &  & - &  &0.0304 &  &0.0284 &  & - &  &-\\
      &ForLion exact grid2.5 &12.5 & &97.5&  & -  &  & - & & 54 &  &51&  &- &  &-\\
       &EW Bu-grid2.5 &12.5& &92.5&  & 95.0  &  &-  & &0.0303 &  &0.0120&  &0.0158 &  &-\\
      &EW Bu-grid2.5 exact&12.5& &92.5&  & 95.0 &  &-  & &54&  &25&  &25&  &-\\
      &EW ForLion  &12.194& &94.016&  & -  &  & - & &0.0301&  &0.0279&  &- &  &-\\
      &EW ForLion exact  grid2.5 &12.5& &95.0&  & -  &  & - & & 54 &  &50&  &- &  &-\\

    \midrule

    13   & Bu-grid2.5 &12.5& &15.0&  & 92.5&  & 95.0& & 0.0187&  &0.0130& &0.0224 &  &0.0061\\
       &  Bu-grid2.5 exact&12.5& &15.0&  & 92.5&  & 95.0 & & 33&  &24& &41 &  &10\\
    &ForLion   & 14.693& &92.175&  & -  &  & - &  & 0.0312 &  & 0.0287 &  & - &  & -\\
       &ForLion exact grid2.5 &15.0 & &92.5&  & -  &  & - & & 56 &  &51&  &- &  &-\\
         &EW Bu-grid2.5 &12.5& &15.0&  &  92.5 &  &95.0 & &0.0034 &  &0.0276&  &0.0098&  &0.0179\\
      &EW Bu-grid2.5 exact &12.5& &15.0&  &  95 &  &- & & 6&  & 49 &  & 49 &  &-\\
      &EW ForLion  &14.488& &93.508&  & -  &  & - & &0.0306&  &0.0279&  &- &  &-\\
       &EW ForLion exact grid2.5 &15.0& &92.5&  & -  &  & -& & 55 &  &50&  &- &  &-\\

    \midrule

    14  & Bu-grid2.5 &15.0& &17.5&  & 77.5&  &80.0 & & 0.0082&  &0.0217& &0.0025 &  &0.0235\\
     &  Bu-grid2.5 exact &15.0& &17.5&  & 80&  &-& & 14&  &39& &46 &  &-\\
    &ForLion   &16.439& &18.660&  & 76.924&  & - &  &0.0207 &  &  0.0090 &  & 0.0269 &  & -\\
    &ForLion exact grid2.5 &17.5 & &77.5&  & -  &  & - & & 53 &  &48&  &- &  &-\\
      &EW Bu-grid2.5 &17.5& &80.0&  & 82.5  &  &-  & &0.0291 &  &0.0192&  &0.0056&  &-\\
     &EW Bu-grid2.5 exact &17.5& &80.0&  & -  &  &-  & & 52 &  &44&  & - &  &-\\
    &EW ForLion &17.449& &80.152&  & -  &  & - & &0.0291&  &0.0249&  &- &  &-\\
    &EW ForLion exact grid2.5 &17.5& &80.0&  & -  &  & - & & 52 &  &44&  &- &  &-\\

    \midrule

    15 & Bu-grid2.5 &0.0& &32.5&  & -&  & - & & 0.0242&  &0.0247& &-&  &-\\
    &  Bu-grid2.5 exact &0.0& &32.5&  & -&  & -& & 43&  &44& &-&  &-\\
    &ForLion  & 0.462  & &32.783 &  & -  &  & -&  &  0.0247&  &0.0246 &  & - &  & -\\
      &ForLion exact grid2.5  &0.0 & &32.5&  & -  &  & - & & 44 &  &44&  &- &  &-\\
      &EW Bu-grid2.5 &0.0& &32.5&  & -  &  &- & &0.0226 &  &0.0248&  &-&  &-\\
    &EW Bu-grid2.5 exact&0.0& &32.5&  & -  &  &-& & 40 &  &44&  &-&  &-\\
      &EW ForLion  &0.483& &32.917&  & -  &  & - & &0.0245&  &0.0248&  &- &  &-\\
      &EW ForLion exact grid2.5 &0.0& &32.5&  & -  &  & - & & 44 &  &44&  &- &  &-\\

    \midrule

    16  & Bu-grid2.5 &5.0& &65.0&  & -&  & -  & & 0.0259&  &0.0246& &-&  &-\\
      &  Bu-grid2.5 exact&5.0& &65.0&  & -&  & -  & & 46&  &44& &-&  &-\\
    &ForLion   &4.119& &64.968&  & -  &  & - &  &0.0253 &  &0.0247 &  & - &  &-\\
     &ForLion exact  grid2.5 &5.0 & &65.0&  & -  &  & - & & 45 &  &44&  &- &  &-\\
    &EW Bu-grid2.5 &5.0& &65.0&  & -  &  &- & &0.0255 &  &0.0248&  &-&  &- \\
    &EW Bu-grid2.5 exact&5.0& &65.0&  & -  &  &-  & & 45&  &44&  &-&  &-\\
     &EW ForLion  &4.185& &64.902&  & -  &  & - & &0.0251&  &0.0249&  &- &  &-\\
    &EW ForLion exact grid2.5 &5.0& &65.0&  & -  &  & - & & 45 &  &44&  &- &  &-\\

    \midrule

    17    & Bu-grid2.5 &17.5& &142.5&  & 145.0&  & -  & & 0.0268&  &0.0015& &0.0237&  &-\\
     &  Bu-grid2.5 exact&17.5& &145&  &-&  & - & & 48&  &45& &-&  &-\\
    &ForLion  & 17.181 & &144.491&  & -  &  & -&  &0.0265&  & 0.0254&  & - &  & -\\
       &ForLion exact grid2.5 &17.5 & &145&  & -  &  &  -& &47 &  &45&  &- &  &-\\
       &EW Bu-grid2.5 &17.5& &137.5&  & 140.0 &  &142.5  & &0.0266 &  &0.0005&  &0.0193&  &0.0052\\
      &EW Bu-grid2.5 exact&17.5& &140&  & - &  &- & & 47&  &45&  &-&  &-\\
      &EW ForLion  &17.437& &140.331&  & -  &  & - & &0.0262&  &0.0251&  &- &  &-\\
      &EW ForLion exact grid2.5 &17.5& &140.0&  & -  &  &-  & & 47 &  &45&  &- &  &-\\

    \midrule

    18   & Bu-grid2.5 &17.5& &20.0&  & 92.5&  & 95.0 & & 0.0280&  &0.0042& &0.0115&  &0.0176\\
    & Bu-grid2.5 exact&17.5& &20.0&  & 92.5&  & 95.0  & & 50&  &7& &23&  &29\\
    &ForLion   & 17.157 & & 19.653&  &93.331 &  & - &  & 0.0301&  & 0.0019 &  & 0.0298 &  & - \\
      &ForLion exact grid2.5  &17.5 & &92.5&  & -  &  & - & & 57 &  &53&  &- &  &-\\
         &EW Bu-grid2.5 &17.5& &92.5&  & 95.0  &  &- & &0.0312  &  &0.0123&  &0.0159&  &-\\
      &EW Bu-grid2.5 exact&17.5& &95&  & -  &  &- & & 56 &  &50&  &-&  &-\\
      &EW ForLion  &17.933& &94.525&  & -  &  & - & &0.0314&  & 0.0282&  &- &  &-\\
      &EW ForLion exact grid2.5 &17.5& &95.0&  & -  &  & - & & 56 &  &50&  &- &  &-\\  
\end{longtable}
}

\normalsize

\section{Robustness of Sample-based EW Designs}\label{sec:robustness_sampled_based}

In this section, we investigate the robustness of the sample-based EW D-optimal designs against different sets of parameter vectors through two  case studies. In the first case study on the minimizing surface defects experiment (see Section~\ref{sec:Msd_experiment}), we compare optimal designs based on different sets of parameter vectors sampled from a prior distribution. In the second case study on an emergence of house flies experiment \citep{itepan1995}, we consider optimal designs based on different set of parameters estimated from bootstrapped datasets. Both case studies show that the sample-based EW D-optimal designs are fairly robust in terms of relative efficiencies against different sets of parameter vectors.

\subsection{Case study on the minimizing surface defects experiment}
\label{sec:case_study_minimizing_surface}

We simulate five random sets of parameter vectors from the prior distributions listed in Table~S1 of \cite{huang2024forlion}, denoted by $\boldsymbol{\Theta}^{(r)} = \{\boldsymbol{\theta}_1^{(r)}, \ldots, $ $ \boldsymbol{\theta}_B^{(r)}\}$, with $B=1,000$, and $r\in \{1,2,3,4,5\}$. For each $r = 1,\ldots, 5$, we obtain the corresponding sample-based EW D-optimal design $\boldsymbol{\xi}^{(r)}$ for the minimizing surface defects experiment, which maximizes $f_{\rm SEW}(\boldsymbol{\xi}) = |\hat{E}\{{\mathbf F}(\boldsymbol{\xi}, \boldsymbol{\Theta}^{(r)})\}| = |B^{-1} \sum_{j=1}^B {\mathbf F}(\boldsymbol{\xi}, \boldsymbol{\theta}_j^{(r)})|$. The numbers of design points are 15, 13, 21, 19, and 17, respectively.  

For $l, r\in \{1,2,3,4,5\}$, we calculate the relative efficiencies of $\boldsymbol{\xi}^{(l)}$ with respect to $\boldsymbol{\xi}^{(r)}$ under $\boldsymbol{\Theta}^{(r)}$, that is,
\[
\left(\frac{|\hat{E}\{{\mathbf F}(\boldsymbol{\xi}^{(l)}, \boldsymbol{\Theta}^{(r)})\}|}{|\hat{E}\{{\mathbf F}(\boldsymbol{\xi}^{(r)}, \boldsymbol{\Theta}^{(r)})\}|}\right)^{1/p}
\]
with $p=10$. The relative efficiencies are shown in the following matrix
\[
\bordermatrix{
  & $l=1$ & $l=2$ & $l=3$ & $l=4$ & $l=5$ \cr
$r=1$ & 1.00000 & 0.99179& 0.99586& 0.99284& 0.98470 \cr
$r=2$ & 0.98258& 1.00000& 0.98508& 0.98134& 0.98117 \cr
$r=3$ & 0.99078 &0.98748 &1.00000 &0.98165 &0.98632\cr
$r=4$ & 0.97629 & 0.98376 &0.97910& 1.00000& 0.98501\cr
$r=5$ & 0.97947& 0.98850& 0.98804& 0.98842& 1.00000\cr
}
\]
with the minimum relative efficiency $0.97629$.

By applying the rounding algorithm (Algorithm~\ref{alg:appro to exact}) with merging threshold $L=1$ and the number of experimental units $n=1,000$, we obtain the corresponding exact designs, whose numbers of design points are still 15, 13, 21, 19, and 17, respectively. Their relative efficiencies are
\[
\bordermatrix{
  & $l=1$ & $l=2$ & $l=3$ & $l=4$ & $l=5$ \cr
$r=1$ & 1.00000 &0.99144& 0.99593& 0.99299& 0.98448 \cr
$r=2$ &0.98287 &1.00000 &0.98535& 0.98174& 0.98128\cr
$r=3$ & 0.99080& 0.98716& 1.00000& 0.98184& 0.98614\cr
$r=4$ &0.97634& 0.98359& 0.97916& 1.00000& 0.98488\cr
$r=5$ & 0.97959& 0.98836& 0.98797& 0.98846 &1.00000\cr
}
\]
with the minimum relative efficiency $0.97634$.

In terms of relative efficiencies, the sample-based EW D-optimal designs for this experiment are fairly robust against the different sets of parameter vectors.

\subsection{Case study on an emergence of house flies experiment}
\label{sec:case_study_emergence_house_flies}

In this section, we consider an emergence of house flies experiment, originally described by \cite{itepan1995}. In the original design, $n_i=500$ pupae out of $n=3,500$ were assigned to each of $m=7$ evenly distributed gamma radiation levels $x_i=80, 100, \ldots, 200$ (in units Gy). Having exposed to the gamma radiation of a device for a certain period of time, each pupa had one of three possible outcomes, namely unopened, opened but died, and completed emergence. The outcomes of the original experiment had been recorded (see, e.g., Table~1 in \cite{atkinson1999}). A continuation-ratio logistic model with non-proportional odds (npo) has been adopted in the literature \citep{atkinson1999, bu2020, ai2023locally, huang2024forlion} with parameters $\boldsymbol{\theta} \in \mathbb{R}^5$.
By bootstrapping the original observations for $B=1,000$ times, an EW D-optimal design $\boldsymbol{\xi}_{\rm Bu} = \{(80, 0.3120), (120, 0.2911), (140, 0.1087), (160, 0.2882)\}$ restricted to the original seven grid points was obtained by \cite{bu2020}, which is slightly more robust than the corresponding Bayesian D-optimal design.

In this paper, we use this experiment with a continuous design range ${\cal X} = [0, 200]$ to investigate the robustness of sample-based EW D-optimal designs against 
the randomness of the bootstrapping samples. For illustration purposes, we repeat five times the following simulation study independently. In each simulation, we first generate $B=100$ or $B=1,000$ bootstrapped parameter vectors by bootstrapping the original data and fitting the corresponding continuation-ratio npo model.
Then we apply Algorithm~\ref{alg:EW ForLion} with this set of bootstrapped parameter vectors, merging threshold $\delta = 0.8$, and converging threshold $\epsilon = 10^{-6}$ to obtain the corresponding sample-based EW D-optimal design $\boldsymbol{\xi}_i^B$, $i=1, \ldots, 5$. 
For $B=100$, the numbers of design points of $\boldsymbol{\xi}_i^B$ are $4, 4, 4, 5, 6$, respectively; while for $B=1,000$, the corresponding numbers $4, 5, 5, 5, 4$ seem more stable.

Similarly to Section~\ref{sec:case_study_minimizing_surface}, we evaluate the robustness of $\boldsymbol{\xi}_i^B, i=1, \ldots, 5$ by calculating their pairwise relative efficiencies. For $B=100$, the relative efficiencies are shown in the following matrix
\[
\bordermatrix{
  & $l=1$ & $l=2$ & $l=3$ & $l=4$ & $l=5$ \cr
$r=1$ & 1.00000   &  0.99989 & 1.00002   & 0.99972    & 0.99985   \cr
$r=2$ & 1.00000   & 1.00000  & 1.00002   &  0.99989   & 0.99986   \cr
$r=3$ & 0.99998   & 0.99990  &  1.00000  & 0.99976    &  0.99984  \cr
$r=4$ & 1.00009   & 1.00007  & 1.00012   &  1.00000   & 1.00000   \cr
$r=5$ &  1.00008  & 1.00003  & 1.00011   &  0.99995   &  1.00000  \cr
}\ ,
\]
ranging in $[0.99972, 1.00012]$ and indicating their performances are nearly identical.  As for $B=1,000$, the corresponding matrix of relative efficiencies is
\[
\bordermatrix{
  & $l=1$ & $l=2$ & $l=3$ & $l=4$ & $l=5$ \cr
$r=1$ & 1.00000   & 0.99971  & 0.99983   &  0.99986   & 0.99997   \cr
$r=2$ & 1.00026   & 1.00000  &  1.00010  & 1.00011    &  1.00021  \cr
$r=3$ & 1.00016   & 0.99988  &  1.00000  & 1.00002    &  1.00012  \cr
$r=4$ & 1.00014   & 0.99987  &  0.99997  & 1.00000    & 1.00011  \cr
$r=5$ & 1.00003   &  0.99975 &  0.99987  & 0.99989    &  1.00000  \cr
}\ ,
\]
ranging between $0.99971$ and $1.00026$, which are fairly close to $1$ as well.

We conclude that in this case the sample-based EW D-optimal designs are fairly stable against the choice of set of bootstrapped parameter vectors in terms of relative efficiencies.

\section{More Examples}\label{sec:Additional Examples: Generalized Linear Models}

In this section, we provide two more examples as applications of the proposed EW D-optimal designs, including an electrostatic discharge experiment example and a three-continuous-factor artificial example, both under generalized linear models.

\subsection{Electrostatic discharge experiment}\label{sec:Electrostatic discharge (ESD) experiment}

An electrostatic discharge (ESD) experiment was considered by \cite{lukemire2018} and \cite{huang2024forlion}, which involves a binary response, four discrete factors, namely  LotA, LotB, ESD, and Pulse, all taking values in $\{-1,1\}$, and one continuous factor Voltage in $[25,45]$. The list of factors, corresponding parameters, the nominal values adopted by \cite{lukemire2018}, and the prior distribution adopted by \cite{huang2024forlion} are summarized in Table~\ref{tab:ESD experiment variable}.  \cite{lukemire2018} employed a $d$-QPSO algorithm and derived a locally D-optimal approximate design with 13 support points for a logistic model $Y \sim \operatorname{Bernoulli}(\mu)$ with $ \text{logit}(\mu)=\beta_0+\beta_1\texttt{LotA}+\beta_2\texttt{LotB}+\beta_3\texttt{ESD}+\beta_4\texttt{Pulse}+\beta_5\texttt{Voltage}+\beta_{34}(\texttt{ESD}\times\texttt{Pulse})$. \cite{huang2024forlion} used their ForLion algorithm and found a slightly better locally D-optimal design in terms of relative efficiency, which consists of 14 design points. Both designs were constructed under the assumption that the nominal parameter values listed in Table~\ref{tab:ESD experiment variable} are the true values.

In practice, however, an experimenter may not possess precise parameter values before the experiment. As an illustration, we adopt the prior distribution used in \cite{huang2024forlion} for model parameters, i.e., independent uniform distributions (see Table~\ref{tab:ESD experiment variable}). 
By applying our EW ForLion algorithm on the same model with the specified prior distribution, we obtain the integral-based EW D-optimal approximate design, and present it on the left side of Table~\ref{tab:Electrostatic discharge (ESD) EW ForLion design}, which consists of 20 design points. Subsequently, we utilize our rounding algorithm with merging threshold $\delta_2=0.5$, the rounding level $L=0.1$, and the number of experimental units $n=100$ (as an illustration) to obtain an exact design, and present it on the right side of Table~\ref{tab:Electrostatic discharge (ESD) EW ForLion design}. Since two pairs of close points are merged,
and one point with zero allocation is removed, the exact design consists of only 17 design points.

\begin{table}[ht]
    \centering
    \captionsetup{skip=2pt} %
    \caption{Predictors and parameters for electrostatic discharge experiment}
    \footnotesize %
    {
    \renewcommand{\arraystretch}{0.5}
        \resizebox{\textwidth}{!}{
    \begin{tabular}{@{}lllll@{}} %
    \toprule
    & \textbf{Factor/Predictor} & \textbf{Factor} & \textbf{Parameter} & \textbf{Parameter} \\
    \textbf{Type} & \multicolumn{1}{l}{\textbf{(Parameter)}}  & \textbf{Levels/Range}  & \textbf{Prior Distribution} & \textbf{Nominal Value} \\
    \midrule
    & Intercept($\beta_0$)  &   & $U(-8,-7)$ & -7.5 \\
    Discrete & Lot A ($\beta_1$) & \{-1,1\} & $U(1,2)$ & 1.5 \\
    Discrete & Lot B ($\beta_2$) & \{-1,1\} & $U(-0.3,-0.1)$ & -0.2 \\
    Discrete & ESD ($\beta_3$) & \{-1,1\} & $U(-0.3,0)$ & -0.15 \\
    Discrete & Pulse ($\beta_4$) & \{-1,1\} & $U(0.1,0.4)$ & 0.25 \\
    Continuous & Voltage ($\beta_5$) & [25,45] & $U(0.25,0.45)$ & 0.35 \\
    Interaction & ESD $\times$ Pulse ($\beta_{34}$) &  & $U(0.35,0.45)$ & 0.4 \\
    \bottomrule
    \end{tabular}
    }
    }
    \label{tab:ESD experiment variable}
\end{table}

\begin{table}[ht]
    \centering
\captionsetup{skip=2pt} 
    \caption{EW D-optimal approximate design (left) and exact design (right, $n=100$) for electrostatic discharge experiment}
    {
    \renewcommand{\arraystretch}{0.5}
        \resizebox{\textwidth}{!}{
    \begin{tabular}{c|rrrrrr|c|rrrrrr}
\hline Support &   &   &  &   &   &   & Support  &  \multicolumn{6}{c}{$n=100, L=0.1, \delta_2=0.5$}  \\ 
point & Lot A & Lot B & ESD & Pulse & Voltage & $p_i$ (\%) & point & Lot A & Lot B & ESD & Pulse & Voltage  & $n_i$\\
\hline
 1  & -1  & -1  & -1  & 1   & 25.0000  & 8.75  &
 1  & -1  & -1  & -1  & 1   & 25.0  & 9 \\
 
2  & -1 & 1 & 1 & 1 & 25.0000 & 8.45  & 
2  & -1 & 1 & 1 & 1 & 25.0 & 8 \\

3  & -1 & -1 & -1 & -1 & 25.0000 & 8.48  &
3  & -1 & -1 & -1 & -1 & 25.0 & 8 \\

4  & 1 & 1 & -1 & 1 & 25.0000 & 6.21  &
4  & 1 & 1 & -1 & 1 & 25.0  & 6 \\

5  & -1 & 1 & 1 & -1 & 38.9047 & 2.14  &
5  & -1 & 1 & 1 & -1 & 38.9  & 8 \\

6  & 1 & 1 & -1 & -1 & 25.0000 & 3.56  & 
6  & 1 & 1 & -1 & -1 & 25.0 & 4\\

7  & -1 & -1 & 1 & 1 & 25.0000 & 8.56  & 
7  & -1 & -1 & 1 & 1 & 25.0 & 9 \\

8  & -1 & 1 & -1 & 1 & 25.0000 & 5.15 &
8  & -1 & 1 & -1 & 1 & 25.0 & 5 \\

9  & -1 & 1 & -1 & -1 & 25.0000 & 6.90  &
9  & -1 & 1 & -1 & -1 & 25.0  & 7 \\

10  & -1 & 1 & 1 & 1 & 33.1161 & 0.22 &
10  & -  & - & - & - & -       & - \\

11  & -1 & 1 & -1 & 1 & 35.4140 & 0.28  & 
11  & -1 & 1 & -1 & 1 & 35.4 & 4 \\

12  & 1 & 1 & 1 & -1 & 25.0000 & 4.43  &
12  & 1 & 1 & 1 & -1 & 25.0 & 4 \\

13  & 1 & 1 & 1 & 1 & 25.0000 & 0.90  &
13  & 1 & 1 & 1 & 1 & 25.0 & 1 \\

14  & -1 & 1 & -1 & 1 & 35.3993 & 3.52  &
14  & - & - & - & - & -  & - \\

15  & -1 & 1 & 1 & -1 & 25.0000 & 9.01  &
15  & -1 & 1 & 1 & -1 & 25.0  & 9 \\

16  & 1 & -1 & 1 & -1 & 25.0000 & 7.43  & 
16  & 1 & -1 & 1 & -1 & 25.0 & 7 \\

17  & -1 & 1 & -1 & -1 & 34.0238 & 1.57  & 
17  & -1 & 1 & -1 & -1 & 34.0  & 2 \\

18  & -1 & -1 & 1 & -1 & 37.1975 & 4.55  &
18  & -1 & -1 & 1 & -1 & 37.2  & 5 \\

19  & -1 & -1 & 1 & -1 & 25.0000 & 4.10  &
19  & -1 & -1 & 1 & -1 & 25.0 & 4 \\

20  & -1 & 1 & 1 & -1 & 38.9522 & 5.80  &
20  & - & - & - & - & - & - \\
  
\hline
\end{tabular}
    }
    }
    \label{tab:Electrostatic discharge (ESD) EW ForLion design}
\end{table}

\normalsize
To make comparisons, we randomly generate $N=10,000$ parameter vectors $\{\boldsymbol{\theta}_1, \ldots, $ $\boldsymbol{\theta}_N\}$ from the prior distributions listed in Table~\ref{tab:ESD experiment variable}. For each $j \in \{1, \ldots, N\}$ and each $\boldsymbol{\xi}$ under comparison, we compute  $|{\mathbf F}(\boldsymbol{\xi}, \boldsymbol{\theta}_j)|^{1/p} = |{\mathbf X}_{\boldsymbol{\xi}}^T {\mathbf W}_{\boldsymbol{\xi}}(\boldsymbol{\theta}_j) {\mathbf X}_{\boldsymbol{\xi}}|^{1/p}$ with $p=7$ in this case. In Figure~\ref{fig:side_by_side}, we present the resulting values as frequency polygons for our EW ForLion design, the locally D-optimal design (``ForLion'') obtained by \cite{huang2024forlion}, and the $d$-QPSO design (``PSO'') obtained by \cite{lukemire2018}. Compared with the ForLion or PSO design, the frequency polygon for EW ForLion design shows a slightly higher density with low objective functions values, but also a noticeably increased frequency in the higher-value range. Actually, in Table~\ref{tab:ESD comparison}, it clearly shows that the overall mean and median objective function values for EW ForLion designs (approximate and exact designs listed in Table~\ref{tab:Electrostatic discharge (ESD) EW ForLion design}) are higher than ForLion and PSO designs.

\begin{figure}[ht]
  \centering
  \begin{subfigure}{0.45\textwidth}
    \centering
    \includegraphics[width=\textwidth]{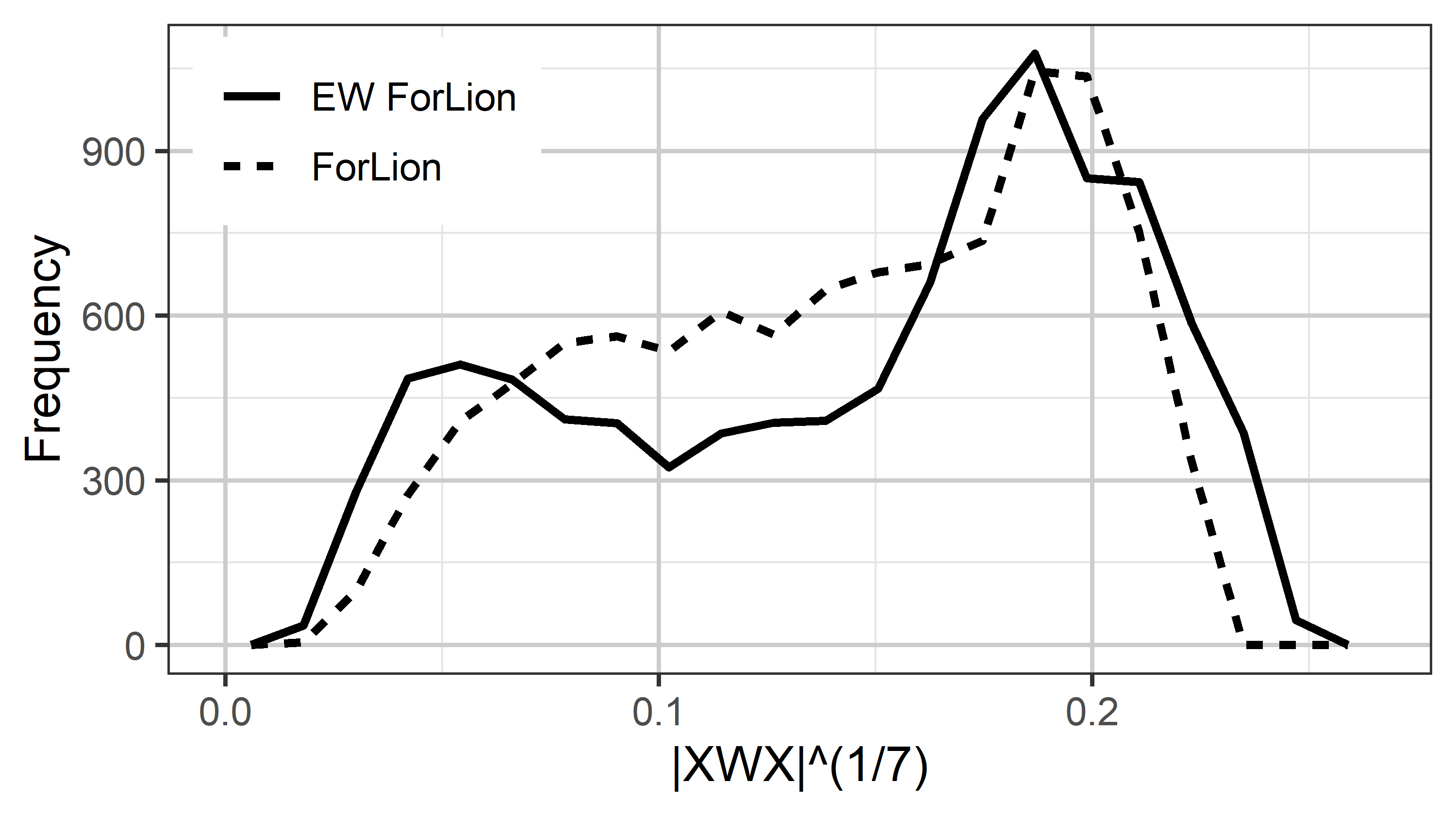} %
    \label{fig:sub1}
  \end{subfigure}
  \hfill
  \begin{subfigure}{0.45\textwidth}
    \centering
    \includegraphics[width=\textwidth]{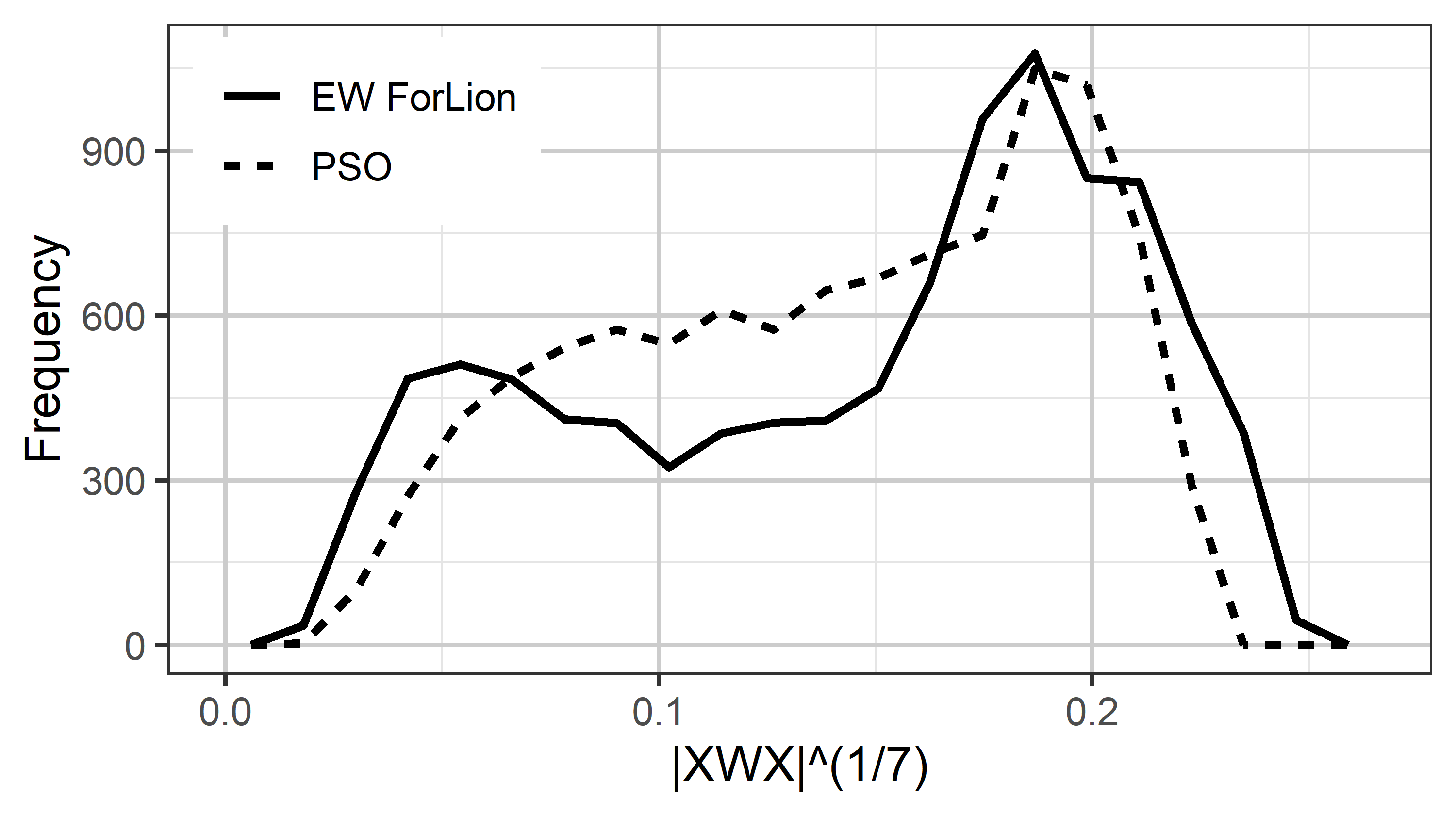} %
    \label{fig:sub2}
  \end{subfigure}\\ [-3ex]
  \caption{Frequency polygons of objective function values based on 10,000 simulated parameter vectors for electrostatic discharge experiment}
  \label{fig:side_by_side}
\end{figure}

\normalsize
\begin{table}[ht]
    \centering
\captionsetup{skip=2pt} 
    \caption{Mean and median objective function values based on 10,000 simulated parameter vectors for electrostatic discharge experiment}
    {
   \renewcommand{\arraystretch}{0.5}   
    \begin{tabular}{ccc}
\hline Design & Mean $|{\mathbf X}^T {\mathbf W} {\mathbf X}|^{1/7}$& Median $|{\mathbf X}^T {\mathbf W} {\mathbf X}|^{1/7}$ \\
\hline EW ForLion & 0.1467647 & 0.1642085\\
 EW ForLion exact & 0.1467639 &  0.1640990\\
 ForLion &  0.1425194 & 0.1498971 \\
PSO & 0.1418699 & 0.1487820\\
\hline
\end{tabular}
    }
    \label{tab:ESD comparison}
\end{table}

\subsection{A three-continuous-factor example under a GLM}\label{sec:Three-continuous-factor}

This example refers to Example~4.7 of \cite{stufken2012} on a logistic regression model with three continuous factors, represented as ${\rm logit} (\mu_i)=\beta_0+\beta_1 x_{i 1}+\beta_2 x_{i 2}+\beta_3 x_{i 3}$~, where the factors were originally defined as $x_{i 1} \in [-2,2]$, $x_{i 2} \in [-1,1]$, and $x_{i 3} \in (-\infty, \infty)$, and a locally D-optimal design was obtained with respect to nominal parameter values $(\beta_0, \beta_1, \beta_2, \beta_3) = (1, -0.5, 0.5, 1)$. 

In this section, for illustration purposes,  we reconsider this example by letting $x_{i 3} \in [-3, 3]$ (see also Example~2 in \cite{huang2024forlion}) and assuming independent priors
$\beta_0 \sim N(1, \sigma^2)$, $\beta_1 \sim N(-0.5, \sigma^2)$, 
$\beta_2 \sim N(0.5, \sigma^2)$, 
and $\beta_3 \sim N(1, \sigma^2)$,  with $\sigma=1$. Note that the normal priors actually imply unbounded ranges for $\beta_i$'s, which seems violating the boundedness condition in Theorem~\ref{thm:GLM design points}. Nevertheless, from a practical point of view, a normal $N(\mu, \sigma^2)$ prior is essentially the same as its truncated version on $[\mu - s\sigma,\ \mu + s\sigma]$ with, for example, $s=5$. By using our EW ForLion algorithm anyway, we obtain an EW D-optimal approximate design, which costs 783 seconds. We further apply our rounding algorithm to the approximate design with $L=0.0001$ and $n=100$, and obtain an exact design in just 0.45 second.  Remarkably, both designs consist of 7 design points, as detailed in Table~\ref{tab:Three-continuous-factor EW ForLion design}, which are different from the 8-point design recommended by \cite{stufken2012} with unconstrained $x_{i3}$~.

\begin{table}[ht]
    \centering
\captionsetup{skip=2pt} 
    \caption{EW D-optimal approximate design (left) and exact design (right, $n=100$) for the three-continuous-factor example}
    {
    \renewcommand{\arraystretch}{0.5}
        \resizebox{0.8\textwidth}{!}{
    \begin{tabular}{c|cccr|c|cccr}
\hline Support &   &   &   &   &Support  &   &   &   &     \\
point & $x_{1}$ & $x_{2}$ & $x_{3}$ & $p_i$ (\%) & point & $x_{1}$ & $x_{2}$ & $x_{3}$ & $n_i$\\
\hline
 1  & -2  & -1  & -3  & 7.231  &
 1  & -2  & -1  & -3  & 7 \\
 
2  & 2  & -1  & -3  & 20.785  &
2  & 2  & -1  & -3  & 21\\

3  & -2  & 1  & -1.8  & 19.491  & 
3  & -2  & 1  & -1.8  & 19 \\

4  & 2  & 1  & 3  & 2.718  & 
4  & 2  & 1  & 3  & 3 \\

5  & 2  & 1  & -0.3321  & 18.870  & 
5  & 2  & 1  & -0.3321  & 19 \\

6  & -2  & -1  & -0.0867  & 10.951  &
6  & -2  & -1  & -0.0867  & 11 \\

7  & 0.9467  & -0.9969  & 2.9932  & 19.954  & 
7 &0.9467  & -0.9969  & 2.9932  & 20\\
\hline
\end{tabular}
    }}
    \label{tab:Three-continuous-factor EW ForLion design}
\end{table}

We first compare our designs (``EW ForLion'' and ``EW ForLion exact-0.0001'') with EW D-optimal designs obtained by applying \cite{bu2020}'s EW lift-one and exchange algorithms to grid points of continuous factors by discretizing them into 2, 4, 6, 8, or 10 evenly distributed grid points within the corresponding design regions. When the continuous variables are discretized into 10 evenly distributed grid points, \cite{bu2020}'s EW lift-one and exchange algorithms no longer work due to computational intensity. The summary information of these designs are presented in Table~\ref{tab:Robust Designs for the Three-continuous-factor}, from which we can see that our EW ForLion algorithm produces a design that is more robust against parameter misspecifications and also requires fewer design points.

\begin{table}[ht]
    \centering
    \caption{Robust designs for three-continuous-factor example}
    {
    \renewcommand{\arraystretch}{0.5}
        \resizebox{0.8\textwidth}{!}{
    \begin{threeparttable}
          \begin{tabular}{cccccrccccr}
    \toprule
    Design &  &$m$& & & Time (s) & & &$|E\{{\mathbf F}(\boldsymbol{\xi}, \boldsymbol\Theta)\}|$ &   &Relative Efficiency \\
    \midrule
   EW Grid-2 & & 7 &  &  & 0.06s &  &  & 0.0009563&   &90.711\%\\
    EW Grid-2 exact &  & 7 &  &  & 0.11s &  &  &0.0009559&   &90.700\%\\

   EW Grid-4&   &8 &  &  &  4.33s &  &  &0.0012320&   &96.641\%\\
     EW Grid-4 exact&  &10 &  &  &  5.03s &  &  &0.0012314&  &96.629\%\\

 EW Grid-6&   &8 &  &  &  13.87s &  &  &0.0013278&   &98.469\%\\
     EW Grid-6 exact&  &8 &  &  &  19.20s &  &  &0.0013271&  &98.455\%\\
     
     EW Grid-8&   &8 &  &  &  37.58s &  &  & 0.0013562&   &98.989\%\\
     EW Grid-8 exact&  &9 &  &  &  46.42s &  &  &0.0013553&  &98.973\%\\

      EW Grid-10&   &- &  &  &  - &  &  & -&   &-\\
     EW Grid-10 exact&  &- &  &  &  - &  &  &-&  &-\\
   EW ForLion&  & 7 &  &  &  782.05s &  &  & 0.0014124&   &100.000\%\\
      EW ForLion exact-0.0001&   & 7 &  &  &   0.45s &  &  & 0.0014119 &  &99.991\%\\
    \bottomrule
    \end{tabular}
     \begin{tablenotes}
         \footnotesize
          \setlength{\baselineskip}{8pt} %
         \item Note: Relative Efficiency = $(|E\{{\mathbf F}(\boldsymbol{\xi}, \boldsymbol\Theta)\}|/|E\{{\mathbf F}(\boldsymbol{\xi}_{\rm EW\ ForLion}, \boldsymbol\Theta)\}|)^{1/p}$ with $p=4$; ``-'' indicates ``unavailable'' due to computational intensity.
     \end{tablenotes}
    \end{threeparttable}
      }}
    \label{tab:Robust Designs for the Three-continuous-factor}
\end{table}

For illustration purposes, we also compare our EW ForLion design with Bayesian and minimax D-optimal designs obtained via the R package {\tt ICAOD} (version 1.0.1, \cite{masoudi2022icaod}), which was developed for nonlinear or generalized linear models (but not for multinomial logistic models). As this package requires a bounded parameter space, for each parameter $\beta_i$ we adopt the corresponding finite interval $[\mu_i - 3\sigma, \mu_i + 3\sigma]$, which covers about $99.73\%$ of the range in a probability sense. As the package also requires a prespecified number of design points, we consider two scenarios: {\it (i)} the number of design points $k=7$, namely ``minimax k7'' and ``Bayesian k7'', which matches our EW ForLion design; and {\it (ii)} $k=216$, namely ``minimax k216'' and ``Bayesian k216'', which matches the initial number of design points used by the EW ForLion algorithm. The performance of each design is evaluated under three criteria (in our notations): EW D-efficiency $|\int_{\boldsymbol{\theta}} {\mathbf F}(\boldsymbol{\xi}, \boldsymbol{\theta}) Q(d\boldsymbol{\theta})|$, minimax D-efficiency $\min_{\boldsymbol{\theta}} \log|{\mathbf F}(\boldsymbol{\xi}, \boldsymbol{\theta})|$, and Bayesian D-efficiency $\int_{\boldsymbol{\theta}} \log|{\mathbf F}(\boldsymbol{\xi}, \boldsymbol{\theta})| Q(d\boldsymbol{\theta})$. The results are summarized in Table~\ref{tab:Comparison_efficiencies_three_continuous_factors}. Under EW D-efficiency, the EW ForLion design is no doubt the best, the Bayesian k7 design achieves 88.85\%, and minimax designs are not satisfactory. Under minimax D-efficiency with respect to the minimax k216 design, not only the EW ForLion, Bayesian k7 designs perform poorly, but the minimax k7 design is not satisfactory. In terms of Bayesian D-efficiency, both the EW ForLion and minimax k216 designs attain an efficiency higher than 80\%, while the minimax k7 design is not satisfactory.

Overall in this case the EW ForLion design is a good surrogate for Bayesian D-optimal designs, but not for minimax D-optimal designs. Compared with Bayesian D-optimal designs, it is computationally much more convenient. 
Actually, the Bayesian k7 design costs about 6 hours, while the Bayesian k216 design is not available in Table~\ref{tab:Comparison_efficiencies_three_continuous_factors} because it would cost more than 20 days.
As for minimax D-optimal designs, the minimax k7 design is not satisfactory, while the minimax k216 design contains so many more distinct design points. Actually, the three smallest weights out of 216 in the minimax k216 design are 0.0011\%, 0.0051\%, and 0.0084\%, respectively, which are not zeros. 
On the contrary, the EW ForLion design only needs 7 distinct design points, which can save lots of time and cost in practice \citep{huang2024forlion}.

\begin{table}[ht]
    \centering
    \captionsetup{skip=2pt}
    \caption{Comparison between different robust designs for the three-continuous-factor example}
    \label{tab:Comparison_efficiencies_three_continuous_factors} 

    \begin{threeparttable}
        \begin{tabular*}{\textwidth}{@{\extracolsep{\fill}} l r r r r r} 
        \toprule
        Designs & $m$ & Time (s)&eff.EW  & eff.minimax & eff.Bayesian \\
        \midrule
         minimax k7    & 7 & 441.43& 48.43$\%$  &64.12$\%$ & 63.48$\%$ \\
        Bayesian k7   & 7  & 21394.99 & 88.85$\%$   & 7.93$\%$ & 100$\%$ \\
        minimax k216    & 216  & 7243.43& 65.57$\%$   & 100$\%$ & 83.78$\%$ \\
        Bayesian k216   & -  & -    & - & - & -\\
         EW ForLion & 7 & 782.05 & 100$\%$      & 0.57$\%$& 82.31$\%$ \\
        \bottomrule
        \end{tabular*}
        \begin{tablenotes}[flushleft]
          \footnotesize
          \setlength{\baselineskip}{8pt}
          \item Note: $m$ stands for the number of distinct design points; eff.EW = $(|\int_{\boldsymbol{\theta}} {\mathbf F}(\boldsymbol{\xi}, \boldsymbol{\theta}) Q(d\boldsymbol{\theta})|/|\int_{\boldsymbol{\theta}} {\mathbf F}(\boldsymbol{\xi}_{\rm EW\ ForLion}, \boldsymbol{\theta}) Q(d\boldsymbol{\theta})|)^{1/p}$ with $p=4$; eff.minimax = $\exp((\min_{\boldsymbol{\theta}} \log|{\mathbf F}(\boldsymbol{\xi}, \boldsymbol{\theta})|-\min_{\boldsymbol{\theta}} \log|{\mathbf F}(\boldsymbol{\xi}_{\rm minimax\ k216}, \boldsymbol{\theta})|)/p)$; eff.Bayesian = $\exp((\int_{\boldsymbol{\theta}} \log|{\mathbf F}(\boldsymbol{\xi}, \boldsymbol{\theta})| Q(d\boldsymbol{\theta})-\int_{\boldsymbol{\theta}} \log|{\mathbf F}(\boldsymbol{\xi}_{\rm Bayesian\ k7}, \boldsymbol{\theta})| Q(d\boldsymbol{\theta}))/p)$; and ``-'' indicates ``unavailable'' due to computational intensity. 
        \end{tablenotes}

    \end{threeparttable}
\end{table}

\end{document}